\documentclass[onecolumn,draftcls]{IEEEtran}
\usepackage{epsfig}
\usepackage{times}
\usepackage{float}
\usepackage{afterpage}
\usepackage{amsmath}
\usepackage{amstext}
\usepackage{amssymb,bm, bbm}
\usepackage{latexsym}
\usepackage{color}
\usepackage{graphicx}
\usepackage{amsmath}
\usepackage{amsthm}
\usepackage{graphicx}
\usepackage[utf8]{inputenc}
\usepackage{epsfig}
\usepackage[center]{caption}
\usepackage{booktabs}
\usepackage{multicol}
\usepackage{lipsum}
\usepackage{dblfloatfix}
\usepackage{mathrsfs}
\usepackage{cite}
\usepackage{tikz}
\usepackage{pgfplots}
\pgfplotsset{compat=newest}
 \usepackage{enumitem}
\usepackage{makecell}
\usepackage{float}
\restylefloat{table}
\usepackage{subcaption}
\usepackage{enumitem}
\usepackage{url}

\allowdisplaybreaks

\newcommand{\rmd}{\mathrm{d}}
\newcommand{\bbE}{\mathbb{E}}\newcommand{\rme}{\mathrm{e}}

\newcommand{\bbN}{\mathbb{N}}

\newcommand{\bbP}{\mathbb{P}}

\newcommand{\bbR}{\mathbb{R}}

\newcommand{\bbZ}{\mathbb{Z}}

\newcommand{\sfK}{\mathsf{K}}

\newcommand{\sfN}{\mathsf{N}}

\newcommand{\sfQ}{\mathsf{Q}}

\newcommand{\sfT}{\mathsf{T}}
\newcommand{\sfU}{\mathsf{U}}

\newcommand{\sfW}{\mathsf{W}}
\newcommand{\sfX}{\mathsf{X}}
\newcommand{\sfY}{\mathsf{Y}}
\newcommand{\sfZ}{\mathsf{Z}}

\newcommand{\cA}{\mathcal{A}}
\newcommand{\cB}{\mathcal{B}}

\newcommand{\cD}{\mathcal{D}}

\newcommand{\cI}{\mathcal{I}}

\newcommand{\cN}{\mathcal{N}}

\newcommand{\cP}{\mathcal{P}}

\newcommand{\cS}{\mathcal{S}}

\newcommand{\cX}{\mathcal{X}}
\newcommand{\cY}{\mathcal{Y}}

\theoremstyle{mystyle}
\newtheorem{theorem}{Theorem}
\theoremstyle{mystyle}
\newtheorem{lemma}{Lemma}
\theoremstyle{mystyle}
\newtheorem{prop}{Proposition}
\theoremstyle{mystyle}
\newtheorem{corollary}{Corollary}
\theoremstyle{mystyle}
\newtheorem{definition}{Definition}
\theoremstyle{remark}
\newtheorem{rem}{Remark}
\theoremstyle{mystyle}
\theoremstyle{mystyle}
\newtheorem{exa}{Example}
\theoremstyle{mystyle}
\theoremstyle{discussion}
\theoremstyle{mystyle}
\theoremstyle{mystyle}
\theoremstyle{mystyle}
\theoremstyle{mystyle}
 \usepackage{enumitem}
\usepackage{pifont}

\usepackage{adjustbox} 
\usepackage{booktabs}  
\usepackage{float}  

\usepackage[most]{tcolorbox}
\tcbset{colback=blue!5, colframe=blue!40!black, boxrule=0.5pt, arc=2pt, left=8pt, right=8pt, top=6pt, bottom=6pt}
\usepackage{xcolor}

\usepackage{pgfplots}
\pgfplotsset{compat=1.18}
\usepackage{xcolor}

\definecolor{matblue}{rgb}{0, 0.447, 0.741}
\definecolor{matgreen}{rgb}{0, 0.6, 0}
\definecolor{matred}{rgb}{1, 0, 0}

\definecolor{matblue}{rgb}{0, 0.4470, 0.7410}
\definecolor{matred}{rgb}{1, 0, 0}

\usepackage{graphicx}
\usepackage{subcaption}

\newcommand{\snr}{{\rm snr}}

\usepackage{hyperref}

\title{$\alpha$-Mutual Information \\ for the Gaussian Noise Channel}

\author{%
   \IEEEauthorblockN{Mohammad Milanian\IEEEauthorrefmark{1},
Alex Dytso\IEEEauthorrefmark{2},
                     Martina Cardone\IEEEauthorrefmark{1}                     
                     }
                     
   \IEEEauthorblockA{\IEEEauthorrefmark{1}%
                     University of Minnesota, Minneapolis, MN 55404, USA,
                     \{milanian, mcardone\}@umn.edu}      
                     
  \IEEEauthorblockA{\IEEEauthorrefmark{3}%
                    Qualcomm Flarion Technology, Inc., Bridgewater, NJ 08807, USA, 
                    odytso2@gmail.com
 }
 \thanks{The research at the University of Minnesota was supported in part by the NSF under grant no. \mbox{CCF-2045237}.}
 }

\begin{document}
\maketitle

\begin{abstract}
In this paper, we study Sibson’s $\alpha$-mutual information, which we refer to as $\alpha$-mutual information, in the context of the additive Gaussian noise channel. While the classical case $\alpha = 1$ is well understood and admits deep connections to estimation-theoretic quantities, such as the minimum mean-square error (MMSE) and Fisher information, many of the corresponding structural properties for general values of $\alpha$ remain less explored.

We establish several regularity properties, including finiteness conditions, continuity with respect to the signal-to-noise ratio (SNR) and the input distribution, and strict concavity/convexity properties that guarantee uniqueness in associated optimization problems. 

A central contribution of this work is the development of an $\alpha$-I-MMSE relationship, which generalizes the classical identity by relating the derivative of $\alpha$-mutual information with respect to SNR to the MMSE evaluated under appropriately tilted distributions. This connection further leads to a generalized de Bruijn's identity and new estimation-theoretic representations of R\'enyi  entropy and differential R\'enyi  entropy.

We also show that in the low-SNR regime the first-order behavior depends only on the variance of the input, mirroring the classical case. In the high-SNR regime, for discrete distributions, we show that $\alpha$-mutual information converges to the R\'enyi  entropy of order $1/\alpha$ of the input. Furthermore, for general distributions, we establish connections between the high-SNR behavior of $\alpha$-mutual information and $\alpha$-information dimension.

Overall, our results demonstrate that many of the fundamental relationships between information measures and estimation-theoretic quantities extend beyond the Shannon setting, albeit in a modified form involving $\alpha$-tilted distributions.
\end{abstract}



\section{Introduction}
In~\cite{renyi1961measures}, R\'enyi  introduced a one-parameter generalization of Shannon entropy and relative entropy, which are now known as R\'enyi  entropy and R\'enyi  divergence, respectively. These measures have found a wide variety of applications in source coding~\cite{courtade2014cumulant},  hypothesis testing~\cite{blahut1974hypothesis,han1989strong,ben2003renyi,sason2017arimoto,bruno2026finite}, and guessing~\cite{arikan2002inequality}; an interested reader is referred to~\cite{rioul2021primer} for an overview of this topic.  
In contrast, for mutual information, several different generalizations have been proposed: Arimoto's mutual information~\cite{arimoto1977information}, Sibson's mutual information~\cite{sibson1969information}, Csisz\'ar's mutual information~\cite{csiszar2002generalized} (see also Augustin~\cite{Augustin1978NoisyChannels}), and Lapidoth and Pfister's mutual information~\cite{LapidothPfister}; the reader is referred to~\cite{Augustin1978NoisyChannels} and~\cite{esposito2025sibson} for a comprehensive review of the properties of these measures. 

In this work, we focus on {\em Sibson’s mutual information}, which we refer to throughout as $\alpha$-mutual information. The $\alpha$-mutual information has a long history in information theory, with classical applications in hypothesis testing and coding~\cite{csiszar2002generalized,polyanskiy2010arimoto}, and has recently seen renewed interest in statistical learning~\cite{esposito2021generalization}, and modern estimation theory~\cite{saito2022meta}. Despite this renewed attention, many of its fundamental structural properties remain significantly less understood than those of classical mutual information.

The goal of this work is to investigate properties of $\alpha$-mutual information in the context of the classical additive Gaussian noise channel. The Gaussian setting plays a distinguished role in information theory and estimation theory, serving as a canonical model in which sharp characterizations and tight bounds are often attainable. For the special case $\alpha=1$, the Gaussian noise model admits explicit formulas and deep connections to estimation-theoretic quantities~\cite{GuoIT2005,wu2011derivative}, such as the minimum mean-square error (MMSE), Fisher information, and the I-MMSE relationship, making it a natural benchmark for comparison. By focusing on the Gaussian noise channel, we are able to isolate which properties of classical mutual information are intrinsic to the Gaussian structure and which are specific to the Shannon case 
$\alpha=1$. In particular, since the behavior of mutual information under Gaussian noise is well understood, this setting allows us to systematically examine how, and to what extent, these properties generalize to $\alpha$-mutual information.

\subsection{Literature Review}
Several applications of $\alpha$-mutual information to classical problems have been reported in~\cite{verdu2015alpha}. In particular, Shannon’s zero-error capacity of discrete memoryless channels with feedback~\cite{shannon2003zero} can be expressed in terms of $\alpha$-mutual information of order zero. Gallager’s random-coding analysis of the maximum-likelihood decoder~\cite{gallager2003simple} can be formulated using $\alpha$-mutual information. Moreover, certain converse results for variable-to-fixed data transmission~\cite{courtade2014variable} admit representations in terms of $\alpha$-mutual information, as does Arimoto’s converse on the probability of successful decoding achievable by any code over a discrete memoryless channel~\cite{arimoto1973converse,polyanskiy2010arimoto}. The interested reader is referred to~\cite{verdu2021error} for an in-depth treatment of the connections between $\alpha$-mutual information and error exponents.
More recently,~\cite{yagli2019exact} fully characterized the exponent of the celebrated soft-covering lemma and showed that it depends on $\alpha$-mutual information. The maximization of $\alpha$-mutual information has also been connected to privacy measures, e.g., maximal leakage~\cite{issa2019operational} and $\alpha$-maximal leakage~\cite{liao2019tunable}. In addition, generalization error bounds between empirical risk and population risk in statistical learning have been established, both in expectation and with high probability, using $\alpha$-mutual information~\cite{esposito2021generalization}. In the estimation-theoretic setting, $\alpha$-mutual information has been used to derive Bayesian lower bounds that require minimal regularity assumptions on the prior~\cite{saito2022meta,xu2016information,esposito2021lower,esposito2024lower,omanwar2025bounds}. Connections between problems of universal prediction and $\alpha$-mutual information have also been established in~\cite{universalPredictionGaspar}.

Generic properties of $\alpha$-mutual information, such as convexity and the data-processing inequality, have been studied in~\cite{verdu2015alpha,HoISIt2015}. More recently,~\cite{esposito2025sibson} established several important properties of $\alpha$-mutual information, with a particular emphasis on variational characterizations. In~\cite{kamatsuka2025alternating}, the authors also derived variational representations with the goal of constructing algorithms to compute the maximum $\alpha$-mutual information. The problem of maximizing $\alpha$-mutual information, both with and without constraints, has been thoroughly studied in~\cite{cai2019conditional}.
Derivatives of $\alpha$-information measures along the heat flow have been studied in several works~\cite{guo2009relative,valero2017generalization,wu2025entropic,jizba2021renyi}. These results, however, do not directly extend to $\alpha$-mutual information, whose representation as a R\'enyi divergence involves an additional optimization over the output distribution.

For the special case $\alpha=1$, several results are known for the Gaussian noise channel. A complete characterization of the necessary and sufficient conditions for finiteness was provided in~\cite{FunctionalMMSE}. The derivative of mutual information with respect to the signal-to-noise ratio (SNR) was established~\cite{GuoIT2005}, yielding the celebrated I-MMSE relationship; the interested reader is referred to~\cite{GuoIT2011,GuoNow2012,dytso2017view} for a comprehensive overview of related results.   Finally, the high-SNR behavior of mutual information has been connected to the information dimension through a series of works, including~\cite{guionnet2007classical,FunctionalMMSE,wu2011mmse}. Information dimension has also found operational meaning in universal compressed sensing, where it serves as a measure of source complexity~\cite{jalali2017universal}.

\subsection{Outline and Main Contributions }
This paper is concerned with Sibson's mutual information, which throughout the paper we will term as $\alpha$-mutual information. Specifically, we consider the practically important Gaussian noise setting and derive several properties and identities for $\alpha$-mutual information that mirror those of classical mutual information.

The outline of the paper and its main contributions are as follows:
\begin{itemize}
\item Section~\ref{sec:preliminaries} presents the preliminaries, including the required definitions of quantities such as R\'enyi  entropy, R\'enyi  divergence, and $\alpha$-mutual information. In Section~\ref{sec:preliminaries}, Lemma~\ref{lemma:GaussianAlphaMI} provides several equivalent definitions of $\alpha$-mutual information for the additive Gaussian noise case, and Lemma~\ref{lemma:RelationAlphaMIGaussianUniform} establishes a relationship with $\alpha$-mutual information of an additive uniform noise channel.

\item In Section~\ref{sec:reg_cond_MI}, we study several regularity properties of $\alpha$-mutual information:
\begin{itemize}
\item In Section~\ref{sec:Finiteness}, Theorem~\ref{thm:FinitenessAlphaMI} shows that $\alpha$-mutual information is finite for every input distribution and any SNR when $\alpha \in (0,1)$. For the case of $\alpha > 1$, several equivalent necessary and sufficient conditions for finiteness of $\alpha$-mutual information are established. Moreover, Proposition~\ref{prop:moment_conditions} provides convenient sufficient conditions based on moments of the input for finiteness of $\alpha$-mutual information.  

\item In Section~\ref{sec:ContinuitySNR0}, Proposition~\ref{prop:LimitSNRToZeroAlphaMI} shows that $\alpha$-mutual information is continuous at zero SNR. This continuity property is important for studying the low-SNR behavior of $\alpha$-mutual information as well as developing connections between $\alpha$-mutual information and MMSE.  

\item In Section~\ref{sec:ContDistr}, Theorem~\ref{thm:continuity_distribution} shows that  $\alpha$-mutual information is continuous in distribution.  

\item In Section~\ref{sec:ConcConv}, Proposition~\ref{prop:ConcConv} strengthens previous concavity and convexity results for $\alpha$-mutual information by showing strict concavity and strict convexity, which implies that the optimization of $\alpha$-mutual information over well-behaved convex domains leads to unique maximizers.  
\end{itemize}

\item In Section~\ref{sec:alphaIMMSE}, we study the interplay between $\alpha$-mutual information and estimation measures, such as Fisher information and MMSE:
\begin{itemize}
    \item In Section~\ref{sec:brown}, Theorem~\ref{theorem:GeneralizationBrown} presents a generalization of Brown's identity to $\alpha$-tilted distributions, which naturally arise in the context of $\alpha$-mutual information. In particular, a connection between Fisher information of an $\alpha$-tilted output distribution and MMSE is established.
    \item In Section~\ref{sec:gen_I-MMSE}, Theorem~\ref{thm:IMMSE} derives a generalization of the I-MMSE relationship. In particular, it is shown that the derivative of $\alpha$-mutual information with respect to SNR is equal to the MMSE of the corresponding $\alpha$-tilted distributions.
    \item Section~\ref{sec:DeBruijnIdentity} leverages the newly established I-MMSE relationship to derive a generalization of the de Bruijn's identity and shows that it is equivalent to the formulation recently proved in~\cite{wu2025entropic}.
    \item In Section~\ref{sec:continutity_MMSE_zero_snr}, Proposition~\ref{prop:continutity_mmse_zero_snr} studies the continuity of the MMSE of $\alpha$-tilted distributions as SNR approaches zero. This result is instrumental in establishing an integral generalization of the I-MMSE relationship for the $\alpha$-mutual information case.
    \item In Section~\ref{sec:LowSNRAlphaMI}, Proposition~\ref{prop:LowSNRAlphaMI} builds on the newly derived I-MMSE relationship, the continuity of $\alpha$-mutual information, and the value of the MMSE at zero SNR, and generalizes an important result: the first-order term of $\alpha$-mutual information around zero SNR is given by the variance (multiplied by a factor $\alpha/2$).
\end{itemize}
\item In Section~\ref{sec:High_SNR_Behaviour}, we study the high-SNR behavior of $\alpha$-mutual information:
\begin{itemize}
    \item In Section~\ref{sec:DiscreteHighSNR}, Proposition~\ref{prop:m-point} shows that, for discrete distributions, $\alpha$-mutual information converges to the R\'enyi entropy of order $1/\alpha$.
    \item  Section~\ref{sec:new_entropic_identities} uses the newly developed I-MMSE relationship from Section~\ref{sec:alphaIMMSE} (together with the result of Section~\ref{sec:DiscreteHighSNR} for the discrete case) to provide new estimation-theoretic representations of R\'enyi  entropy and differential R\'enyi  entropy. The results in Proposition~\ref{prop:NewEntrDiscr} and Proposition~\ref{prop:NewEntrCont} parallel earlier results on Shannon entropy. 
    \item  Section~\ref{sec:alpah_info_dim} establishes a fundamental connection between the high-SNR behavior of $\alpha$-mutual information and R\'enyi information dimension.
\end{itemize}
\item Section~\ref{sec:Conclusion} concludes the paper  and discusses several interesting directions for future research.
\end{itemize}

\subsection{Notation}
Random variables are denoted by uppercase sans-serif letters; for a random variable $\sfX$, we indicate its expected value with $\bbE[\sfX]$;
Calligraphic letters indicate sets. The notation $\mathbbm 1_{\mathcal{S}}$ denotes the indicator function over the set $\mathcal{S}$. For a set $\cS$, we let $|\cS|$ denote its cardinality. All logarithms are in base $\rm{e}$. With $\lfloor \cdot \rfloor$ we denote the floor function.
For a random variable $\sfX$ and a set $\cA$, we let $\bbE^\alpha[\sfX] = (\bbE[\sfX])^\alpha$ and $P_\sfX^\alpha(\cA) = (P_\sfX(\cA))^\alpha$.

\section{Preliminaries, Gaussian Noise Channel, and $\alpha$-Mutual Information}
\label{sec:preliminaries}

In this section, we collect definitions and a few preliminary results about various $\alpha$-R\'enyi  information measures, which we will use extensively throughout the paper.

\subsection{R\'enyi  Measures}
\begin{definition}[R\'enyi  entropy of order $\alpha$]
\label{def:ReyniEntropy}
Let $\sfX$ be a discrete random variable with alphabet $\cX$ and probability mass function (PMF) $P_{\sfX}$. Then, the R\'enyi  entropy of order $\alpha$, where $\alpha \in (0,1) \cup (1,\infty)$, 
is defined as
\begin{equation}
H_{\alpha}(\sfX) = \frac{1}{1-\alpha} \log \left( \sum_{x \in \cX}  P^\alpha_{\sfX}(x)  \right ). 
\end{equation}
The limiting value as $\alpha \to 1$ is the Shannon entropy and is given by
\begin{equation}
H_{1}(\sfX) = - \sum_{x \in \cX} P_{\sfX}(x) \log P_{\sfX}(x).
\end{equation}
\end{definition}
\begin{rem}
\label{rem:HalphaFinitealphaGreater1}
Note that for $\alpha >1$, we always have that $H_{\alpha}(\sfX) < \infty$. This is because we have that $\sum_{x \in \cX}  P^\alpha_{\sfX}(x) >0$, which for $\alpha>1$ implies that $H_{\alpha}(\sfX) < \infty$. 
Lemma~\ref{lem:bounds_on_renyi_entorpy} below provides a  condition that ensures that $H_{\alpha}(\sfX) < \infty$ when $\alpha \in (0,1)$.
\end{rem}
\begin{lemma}\label{lem:bounds_on_renyi_entorpy} \footnote{Lemma~\ref{lem:bounds_on_renyi_entorpy} is likely to exist in the literature. We, however, were not able to locate such a result.}
Let $\sfW$ be a non-negative integer random variable. Fix some $\alpha \in (0,1)$. Then, the following properties hold:
    \begin{itemize}
\item  For every $k> \frac{1-\alpha}{\alpha}$ such that $\bbE[ \sfW^k ]<\infty$, we have that $H_\alpha(\sfW)<\infty$; and
\item 
For $k =\frac{1-\alpha}{\alpha}$, there exists a $\sfW$  such that $\bbE[ \sfW^k ]<\infty$  but $H_\alpha(\sfW)=\infty$. 
\end{itemize}
\end{lemma}
\begin{IEEEproof}
The proof is given in Appendix~\ref{app:proof_entropy_moments}. 
    \end{IEEEproof}

We now also present the continuous counterpart of Definition~\ref{def:ReyniEntropy}. 
\begin{definition}[Differential R\'enyi  entropy of order $\alpha$]
\label{def:diff_renyi_entropy} Let $\sfX \in \bbR$ be a continuous random variable with probability density function (PDF) $f_\sfX$.  The differential R\'enyi entropy of order $\alpha $, where $\alpha \in (0,1) \cup (1,\infty)$, is defined as
\begin{equation}
    h_\alpha(\sfX) =  \frac{1}{1-\alpha} \log \left( \int_\bbR f^\alpha_\sfX(x) \, \rmd x\right). 
\end{equation}
\end{definition}
As an important example, consider $\sfZ \sim \mathcal{N}(0,1)$, see for example~\cite[eq.~127]{ram2016renyi}, for which we have: 
\begin{equation}
    h_\alpha(\sfZ) =  \frac{1}{2 (\alpha-1)} \log (\alpha) + \frac{1}{2} \log(2\pi) . \label{eq:diff_renyi_entropy_Gaussian}
\end{equation}
\begin{definition}[R\'enyi  divergence of order $\alpha$]
\label{def:RDalpha}
The R\'enyi  divergence of order $\alpha \in (0,1) \cup (1,\infty)$, is defined as
\begin{equation}
D_{\alpha}\left(P \| Q \right ) = \frac{1}{\alpha-1}  \log \int p^\alpha q^{1-\alpha} \ {\rmd}\mu,
\end{equation}
where $p$ and $q$ are the Radon-Nikodym derivatives of the probability measures $P$ and $Q$, respectively, with respect to a common dominating $\sigma$-finite measure $\mu$.
\end{definition}
For a comprehensive overview of properties of R\'enyi divergence,  and more generally $f$-divergences, the interested reader is referred to~\cite{van2014renyi} and~\cite{sason2016f}, respectively.
\begin{definition}[$\alpha$-mutual information]
\label{def:MIalpha}
Let $P_\sfX \to P_{\sfY|\sfX} \to P_\sfY$ with $\sfY \in \cY$ and $\sfX \in \cX$. The $\alpha$-mutual information with $\alpha \in (0,1) \cup (1,\infty)$ is defined as
\begin{equation}
I_{\alpha}(\sfX;\sfY) 
 = \inf_{Q_{\sfY}} D_{\alpha}\left( P_{\sfX,\sfY} \| P_{\sfX} \times Q_{\sfY} \right ), \label{eq:Def_I}
\end{equation}
with $D_{\alpha}(\cdot \| \cdot)$ being the R\'enyi  divergence of order $\alpha$ in Definition~\ref{def:RDalpha}.
\end{definition}
The minimizer in~\eqref{eq:Def_I}, provided that it exists, is denoted by
$P_{\sfY_\alpha}$ and characterized by
\begin{equation}
\frac{\rmd P_{\sfY_\alpha}}{\rmd P_\sfY}(y)
=
\kappa_\alpha
\bbE^{ \frac{1}{\alpha} }\left[
\left(
\frac{\rmd P_{\sfY|\sfX=\sfX}}{\rmd P_\sfY}(y)
\right)^\alpha
\right], \, \quad y\in \cY,
 \label{eq:MinimizerAlphaMI}
\end{equation}
where $\kappa_\alpha$ is the normalizing constant.

The high-SNR behavior of $\alpha$-mutual information will be governed by R\'enyi  $\alpha$-information dimension~\cite{renyi1970probability,csiszar1962dimension}, which is defined next.  
\begin{definition}[$\alpha$-information dimension]
\label{def:AlphaInfDim}
Let $\alpha \in (0,1) \cup (1,\infty)$. The information dimension of order $\alpha$ of $\sfX$ is defined as
\begin{equation}
d_\alpha(\sfX)
=\lim_{\varepsilon\downarrow 0}
\frac{\log \bbE \left[P^{\alpha-1}_{\sfX}(\cB_\infty(\sfX,\varepsilon))\right]}
{(\alpha-1)\log \varepsilon} ,
\label{eq:def_dalpha}
\end{equation}
where we let\footnote{If the limit does not exist, then we use $\liminf$ and $\limsup$ to define the lower and upper $\alpha$-information dimension, respectively.}  $\cB_{\infty}(y,\varepsilon) = [y-\varepsilon, y+ \varepsilon]$. 
\end{definition}
There are several equivalent characterizations of $\alpha$-information dimension, such as the one below in terms of a uniform quantizer~\cite{WuAnlogCompr}:
\begin{equation}
d_\alpha(\sfX)=
\lim_{m \to \infty} \frac{H_{\alpha} \left(\langle\sfX\rangle_m \right)}{\log(m)},
\end{equation}
where $m \in \bbN$ and $\langle\sfX\rangle_m = \frac{\lfloor m \sfX\rfloor }{m}$. The next lemma provides bounds on the R\'enyi  entropy of order $\alpha$ that will be helpful in proving some regularity properties of $\alpha$-mutual information. 
\begin{lemma}
\label{lemma:UBLBHalpha}
Let $\alpha \in (0,1) \cup (1,\infty)$ and let $m \in \bbN$. Then, it holds that
\begin{equation}
H_{\alpha}(\lfloor \sfX \rfloor)  \leq H_{\alpha}(\langle\sfX\rangle_m) \leq H_{\alpha}\left(\lfloor \sfX \rfloor  \right ) + \log(m),
\label{eq:LBUBQuantizedX}
\end{equation}
where $\langle\sfX\rangle_m = \frac{\lfloor m \sfX\rfloor }{m}$.
\end{lemma}
\begin{IEEEproof}
A looser version of this lemma can be derived from~\cite[Lemma~7]{wu2014information} by noting that
\begin{align}
|  \lfloor \sfX \rfloor -  \langle\sfX\rangle_m | &\le 1,\\
d_{\min} (  \lfloor \sfX \rfloor )  &\ge  1, \\
d_{\min} (  \langle\sfX\rangle_m )  &\ge \frac{1}{m},
 \end{align}
 where $d_{\min}(\sfU)$ denotes the minimum distance between the support points of the random variable $\sfU$. The current version, with slightly better constants, is shown in Appendix~\ref{app:lemma:UBLBHalpha}. 
\end{IEEEproof}

\subsection{Gaussian Noise Channel and Relevant Distributions}
\label{sub:GaussianNoise}
In this work, we consider the additive Gaussian noise channel defined as
\begin{equation}
\label{eq:GaussianChannel}
\sfY =  \sfX + \frac{1}{\sqrt{\mathrm{snr}}}\sfZ,
\end{equation}
where $\sfZ$ is the standard Gaussian random variable independent of $\sfX$. We denote the PDF of $\sfZ$  by $\phi$ and the PDF of $\sfY$ by $f_\sfY$, which is given by
\begin{equation}
\label{eq:fYsnr}
    f_\sfY(y; {\rm snr}) = \sqrt{{\rm snr}} \,  \bbE \left [  \phi \left( \sqrt{{\rm snr}} (y-\sfX)  \right) \right].
\end{equation}
The conditional PDF of $\sfY$ given $\sfX$ is denoted as $f_{\sfY| \sfX ;{\rm snr} }(y|x)$ and given by
\begin{equation}
f_{\sfY| \sfX ;{\rm snr} }(y|x) = \sqrt{\rm snr} \; \phi\left(\sqrt{  \rm{snr}} \; (y - x)\right).
\label{eq:PDFYGivenX}
\end{equation}
For the Gaussian channel model in~\eqref{eq:GaussianChannel}, the minimizer in~\eqref{eq:MinimizerAlphaMI} has density with respect to Lebesgue measure, which is given by
\begin{equation}
    f_{\sfY_\alpha}(y) =  \frac{f_\sfY^{ \frac{1}{\alpha}}(y; \alpha \; {\rm snr}) }{ \int_{\bbR} f_\sfY^{ \frac{1}{\alpha}}(t; \alpha \; {\rm snr}) \; \rmd t},  \, \quad  y \in \bbR,\label{eq:PDFYalpha}
\end{equation}
which exists provided that $\int_{\bbR} f_\sfY^{ \frac{1}{\alpha}}(t; \alpha \; {\rm snr}) \; \rmd t <\infty$.  
In addition to $f_{\sfY_\alpha}$, it is also useful to define the following distributions.
Let $P_{\sfX|\sfY=y;\alpha \;\rm{snr}}$ denote the posterior distribution for the Gaussian noise channel in~\eqref{eq:GaussianChannel} with SNR of $\alpha\; \mathrm{snr}$. Define
\begin{equation}
P_{\sfX_\alpha|\sfY_\alpha=y}(x)=
P_{\sfX|\sfY=y ; \alpha \; \rm{snr}} (x) ,
\label{eq:conditional_expression}
\end{equation}
which together with~\eqref{eq:PDFYalpha}, defines the joint distribution of $(\sfX_\alpha,\sfY_\alpha)$ given by
\begin{equation}
    \rmd P_{\sfX_\alpha,\sfY_\alpha}(x,y) = f_{\sfY| \sfX ;\alpha \; {\rm snr} }(y|x)   \frac{ f_{\sfY_\alpha}(y) }{f_\sfY(y; \alpha \; {\rm snr})}     \rmd P_{\sfX}(x) \; \rmd y.  \label{eq:joint_expression}
\end{equation}
\begin{rem}
With reference to~\eqref{eq:PDFYalpha}, note that for $\alpha=1$, we have that $f_{\sfY_1} = f_\sfY(\cdot; {\rm snr})$.
Note also that, from the perspective of the conditional distribution, $\rmd P_{\sfX_\alpha|\sfY_\alpha=y}$, the $\alpha$ gets absorbed inside $\rm{snr}$. 
\end{rem}
In general, $\rmd P_{\sfX_\alpha}$ is different from $\rmd P_{\sfX}$. To see this, consider the following example.

\begin{exa} Let $\sfX \sim \cN(0,\sigma^2_\sfX)$. From~\eqref{eq:joint_expression}, it follows that $\sfX_\alpha \sim \cN \left( 0,\sigma^2_{\sfX_\alpha} \right )$, where
\begin{equation}
\sigma^2_{\sfX_\alpha} = \frac{\sigma^2_{\sfX}\left(\alpha^2 \; {\rm{snr}} \; \sigma^2_{\sfX} +1 \right )}{1+\alpha \; {\rm{snr}} \; \sigma^2_{\sfX}}.
\end{equation}
We note that $\sigma^2_{\sfX_\alpha} = \sigma^2_{\sfX}$ if and only if $\alpha=1$ or $\rm{snr} \to 0$.
\end{exa}

\subsection{$\alpha$-Mutual Information in Gaussian Noise}
We now characterize $\alpha$-mutual information in Definition~\ref{def:MIalpha} for the Gaussian noise channel in~\eqref{eq:GaussianChannel}. 
In particular, we provide several different representations of this $\alpha$-mutual information, which we refer to as $I_{\alpha}(\sfX;\mathrm{snr})$.
\begin{lemma}
\label{lemma:GaussianAlphaMI}
Consider the channel model in~\eqref{eq:GaussianChannel}. Then, for $\alpha \in (0,1) \cup (1,\infty)$, it holds that
\begin{align}
        I_{\alpha}(\sfX;\mathrm{snr})&=\frac{\alpha}{\alpha-1} \log \left(   \int_{\bbR} \bbE^{ \frac{1}{\alpha} } \left[    \left( \sqrt{\rm snr} \; \phi \left (\sqrt{\rm snr} \left( y-  \sfX  \right ) \right ) \right)^\alpha \right] \rmd y  \right) \label{eq:AlphaMI} \\
        &= \frac{\alpha}{\alpha-1} \log \left( \int_{\bbR} f_\sfY^{ \frac{1}{\alpha}}(y; \alpha \; {\rm snr})  \; \rmd y  \right)  + \Delta(\alpha, {\rm snr}) \label{eq:AlphaMISecondRepresentation}  \\
        &= \frac{1}{\alpha-1} \log \left( \|  f_\sfY(\cdot; \alpha \; {\rm snr})  \|_{ \frac{1}{\alpha}} \right) + \Delta(\alpha, {\rm snr})   \label{eq:norm_representation} \\ 
        &= h_{ \frac{1}{\alpha} } \left(\sfX  + \frac{1}{\sqrt{ \alpha \, \rm snr}} \sfZ \right) - h_\alpha(\sfZ) +\frac{1}{2} \log({\rm snr})  \label{eq:alpha_mi_diff_entropy_Rep_Ver1} \\
        &= h_{ \frac{1}{\alpha} } \left( \sqrt{\rm snr} \, \sfX  + \frac{1}{\sqrt{ \alpha }} \sfZ \right) -  h_{ \alpha }(\sfZ), \label{eq:alpha_mi_diff_entropy_Rep_Ver2} 
\end{align}
where 
\begin{equation}
\Delta(\alpha, {\rm snr})= \frac{1}{2(1-\alpha)}\log ( \alpha ) + \frac{1}{2}\log({\rm snr}) -  \frac{1}{2}\log( 2 \pi).
\end{equation}
\end{lemma}
\begin{IEEEproof}
The proof of~\eqref{eq:AlphaMI} is provided in Appendix~\ref{app:GaussianAlphaMI}. The representation in~\eqref{eq:AlphaMISecondRepresentation} follows by using the definition of $f_\sfY$ in~\eqref{eq:fYsnr} and the representation in~\eqref{eq:norm_representation} follows by letting
\begin{equation}
\|  f_\sfY(\cdot; \alpha \; {\rm snr})  \|_{ \frac{1}{\alpha}} = \left(\int_{\bbR} f_\sfY^{ \frac{1}{\alpha}}(t; \alpha \; {\rm snr})  \; \rmd t \right )^\alpha.
\end{equation}
The proof of~\eqref{eq:alpha_mi_diff_entropy_Rep_Ver1} follows from Definition~\ref{def:diff_renyi_entropy} and the expression of $h_\alpha(\sfZ)$ in~\eqref{eq:diff_renyi_entropy_Gaussian}.
Finally, the representation in~\eqref{eq:alpha_mi_diff_entropy_Rep_Ver2} follows from the property that for any $a\in \bbR$ we have that $h_\alpha(a \sfU) =  h_\alpha(\sfU) +\log|a|$.
This concludes the proof of Lemma~\ref{lemma:GaussianAlphaMI}.
\end{IEEEproof}
\begin{exa}\label{ex1}
Let $\sfX \sim \cN(0,1)$. 
From Lemma~\ref{lemma:GaussianAlphaMI} we recover the result in~\cite[Example~5]{verdu2015alpha}, i.e.,
\begin{equation}
    I_{\alpha}(\sfX;\mathrm{snr}) = \frac{1}{2}  \log \left( 1+\alpha \; {\rm{snr}}\right ),
\end{equation}
A graphical representation of this result is provided in Fig.~\ref{fig:IMMSEG} and Fig.~\ref{fig:VersusAlphaGaussian} (solid lines).
\end{exa}
\begin{exa}\label{ex2}
Let $\sfX$ be equiprobable on $\{-1,1\}$. From Lemma~\ref{lemma:GaussianAlphaMI}, we have that 
\begin{align}
     I_{\alpha}(\sfX;\mathrm{snr})
    & = \frac{\alpha}{\alpha-1} \log \left( \sqrt{\rm snr}  \int_{\bbR} \left( \frac{1}{2} \phi^\alpha \left( \sqrt{{\rm snr}} \left( y+1\right )\right ) + \frac{1}{2} \phi^\alpha \left( \sqrt{{\rm snr}} \left( y-1\right )\right )\right)^{\frac{1}{\alpha}} \rmd y  \right)
\\& = \frac{\alpha}{\alpha-1} \log \left( \sqrt{\rm snr} \int_\bbR \frac{1}{\sqrt{2 \pi}} \exp \left( - \frac{{\rm{snr}} \; (y^2+1)}{2}\right ) \left( \cosh \left( \alpha \; {\rm{snr}} \; y \right ) \right )^{\frac{1}{\alpha}}\rmd y  \right)
\\& = \frac{\alpha}{\alpha-1} \log \left( \exp \left( - \frac{{\rm{snr}}}{2}\right ) \int_\bbR \frac{1}{\sqrt{2 \pi}} \exp \left (- \frac{z^2}{2} \right ) \left( \cosh \left( \alpha \; \sqrt{{\rm{snr}}} \; z \right ) \right )^{\frac{1}{\alpha}}\rmd z \right)
\\& = \frac{\alpha}{\alpha-1} \left( \log \left( \bbE \left [ \left( \cosh \left(\alpha \, \sfZ \, \sqrt{{\rm{snr}}} \right ) \right )^{\frac{1}{\alpha}} \right ] \right ) -\frac{{\rm{snr}}}{2}\right ),
\label{eq:alphaMIBinary}
\end{align}
where $\sfZ$ is the standard normal random variable. A graphical representation of this result is provided in Fig.~\ref{fig:IMMSERB} and Fig.~\ref{fig:VersusAlphaBinary} (solid lines).
\end{exa}

\subsection{$\alpha$-Mutual Information in Uniform Noise}
Similarly to the Gaussian noise case, we now provide an expression for $\alpha$-mutual information for the additive uniform noise channel defined as 
\begin{equation}
\label{eq:UnifoNoiseChannel}
\sfY_{\varepsilon} = \sfX +\varepsilon \sfU,
\end{equation}
where $\varepsilon >0$ and where $\sfU \sim{\sf Unif} \left ( -\frac{1}{2},\frac{1}{2} \right )$ independent of $\sfX$.
\begin{lemma}
\label{lemma:UniformMIProb}
Consider the channel model in~\eqref{eq:UnifoNoiseChannel}. Then, for $\alpha \in (0,1) \cup (1,\infty)$, it holds that
\begin{equation}
\label{eq:ExpIAlpha1}
I_\alpha(\sfX;\sfY_{\varepsilon}) = - \frac{\alpha}{\alpha-1} \log(\varepsilon) + \frac{\alpha}{\alpha-1} \log \left( \int_\bbR P_{\sfX}^{ \frac{1}{\alpha}} \left(\cB_\infty \left (y,\frac{\varepsilon}{2} \right )\right) \, \rmd y \right ),
\end{equation}
or equivalently,
\begin{equation}
\label{eq:ExpIAlpha2}
I_\alpha(\sfX;\sfY_{\varepsilon}) = \frac{\alpha}{\alpha-1} \log \left( \bbE \left [ \left( \frac{1}{P_{\sfX}\left(\cB_\infty \left (\sfY_{\varepsilon}, \frac{\varepsilon}{2}  \right )\right)}\right )^{\frac{\alpha-1}{\alpha}}\right ] \right ).
\end{equation}
\end{lemma}
\begin{IEEEproof}
The proof is given in Appendix~\ref{app:proofUniformNoiseAlphaMI}. 
    \end{IEEEproof}
We now provide a relationship between $\alpha$-mutual information for the Gaussian noise channel in~\eqref{eq:GaussianChannel} and that of the uniform noise channel in~\eqref{eq:UnifoNoiseChannel}. This result will play a key role in the characterization of $\alpha$-information dimension in Section~\ref{sec:High_SNR_Behaviour}.

\begin{lemma}
\label{lemma:RelationAlphaMIGaussianUniform}
Consider the channel models in~\eqref{eq:GaussianChannel} and in~\eqref{eq:UnifoNoiseChannel}. Then, for $\alpha \in (0,1) \cup (1,\infty)$, it holds that
\begin{equation}
\label{eq:RelationAlphaMIGaussianUniform}
c_\alpha \exp \left( \frac{\alpha-1}{\alpha} I_\alpha(\sfX;\sfY_{\varepsilon}) \right ) \leq \exp \left( \frac{\alpha-1}{\alpha} I_\alpha(\sfX;\sfX +\varepsilon \sfZ) \right ) \leq  C_\alpha \exp \left( \frac{\alpha-1}{\alpha} I_\alpha(\sfX;\sfY_{\varepsilon}) \right ),
\end{equation}
where $c_\alpha$ and $C_\alpha$ are positive constants that only depend on $\alpha$. 
\end{lemma}
\begin{IEEEproof}
The proof is provided in Appendix~\ref{app:RelationAlphaMIGaussianUniform}
\end{IEEEproof}

\section{Regularity Properties of $\alpha$-Mutual Information}
\label{sec:reg_cond_MI}
In this section, we establish several technical regularity properties of $\alpha$-mutual information characterized in Lemma~\ref{lemma:GaussianAlphaMI}. These conditions are needed to ensure that the fundamental properties of mutual information continue to hold in the $\alpha$-setting. In particular, Section~\ref{sec:Finiteness} provides conditions guaranteeing finiteness; Section~\ref{sec:ContinuitySNR0} shows continuity at $\mathrm{snr} = 0^+$; Section~\ref{sec:ContDistr} proves continuity with respect to the input distribution $P_{\sf X}$; and Section~\ref{sec:ConcConv} analyzes concavity and convexity as a function of $P_\sfX$. These properties will be important  for providing a rigorous generalization of the I-MMSE relationship in Section~\ref{sec:alphaIMMSE}, also guaranteeing that optimization problems involving $\alpha$-mutual information are well-behaved (e.g., have unique solutions).

\subsection{Finiteness}
\label{sec:Finiteness}
We show some equivalent conditions for the finiteness of $\alpha$-mutual information. The case of $\alpha=1$ has been fully characterized in \cite{FunctionalMMSE} and is not treated here. 
\begin{theorem}
    \label{thm:FinitenessAlphaMI}
We have the following two cases: 
\begin{enumerate}
\item Case 1: $\alpha \in (0,1)$.   For every $\sfX$ and $0<{\rm snr}<\infty$, we have that  $I_\alpha(\sfX;{\rm snr})<\infty$ and $0\le d_{\frac{1}{\alpha}}(\sfX) \le1$.
\item Case 2:  For  $\alpha \in (1,\infty)$, we have that the following statements are equivalent: 
\begin{enumerate}[label=(\alph*)]
\item   $I_\alpha(\sfX;{\rm snr})<\infty$ for some $0<{\rm snr}<\infty$; 
\item   $I_\alpha(\sfX;{\rm snr})<\infty$ for \emph{every} $0<{\rm snr}<\infty$;
\item  $\|  f_\sfY(\cdot; \alpha \; {\rm snr})  \|_{ \frac{1}{\alpha}}<\infty$ for some $0<{\rm snr}<\infty$; 
\item  $\|  f_\sfY(\cdot; \alpha \; {\rm snr})  \|_{ \frac{1}{\alpha}}<\infty$ for \emph{every} $0<{\rm snr}<\infty$;
\item $f_{\sfY_\alpha}$ exists for some $0<{\rm snr}<\infty$;
\item $f_{\sfY_\alpha}$ exists for \emph{every} $0<{\rm snr}<\infty$;
\item $H_{ \frac{1}{\alpha}}( \lfloor \sfX \rfloor )<\infty$;
\item  $0 \leq d_{\frac{1}{\alpha}}(\sfX) \leq 1$.
\end{enumerate}
\end{enumerate}
\end{theorem} 
\begin{IEEEproof}
The proof can be found in Appendix~\ref{app:FinitenessAlphaMI}.
\end{IEEEproof}

The condition $H_{\frac{1}{\alpha}}(\lfloor \sfX \rfloor)<\infty$ in Theorem~\ref{thm:FinitenessAlphaMI} is helpful since it provides a sufficient and necessary condition for the finiteness of $\alpha$-mutual information as a function of the channel input $\sfX$. However, in practice, we might want to amend this with a more practical moment condition. The next result provides a sufficient condition for the finiteness of $\alpha$-mutual information in terms of some moments of $\sfX$. 
\begin{prop} \label{prop:moment_conditions}
    Fix some $\alpha>1$ and ${\rm snr}>0$.  Then, the following properties hold:
\begin{itemize}
\item  For every $k >\alpha-1$ such that $\bbE[|\sfX|^k]<\infty$, we have that    $I_\alpha(\sfX;{\rm snr})<\infty$; and 
\item  For $k= \alpha-1$, there exists  an $\sfX$ such that $\bbE[|\sfX|^k]<\infty$ but $I_\alpha(\sfX;{\rm snr})=\infty$.  
\end{itemize}
\end{prop}
\begin{IEEEproof}
Assume that $\bbE[|\sfX|^k]<\infty$ for some $k>\alpha-1$. Then, since $\left| \lfloor \sfX \rfloor \right| \leq |\sfX| +1$, we have that
\begin{equation}
\bbE\left[ \left| \lfloor \sfX \rfloor \right|^k \right]<\infty.
\end{equation}
Since  
$\left| \lfloor \sfX \rfloor \right|$ 
is a non-negative
integer-valued random variable and $k>\frac{1-\frac{1}{\alpha}}{\frac{1}{\alpha}} = \alpha -1$ with $\alpha>1$, Lemma~\ref{lem:bounds_on_renyi_entorpy} yields
\begin{equation}
H_{\frac{1}{\alpha}}(|\lfloor \sfX \rfloor |)<\infty.
\label{eq:using_lemma_ent_fini}
\end{equation}
Now we can use the inequality for integer-valued random variables
\begin{equation}
H_{\frac{1}{\alpha}}(\lfloor \sfX \rfloor ) \leq H_{\frac{1}{\alpha}}(|\lfloor \sfX \rfloor |) +\log(2),
\end{equation}
which gives
\begin{equation}
H_{\frac{1}{\alpha}}(\lfloor \sfX \rfloor )<\infty,
\end{equation}
and hence, by Theorem~\ref{thm:FinitenessAlphaMI},
\begin{equation}
I_\alpha(\sfX;{\rm snr})<\infty.
\label{eq:final_implication_on_suff}
\end{equation}
This proves the first property. The second property follows by choosing $\sfX$ as in Lemma~\ref{lem:bounds_on_renyi_entorpy}.
\end{IEEEproof}

\begin{rem}
Problems involving the optimization of mutual information (i.e., the case $\alpha = 1$) under moment constraints are ubiquitous in the literature and have been studied extensively; see~\cite{smith1971information,dytso2018discrete,eisen2023capacity} among many others. In this classical setting, the mutual information remains finite under any finite moment constraint, which ensures that such optimization problems are always well-posed.

In contrast, for $\alpha$-mutual information with $\alpha > 1$, this property no longer holds. In fact, Proposition~\ref{prop:moment_conditions} implies that whenever ${0\leq}k \leq \alpha-1$, we have that
\begin{equation}
    \max_{\sfX:\, \mathbb{E}[|\sfX|^k] \le P} I_\alpha(\sfX;{\rm snr}) = \infty, \,  P>0. 
\end{equation}
Thus, moment constraints of order ${0\leq}k \le \alpha-1$ do not suffice to make the optimization problem well-defined. For example, under a practically important second-moment constraint, maximization of $I_\alpha(\sfX;{\rm snr})$ for $\alpha \ge 3$ leads to an infinite maximum value. 
\end{rem}

\subsection{Continuity at ${\rm snr} =0^+$}
\label{sec:ContinuitySNR0}
We show continuity of $\alpha$-mutual information at ${\rm snr} =0^+$.
\begin{prop}
\label{prop:LimitSNRToZeroAlphaMI}
For a fixed $\alpha\in(0,1)\cup(1,\infty)$, assume that there exists a $\delta>0$  such that\footnote{Note that Theorem~\ref{thm:FinitenessAlphaMI} provides alternative equivalent conditions for the finiteness of $\alpha$-mutual information.} $I_{\alpha}(\sfX;\delta)<\infty$.  Then, 
\begin{align}
\label{eq:TrivialLimitSNRToZero}
\lim_{\mathrm{snr} \to 0} I_{\alpha}(\sfX;\,\mathrm{snr}) = 0.
\end{align}
\end{prop}
\begin{proof}
The proof is provided in Appendix~\ref{app:LimitSNRToZeroAlphaMI}.
\end{proof}

Proposition~\ref{prop:LimitSNRToZeroAlphaMI} is a generalization of a similar statement for $\alpha=1$ shown in~\cite{GuoIT2005,FunctionalMMSE}. A more refined statement where we find the rate of convergence as ${\rm snr \to 0}$ will be given in Section~\ref{sec:alphaIMMSE}.   In the same section, the continuity result in Proposition~\ref{prop:LimitSNRToZeroAlphaMI} will also be needed to show an integral version of a generalization of the I-MMSE relationship to $\alpha$-mutual information. 

\subsection{Continuity in Distribution}
\label{sec:ContDistr}
Continuity of mutual information with respect to the input distribution is a fundamental property that is heavily used in problems involving the optimization of mutual information. Lower semicontinuity of mutual information dates back to Pinsker~\cite{Pinsker1964} (see also~\cite{AshInformationTheory}), where it was derived from the lower semicontinuity of relative entropy.  For the additive   Gaussian noise, continuity under support constraints was first shown in~\cite{smith1971information} and continuity under various moment constraints was later established in~\cite{eisen2023capacity}. 
We now present an analogous continuity result for general values of $\alpha > 0$. 

\begin{theorem}\label{thm:continuity_distribution}
   Fix some ${\rm snr}>0$ and suppose that $\sfX_n \to \sfX$ in distribution. Then: 
\begin{itemize}
 \item If $\alpha \in (0,1)$, we have that $I_{\alpha}(\sfX_n; {\rm snr}) \to I_{\alpha}(\sfX; {\rm snr})$.
    \item If $\alpha \in (1, \infty)$ and $\sup_n \bbE[|\sfX_n|^k] < \infty$ for some $k > \alpha$\footnote{Note that the condition $k>\alpha$ does not come from Proposition~\ref{prop:moment_conditions}, but is a consequence of a tail bound used in the proof in Appendix~\ref{app:Continuity}.}, we have that $I_{\alpha}(\sfX_n; {\rm snr}) \to I_{\alpha}(\sfX; {\rm snr})$.
\end{itemize}
\label{thm:Continuity}
\end{theorem}
\begin{proof} 
The proof is provided in Appendix~\ref{app:Continuity}.
\end{proof}

\subsection{Strict Concavity and Strict Convexity}
\label{sec:ConcConv}

Concavity and convexity properties of $\alpha$-mutual information have been studied in \cite{HoISIt2015} and \cite{esposito2025sibson}. However, in certain optimization problems involving mutual information, \emph{strict} concavity (or {\em strict} convexity) is required in order to establish uniqueness of the optimizer.
For the classical case $\alpha = 1$, strict concavity of mutual information for Gaussian noise channels was first shown in~\cite{smith1971information} in the context of amplitude-constrained channels.
We now extend this line of results by establishing strict concavity and strict convexity of a function of $\alpha$-mutual information.
\begin{prop}
\label{prop:ConcConv}
Fix ${\rm snr}>0$. Consider $\zeta_{\alpha}(\sfX; {\rm snr})=\exp \left( (\alpha-1) I_\alpha(\sfX;{\rm snr}) \right )$. Additionally, for $\alpha>1$, define $ \cD_\alpha= \left\{P_{\sf X}: I_\alpha(\sf X;{\rm snr})<\infty \right\}$. Then:
  \begin{itemize}
  \item  If $\alpha \in (1,\infty)$, we have that $\zeta_{\alpha}(\sfX; {\rm snr})$ is \emph{strictly} concave in the input distribution $P_{\sfX}$ on $\cD_\alpha$; and 
  \item  If $\alpha \in (0,1)$, we have that $\zeta_{\alpha}(\sfX; {\rm snr})$ is \emph{strictly}  convex in the input distribution $P_{\sfX}$. 
  \end{itemize}
  Consequently, if $\alpha \in (1,\infty)$, it follows that $I_\alpha(\sfX;{\rm snr})$ is strictly concave in $P_{\sfX}$ on $\cD_\alpha$.
\end{prop}

\begin{proof}
 The proof is provided in Appendix~\ref{app:ConcavityConvexity}.
 \end{proof}
\begin{rem}
In~\cite[Theorem~11]{HoISIt2015}, the authors proved that $\iota_{\alpha}(\sfX; {\rm snr})= \frac{1}{\alpha-1}\exp \left( \frac{\alpha-1}{\alpha} I_\alpha(\sfX;{\rm snr}) \right )$ is concave in $P_{\sfX}$ for all $\alpha \in (0,1) \cup (1,\infty)$. Note that Proposition~\ref{prop:ConcConv} implies that $\iota_{\alpha}(\sfX; {\rm snr})$ is, in fact, {\em strictly} concave in $P_{\sfX}$. To see this, first note that  $\iota_{\alpha}(\sfX; {\rm snr}) = \frac{1}{\alpha-1}\left( \zeta_{\alpha}(\sfX; {\rm snr}) \right )^{\frac{1}{\alpha}}$, and then: (1) if $\alpha \in (1,+\infty)$, we have that $\zeta_{\alpha}(\sfX; {\rm snr})$ is strictly concave in $P_{\sfX}$ on $\cD_\alpha$ and hence, $\left( \zeta_{\alpha}(\sfX; {\rm snr}) \right )^{\frac{1}{\alpha}}$ is strictly concave in $P_{\sfX}$ on $\cD_\alpha$; and (2) if $\alpha \in (0,1)$, we have that $\zeta_{\alpha}(\sfX; {\rm snr})$ is strictly convex in $P_{\sfX}$ and hence, $\left( \zeta_{\alpha}(\sfX; {\rm snr}) \right )^{\frac{1}{\alpha}}$ is strictly convex in $P_{\sfX}$ implying that $\frac{1}{\alpha-1}\left( \zeta_{\alpha}(\sfX; {\rm snr}) \right )^{\frac{1}{\alpha}}$ is strictly concave in $P_{\sfX}$.
\end{rem}

  Equipped with the results in Theorem~\ref{thm:continuity_distribution} and Proposition~\ref{prop:ConcConv}, and using standard convex optimization results (see~\cite{dytso2018discrete} for example) we have the following lemma.
  
\begin{lemma}
Fix $\alpha\in(0,1)\cup(1,\infty)$ and ${\rm snr}>0$.
Let $\cP$ be a nonempty, weakly compact, and convex set of probability
distributions. Suppose that either
\begin{enumerate}
    \item $\alpha\in(0,1)$; or
    \item $\alpha>1$ and there exists some $k>\alpha$ such that
    \begin{equation}
        \sup_{P_{\sf X}\in\cP}
        \int_{\mathbb R}|x|^k\,\rmd P_{\sf X}(x)
        <\infty.
    \end{equation}
\end{enumerate}
Then, the mapping
\begin{equation}
    P_{\sf X}\mapsto I_\alpha(\sf X;{\rm snr})
\end{equation}
is finite and continuous on $\cP$ with respect to weak convergence.
Consequently,
\begin{equation}
    \max_{P_{\sf X}\in\cP}
    I_\alpha(\sf X;{\rm snr})
    <\infty,
    \label{eq:alpha_mul_optimization}
\end{equation}
and the maximizer in \eqref{eq:alpha_mul_optimization} exists and is
unique.
\end{lemma}

\section{$\alpha$-I-MMSE and Low-SNR Behavior of $\alpha$-Mutual Information}
\label{sec:alphaIMMSE}
In this section, we establish a relationship between $\alpha$-mutual information and MMSE, which we refer to as the $\alpha$-I-MMSE. This generalizes the celebrated one in~\cite{GuoIT2005} for $\alpha=1$.

Recall that the MMSE of estimating a random variable $\sfU$ from another random variable $\sfW$ is given by
\begin{equation}
\label{eq:MMSE}
    {\rm mmse}(\sfU | \sfW) =\inf_{f} \bbE \left[ (\sfU - f(\sfW) )^2 \right] = \bbE \left[ ( \sfU - \bbE[\sfU|\sfW] )^2 \right].
\end{equation}
In this section, following the notation in~\cite{GuoIT2005}, we will study the following quantity: 
\begin{equation}
  {\rm mmse}_\alpha(\sfX ; {\rm{snr}})={\rm mmse}(\sfX_\alpha | \sfY_\alpha),\label{eq:mmse_alpha}
\end{equation}
where, as before, $(\sfX_\alpha , \sfY_\alpha)$ is distributed according to \eqref{eq:joint_expression}.  Note that~\eqref{eq:mmse_alpha} is well defined only if $f_{\sfY_\alpha}$ exists (see Theorem~\ref{thm:FinitenessAlphaMI}). Depending on the context, we will use both ${\rm mmse}_\alpha(\sfX ; {\rm{snr}}) $ and ${\rm mmse}(\sfX_\alpha | \sfY_\alpha)$.

\subsection{Generalization of Brown's Identity}
\label{sec:brown}
We show a generalization of Brown's identity~\cite[eq.(58)]{GuoIT2005}. Toward this end, we will need the notions of score function and Fisher information, which are next formally defined.
\begin{definition}[Score function and Fisher Information]
\label{def:ScoreFunction}
For a distribution $P_\sfY$ supported on $\cY$ with a PDF $f_\sfY(\cdot)$, the score function $\rho_\sfY: \cY \mapsto \bbR$ is defined as
\begin{equation}
\rho_\sfY(y) = \frac{\rmd}{\rmd \, y} \log f_\sfY(y),
\end{equation}
and Fisher information is defined as 
\begin{equation}
J(\sfY) = \bbE \left [\rho_\sfY^2(\sfY) \right ].
\end{equation}
\end{definition}
\begin{definition}[R\'enyi-Fisher information of order $\alpha$]
\label{def:Renyi-Fisher}
    Let $\sfX$ be a real-valued random variable with smooth density $f_{\sfX}$. The R\'enyi-Fisher information of order $\alpha \in (0,1) \cup (1,\infty)$ is defined as
\begin{equation}
J_\alpha(\sfX)
=
\alpha
\frac{
\int_{\bbR} |f'_{\sfX}(x)|^2 f_{\sfX}^{\alpha-2}(x)\,\rmd x
}{
\int_{\bbR} f_{\sfX}^{\alpha}(x)\,\rmd x
},
\end{equation}
whenever the expression is finite. For $\alpha=1$, this recovers Fisher information $J(\sfX)$ in Definition~\ref{def:ScoreFunction}.
\end{definition}
We now show that the score function for $f_{\sfY_\alpha}$ can be related to the score function for $f_\sfY$ but at a different ${\rm snr}$. 
 \begin{theorem}
 \label{theorem:GeneralizationBrown}
 Assume that $f_{\sfY_\alpha}$ exists. Then, the score function for $f_{\sfY_\alpha}(\cdot)$ in~\eqref{eq:PDFYalpha} is given by 
     \begin{align}
     \rho_{\sfY_\alpha}(y) =  \frac{1}{\alpha} \rho_{\sfY; \alpha \, {\rm snr}}(y) , \, y \in \bbR, \label{eq:score_function_identity}
     \end{align}
     where $\rho_{\sfY; \alpha \, {\rm snr}}$ is the score function for $f_\sfY(\cdot; \alpha \; {\rm snr})$ in~\eqref{eq:fYsnr}.
    Moreover,  the Fisher information for $f_{\sfY_\alpha}(\cdot)$ in~\eqref{eq:PDFYalpha} is given by
    \begin{equation}
    \label{eq:FisherInfo}
        J(\sfY_\alpha) = {\rm snr} -  \alpha \; {\rm snr}^2 \; {\rm mmse}(\sfX_\alpha | \sfY_\alpha).   
    \end{equation}
 \end{theorem}
 \begin{IEEEproof}
 To see the expression for the score function, note that in view of~\eqref{eq:PDFYalpha}, we have that 
 \begin{equation}
      \rho_{\sfY_\alpha}(y)  = \frac{\rmd}{\rmd y} \log f_{\sfY_\alpha}(y) = \frac{1}{\alpha } \frac{\rmd}{\rmd y} \log f_{\sfY}(y; \alpha \; {\rm snr})  =  \frac{1}{\alpha } \rho_{\sfY; \alpha \, {\rm snr}}(y).
 \end{equation}
 For the proof of Fisher information, we will also need to leverage the following identity~\cite{robbins1956empirical,hatsell1971some,dytso2022conditional}: 
  \begin{equation}
  \label{eq:VarianceAndDerivativeScore} 
  {\rm{Var}}(\sfX | \sfY =y ; {\rm snr})=   \frac{1}{\rm snr} + \frac{1}{\rm snr^2}  \rho_{\sfY;  {\rm snr}}'(y). 
 \end{equation}
Additionally, we will need to show that $\rho_{\sfY; \alpha \, {\rm snr}}(y) \, f_{\sfY_\alpha}(y)$ goes to zero as $y \to \pm \infty$. To see this, first note that 
\begin{align}
     \rho_{\sfY; \alpha \, {\rm snr}}(y) \, f_{\sfY_\alpha}(y)  
     & = \frac{f_{\sfY}'(y; \alpha \; {\rm snr})}{f_{\sfY}(y; \alpha \; {\rm snr})} \, \frac{f_\sfY^{ \frac{1}{\alpha}}(y; \alpha \; {\rm snr}) }{ \int_{\bbR} f_\sfY^{ \frac{1}{\alpha}}(t; \alpha \; {\rm snr}) \; \rmd t} \\ 
     &= C \, f_{\sfY}'(y; \alpha \; {\rm snr}) \; f_{\sfY}^{ \frac{1}{\alpha}-1 } (y; \alpha \; {\rm snr}) , \label{eq:expression_for_product_score_and_f}
\end{align}
where $C^{-1} = \int_{\bbR} f_{\sfY}^{ \frac{1}{\alpha} } (t; \alpha \; {\rm snr}) \; \rmd t$. Now, for some fixed $p \in (0,1)$, let
\begin{equation}
M_{p,\alpha, {\rm snr}}=  \max_{u \in \bbR}    {\alpha \; \rm snr} \;  |u|   \; \phi^{1-p}\left(\sqrt{ \alpha \; \rm{snr}} \; u\right) .  \label{eq:definition_of_M}
\end{equation}
 Next, observe that 
\begin{align}
    | f_{\sfY}'(y; \alpha \; {\rm snr})|& =  \left| \bbE \left[  {{\left( \alpha \; {\rm{snr}}\right)^{3/2}}} \;  (y-\sfX) \; \phi\left(\sqrt{ \alpha \; \rm{snr}} \; (y-\sfX)\right)   \right] \right|  \label{eq:differentation_under_gaussin} \\
    &= \left| \bbE \left[ {{\left( \alpha \; {\rm{snr}}\right)^{3/2}}} \;  (y-\sfX) \; \phi^{1-p}\left(\sqrt{ \alpha \; \rm{snr}} \; (y-\sfX)\right) \; \phi^{p}\left(\sqrt{ \alpha \; \rm{snr}} \; (y-\sfX)\right)   \right] \right|  \\
    & {\leq   \bbE \left[ {{\left( \alpha \; {\rm{snr}}\right)^{3/2}}} \;  \left| y-\sfX \right | \; \phi^{1-p}\left(\sqrt{ \alpha \; \rm{snr}} \;  (y-\sfX) \right) \; \phi^{p}\left(\sqrt{ \alpha \; \rm{snr}} \; (y-\sfX)\right)   \right] } \label{eq:TriangleIne} \\
    &\le M_{p,\alpha, {\rm snr}}  \;  \sqrt{ \alpha \; \rm snr} \; \bbE \left[  \phi^{p}\left(\sqrt{ \alpha \; \rm{snr}} \; (y-\sfX)\right)   \right]  \label{eq:using_Def_M}\\
    &\le M_{p,\alpha, {\rm snr}} \;  \sqrt{ \alpha \; \rm snr}  \; \bbE^{p} \left[  \phi\left(\sqrt{ \alpha \; \rm{snr}} \; (y-\sfX)\right)   \right] \label{eq:using_jensens_proof_of_bronw} \\
    &= C_{p, \alpha, {\rm snr}}  \; {f_{\sfY}^{ p }} (y; \alpha \; {\rm snr}), \label{eq:collect_constants}
\end{align}
where~\eqref{eq:differentation_under_gaussin} follows from standard results on Gaussian convolutions~\cite[Prop.~8.10]{folland1999real}{;~\eqref{eq:TriangleIne} is due to the triangle inequality;}~\eqref{eq:using_Def_M} follows from the definition in~\eqref{eq:definition_of_M};~\eqref{eq:using_jensens_proof_of_bronw} follows from Jensen's inequality since $u \to u^p$ is concave {for $p \in (0,1)$;} and~\eqref{eq:collect_constants} follows from collecting all the constants inside $C_{p, \alpha, {\rm snr}}$. 

Combining~\eqref{eq:expression_for_product_score_and_f} and~\eqref{eq:collect_constants}, we have that for every $p \in (0,1)$
\begin{equation}
    \left| \rho_{\sfY; \alpha \, {\rm snr}}(y) \, f_{\sfY_\alpha}(y) \right| \le  {C} \; C_{p, \alpha, {\rm snr}} \; f_{\sfY}^{ \frac{1}{\alpha}-1 +p} (y; \alpha \; {\rm snr}),
\end{equation}
which implies that 
\begin{equation}
\lim_{y\to\pm\infty}
\left|
\rho_{\sfY;\alpha\,{\rm snr}}(y)
f_{\sfY_\alpha}(y)
\right|
=
0,
\label{eq:limit_of_product_score_and_f_y}
\end{equation}
since, {for $\alpha>0$,} we can always choose a $p \in (0,1)$ such that $\frac{1}{\alpha}-1 +p >0$  and since $ \lim_{y \to \pm \infty} f_{\sfY} (y; \alpha \; {\rm snr}) =0$.

Now, the proof of Fisher information goes as follows,
     \begin{align}
         J(\sfY_\alpha) &=\int_{\bbR} \rho_{\sfY_\alpha}^2(y) \; f_{\sfY_\alpha}(y) \; \rmd y\\
         &=  \frac{1}{\alpha} \int_{\bbR} \rho_{\sfY_\alpha}(y) \; \rho_{\sfY; \alpha \, {\rm snr}}(y) \; f_{\sfY_\alpha}(y) \; \rmd y \label{eq:ScoreIdentityFirstStep}\\
         & =  \frac{1}{\alpha} \int_{\bbR}  \rho_{\sfY; \alpha \, {\rm snr}}(y) \; f_{\sfY_\alpha}'(y) \; \rmd y \\
         & = \frac{1}{\alpha} \lim_{R \to \infty} \int_{-R}^R \rho_{\sfY; \alpha \, {\rm snr}}(y) \; f_{\sfY_\alpha}'(y) \; \rmd y \\
         & = \frac{1}{\alpha} [\rho_{\sfY; \alpha \, {\rm snr}}(y) \; f_{\sfY_\alpha}(y)]_{-\infty}^{+ \infty} - \frac{1}{\alpha} \lim_{R \to \infty} \int_{-R}^R \rho_{\sfY; \alpha \, {\rm snr}}'(y) \; f_{\sfY_\alpha}(y) \ \rmd y \label{eq:IBP_step}\\
         &=- \frac{1}{\alpha} \lim_{R \to \infty} \int_{-R}^R  \rho_{\sfY; \alpha \, {\rm snr}}'(y) \; f_{\sfY_\alpha}(y) \ \rmd y  \label{eq:using_score_funct_idenity}\\
           &=- \frac{1}{\alpha} \lim_{R \to \infty} \int_{-R}^R ( \alpha^2 \; {\rm snr}^2 \; {\rm{Var}}(\sfX | \sfY=y ; \alpha \; {\rm snr}) - \alpha \; {\rm snr}) \; f_{\sfY_\alpha}(y) \; \rmd y \label{eq:using_nolte_identity}\\
           &=-  \alpha \; {\rm snr}^2 \lim_{R \to \infty} \int_{-R}^R    {\rm{Var}}(\sfX | \sfY=y ; \alpha \; {\rm snr})  \; f_{\sfY_\alpha}(y) \; \rmd y +  {\rm snr} \lim_{R \to \infty} \int_{-R}^R f_{\sfY_\alpha}(y) \; \rmd y \\
           &=-  \alpha \; {\rm snr}^2 \lim_{R \to \infty} \int_{-R}^R    {\rm{Var}}(\sfX_\alpha | \sfY_\alpha=y)  \; f_{\sfY_\alpha}(y) \; \rmd y + {\rm snr}  \label{eq:using_equality_cond_variance}\\
           &=-  \alpha \; {\rm snr}^2  \int_{\bbR}    {\rm{Var}}(\sfX_\alpha | \sfY_\alpha=y)  \; f_{\sfY_\alpha}(y) \; \rmd y + {\rm snr} \label{eq:MCTStepJ}\\
           &= -  \alpha \; {\rm snr}^2 \;  \bbE \left[  {\rm{Var}}(\sfX_\alpha | \sfY_\alpha) \right]   + {\rm snr}\\
           &= -  \alpha \; {\rm snr}^2 \;  {\rm mmse}(\sfX_\alpha | \sfY_\alpha)   + {\rm snr} , \label{eq:using_def_of_mmse}
     \end{align}
     where:~\eqref{eq:ScoreIdentityFirstStep} follows from~\eqref{eq:score_function_identity};
\eqref{eq:using_score_funct_idenity} follows from~\eqref{eq:limit_of_product_score_and_f_y};~\eqref{eq:using_nolte_identity} is due to~\eqref{eq:VarianceAndDerivativeScore};~\eqref{eq:using_equality_cond_variance} follows from the definition in~\eqref{eq:conditional_expression} that $\rmd P_{\sfX_\alpha|\sfY_\alpha=y}(x) = \rmd P_{\sfX|\sfY=y ; \alpha {\rm snr}} (x)$, which implies that the conditional variances are also equal; \eqref{eq:MCTStepJ} follows from the monotone convergence theorem; and~\eqref{eq:using_def_of_mmse} follows from the definition of MMSE. This concludes the proof of Theorem~\ref{theorem:GeneralizationBrown}.  
 \end{IEEEproof}
The identity in~\eqref{eq:FisherInfo} is a generalization of Brown's identity~\cite{brown1971admissible} that has several important applications in statistics; see for example~\cite{brown1986fundamentals}. In our context,~\eqref{eq:FisherInfo} will be one of the ingredients for the proof of the generalized I-MMSE relationship in Theorem~\ref{thm:IMMSE}.
Since Fisher information is a non-negative quantity, an immediate consequence of Theorem~\ref{theorem:GeneralizationBrown} is given by the next corollary.
\begin{corollary}
\label{cor:UBMMSE}
Assume that $f_{\sfY_\alpha}$ exists. Fix some $\alpha \in (0,1) \cup (1,\infty)$ and ${\rm snr}>0$. Then,
\begin{equation}
    \label{eq:UBMMSE}
    {\rm mmse}(\sfX_\alpha | \sfY_\alpha) \leq \frac{1}{\alpha \; {\rm snr}}. 
    \end{equation}
\end{corollary}

\subsection{$\alpha$-I-MMSE Relationship}
\label{sec:gen_I-MMSE}
We here present an expression that connects $\alpha$-mutual information and MMSE.  In particular, this result shows that the rate of $\alpha$-mutual information increase as SNR increases is equal to a fraction $\alpha/2$ of the MMSE achieved by the optimal estimator of $\sfX_\alpha$ given $\sfY_\alpha$.
The main result, which generalizes the one in~\cite{GuoIT2005} for $\alpha=1$, is presented in the next theorem.
\begin{theorem}
\label{thm:IMMSE}
Let $\alpha \in (0,1) \cup (1,\infty)$ and ${\rm{snr}}>0$. Assume that $f_{\sfY_\alpha}$ exists. Then, it holds that 
\begin{align}
    \frac{\mathrm{d}}{\mathrm{d} \,{\rm{snr}}} I_{\alpha}(\sfX;\mathrm{snr}) & =  \frac{\alpha}{2}{\rm mmse}(\sfX_\alpha | \sfY_\alpha). \label{eq:Deriativenewch3}
\end{align}
\end{theorem}
\begin{proof}
The proof is provided in Appendix~\ref{app:DerivativeMI}.
\end{proof}
Fig.~\ref{fig:IMMSE-comparison} and Fig.~\ref{fig:AlphaComparison} provide an illustration of the results in Lemma~\ref{lemma:GaussianAlphaMI} and Theorem~\ref{thm:IMMSE} for a Gaussian input $\sfX \sim \cN(0,1)$ and for a binary input equiprobable on $\{-1,1\}$. As seen in these figures, the mappings $\alpha \to I_{\alpha}(\sfX;\mathrm{snr})$ and  ${\rm{snr}}  \to I_{\alpha}(\sfX;\mathrm{snr})$ are increasing, while ${\rm{snr}}  \to {\rm mmse}_\alpha(\sfX ; {\rm{snr}})$ is decreasing, as expected. Fig.~\ref{fig:AlphaComparison} further suggests that $\alpha  \to {\rm mmse}_\alpha(\sfX ; {\rm{snr}})$ is decreasing. Proving (or disproving) this property, which is also supported by the bound in Corollary~\ref{cor:UBMMSE}, is beyond the scope of this work and is an interesting direction for future research.

From Theorem~\ref{thm:IMMSE} and the fundamental theorem of calculus, we have that 
\begin{equation}
\label{eq:FTC}
I_{\alpha}(\sfX;\mathrm{snr}_2) - I_{\alpha}(\sfX;\mathrm{snr}_1) =  \frac{\alpha}{2} \int_{\mathrm{snr}_1}^{\mathrm{snr}_2}  {\rm mmse}_\alpha(\sfX ; \gamma) \; \rmd \gamma.
\end{equation}
Furthermore, by using the continuity of $\alpha$-mutual information at $\rm{snr}=0^+$ (see Proposition~\ref{prop:LimitSNRToZeroAlphaMI}) and sending $\mathrm{snr}_1 \to 0$ in~\eqref{eq:FTC}, we have that
\begin{equation}
\label{eq:FTCLowSNR}
I_{\alpha}(\sfX;\mathrm{snr})  = \frac{\alpha}{2}\int_{0}^{\mathrm{snr}}  {\rm mmse}_\alpha(\sfX ;\gamma) \; \rmd \gamma.
\end{equation}

\begin{figure}[t]
    \centering
    \begin{subfigure}[t]{0.49\textwidth}
        \centering
        \begin{tikzpicture}[scale=0.85, transform shape]
\begin{axis}[
    width=9cm,
    height=8cm,
    grid=major,
    xlabel={${\rm{snr}}$},
    xmin=0,
    xmax=15,
    ymin=0,
    ymax=2.5,
    legend style={
        at={(0.015,0.99)},
        anchor=north west,
        font=\fontsize{10}{13}\selectfont,
        inner sep=2pt,
        row sep=1pt,
        fill=white,
        fill opacity=0.85,
        text opacity=1,
        draw=none
    },
    legend columns=3,
    every axis plot/.append style={line width=1.5pt},
    mark repeat=100,
    mark phase=10,
    mark size=2.6pt
]

\addplot[
    red,
    solid,
    line width=1pt,
    mark=star,
    mark size=3.3pt,
    mark options={solid}
] table[
    x=snr,
    y=I05,
    col sep=comma
] {gaussian_data.csv};
\addlegendentry{$I_{\alpha}$}

\addplot[
    red,
    dashed,
    mark=star,
    mark size=3.3pt,
    mark options={solid}
] table[
    x=snr,
    y=mmse05,
    col sep=comma
] {gaussian_data.csv};
\addlegendentry{$\alpha$-mmse}

\addlegendimage{empty legend}
\addlegendentry{$\alpha=0.5$}

\addplot[
    blue,
    solid,
    line width=1pt,
    mark=o,
    mark options={solid}
] table[
    x=snr,
    y=I1,
    col sep=comma
] {gaussian_data.csv};
\addlegendentry{$I_{\alpha}$}

\addplot[
    blue,
    dashed,
    mark=o,
    mark options={solid}
] table[
    x=snr,
    y=mmse1,
    col sep=comma
] {gaussian_data.csv};
\addlegendentry{$\alpha$-mmse}

\addlegendimage{empty legend}
\addlegendentry{$\alpha=1.0$}

\addplot[
    green!60!black,
    solid,
    line width=1pt,
    mark=triangle*,
    mark options={solid}
] table[
    x=snr,
    y=I2,
    col sep=comma
] {gaussian_data.csv};
\addlegendentry{$I_{\alpha}$}

\addplot[
    green!60!black,
    dashed,
    mark=triangle*,
    mark options={solid}
] table[
    x=snr,
    y=mmse2,
    col sep=comma
] {gaussian_data.csv};
\addlegendentry{$\alpha$-mmse}

\addlegendimage{empty legend}
\addlegendentry{$\alpha=2.0$}

\end{axis}
\end{tikzpicture}
        \caption{Gaussian input $\sfX\sim\mathcal N(0,1)$.}
        \label{fig:IMMSEG}
    \end{subfigure}
    \hfill
    \begin{subfigure}[t]{0.49\textwidth}
        \centering
        \begin{tikzpicture}[scale=0.85, transform shape]
\begin{axis}[
    width=10cm,
    height=8cm,
    grid=major,
    xlabel={${\rm{snr}}$},
    xmin=0,
    xmax=15,
    ymin=0,
    ymax=1.05,
    legend style={
        at={(0.42,0.45)},
        anchor=north west,
        font=\fontsize{10}{13}\selectfont,
        inner sep=2pt,
        row sep=1pt,
        fill=white,
        fill opacity=0.85,
        text opacity=1,
        draw=none
    },
    legend columns=3,
    every axis plot/.append style={line width=1.5pt},
    mark repeat=100,
    mark phase=10,
    mark size=2.6pt
]

\addplot[
    red,
    solid,
    line width=1pt,
    mark=star,
    mark size=3.3pt,
    mark options={solid}
] table[
    x=snr,
    y=I05,
    col sep=comma
] {plot_data.csv};
\addlegendentry{$I_{\alpha}$}

\addplot[
    red,
    dashed,
    mark=star,
    mark size=3.3pt,
    mark options={solid}
] table[
    x=snr,
    y=mmse05,
    col sep=comma
] {plot_data.csv};
\addlegendentry{$\alpha$-mmse}

\addlegendimage{empty legend}
\addlegendentry{$\alpha=0.5$}

\addplot[
    blue,
    solid,
    line width=1pt,
    mark=o,
    mark options={solid}
] table[
    x=snr,
    y=I1,
    col sep=comma
] {plot_data.csv};
\addlegendentry{$I_{\alpha}$}

\addplot[
    blue,
    dashed,
    mark=o,
    mark options={solid}
] table[
    x=snr,
    y=mmse1,
    col sep=comma
] {plot_data.csv};
\addlegendentry{$\alpha$-mmse}

\addlegendimage{empty legend}
\addlegendentry{$\alpha=1$}

\addplot[
    green!60!black,
    solid,
    line width=1pt,
    mark=triangle*,
    mark options={solid}
] table[
    x=snr,
    y=I2,
    col sep=comma
] {plot_data.csv};
\addlegendentry{$I_{\alpha}$}

\addplot[
    green!60!black,
    dashed,
    mark=triangle*,
    mark options={solid}
] table[
    x=snr,
    y=mmse2,
    col sep=comma
] {plot_data.csv};
\addlegendentry{$\alpha$-mmse}

\addlegendimage{empty legend}
\addlegendentry{$\alpha=2$}

\end{axis}
\end{tikzpicture}
        \caption{Binary input $\sfX\in\{-1,1\}$ equiprobable.}
        \label{fig:IMMSERB}
    \end{subfigure}
    \caption{$I_{\alpha}(\sfX;\mathrm{snr})$ and the corresponding $\alpha$-MMSE ${\rm mmse}_\alpha(\sfX ; {\rm{snr}})$ versus ${\rm{snr}}$ and different values of $\alpha$.}
    \label{fig:IMMSE-comparison}
\end{figure}

\begin{figure}[t]
\centering
\begin{subfigure}[t]{0.49\textwidth}
\centering
\begin{tikzpicture}[scale=0.85, transform shape]
\begin{axis}[
    width=9cm,
    height=8cm,
    grid=major,
    xlabel={$\alpha$},
    xmin=0,
    xmax=12,
    ymin=0,
    ymax=3.5,
legend style={
        at={(0.015,0.99)},
        anchor=north west,
        font=\fontsize{10}{13}\selectfont,
        inner sep=2pt,
        row sep=1pt,
        fill=white,
        fill opacity=0.85,
        text opacity=1,
        draw=none
    },
    legend columns=3,
    every axis plot/.append style={line width=1.5pt},
    mark repeat=100,
    mark phase=10,
    mark size=2.6pt
]

\addplot[
    red,
    solid,
    line width=1pt,
    mark=star,
    mark size=3.3pt,
    mark options={solid}
] table[x=alpha,y=Ialpha_snr0p5,col sep=comma]
{Gaussian_fixed_differentsnr.csv};
\addlegendentry{$I_{\alpha}$}

\addplot[
    red,
    dashed,
    mark=star,
    mark size=3.3pt,
    mark options={solid}
] table[x=alpha,y=mmse_snr0p5,col sep=comma]
{Gaussian_fixed_differentsnr.csv};
\addlegendentry{$\alpha$-mmse}

\addlegendimage{empty legend}
\addlegendentry{$\mathrm{snr}=0.5$}

\addplot[
    blue,
    solid,
    line width=1pt,
    mark=o,
    mark options={solid}
] table[x=alpha,y=Ialpha_snr2,col sep=comma]
{Gaussian_fixed_differentsnr.csv};
\addlegendentry{$I_{\alpha}$}

\addplot[
    blue,
    dashed,
    mark=o,
    mark options={solid}
] table[x=alpha,y=mmse_snr2,col sep=comma]
{Gaussian_fixed_differentsnr.csv};
\addlegendentry{$\alpha$-mmse}

\addlegendimage{empty legend}
\addlegendentry{$\mathrm{snr}=2$}

\addplot[
    green!60!black,
    solid,
    line width=1pt,
    mark=triangle*,
    mark options={solid}
] table[x=alpha,y=Ialpha_snr10,col sep=comma]
{Gaussian_fixed_differentsnr.csv};
\addlegendentry{$I_{\alpha}$}

\addplot[
    green!60!black,
    dashed,
    mark=triangle*,
    mark options={solid}
] table[x=alpha,y=mmse_snr10,col sep=comma]
{Gaussian_fixed_differentsnr.csv};
\addlegendentry{$\alpha$-mmse}

\addlegendimage{empty legend}
\addlegendentry{$\mathrm{snr}=10$}

\end{axis}
\end{tikzpicture}
  \vspace{-0.9mm}
 \caption{Gaussian input $\sfX\sim\mathcal N(0,1)$.}
\label{fig:VersusAlphaGaussian}
\end{subfigure}
\hfill
\begin{subfigure}[t]{0.49\textwidth}
\centering
\begin{tikzpicture}[scale=0.85, transform shape]
\begin{axis}[
    width=10cm,
    height=8cm,
    grid=major,
    xlabel={$\alpha$},
    xmin=0,
    xmax=12,
    ymin=0,
    ymax=1.02,
    legend style={
        at={(0.42,0.45)},
        anchor=north west,
        font=\fontsize{10}{13}\selectfont,
        inner sep=2pt,
        row sep=1pt,
        fill=white,
        fill opacity=0.85,
        text opacity=1,
        draw=none
    },
    legend columns=3,
    every axis plot/.append style={line width=1.5pt},
    mark repeat=100,
    mark phase=10,
    mark size=2.6pt
]

\addplot[
    red,
    solid,
    line width=1pt,
    mark=star,
    mark size=3.3pt,
    mark options={solid}
] table[
    x=alpha,
    y=Ialpha_snr0p5,
    col sep=comma
] {Binary_fixed_differentsnr.csv};
\addlegendentry{$I_{\alpha}$}

\addplot[
    red,
    dashed,
    mark=star,
    mark size=3.3pt,
    mark options={solid}
] table[
    x=alpha,
    y=mmse_snr0p5,
    col sep=comma
] {Binary_fixed_differentsnr.csv};
\addlegendentry{$\alpha$-mmse}

\addlegendimage{empty legend}
\addlegendentry{$\mathrm{snr}=0.5$}

\addplot[
    blue,
    solid,
    line width=1pt,
    mark=o,
    mark options={solid}
] table[
    x=alpha,
    y=Ialpha_snr2,
    col sep=comma
] {Binary_fixed_differentsnr.csv};
\addlegendentry{$I_{\alpha}$}

\addplot[
    blue,
    dashed,
    mark=o,
    mark options={solid}
] table[
    x=alpha,
    y=mmse_snr2,
    col sep=comma
] {Binary_fixed_differentsnr.csv};
\addlegendentry{$\alpha$-mmse}

\addlegendimage{empty legend}
\addlegendentry{$\mathrm{snr}=2$}

\addplot[
    green!60!black,
    solid,
    line width=1pt,
    mark=triangle*,
    mark options={solid}
] table[
    x=alpha,
    y=Ialpha_snr10,
    col sep=comma
] {Binary_fixed_differentsnr.csv};
\addlegendentry{$I_{\alpha}$}

\addplot[
    green!60!black,
    dashed,
    mark=triangle*,
    mark options={solid}
] table[
    x=alpha,
    y=mmse_snr10,
    col sep=comma
] {Binary_fixed_differentsnr.csv};
\addlegendentry{$\alpha$-mmse}

\addlegendimage{empty legend}
\addlegendentry{$\mathrm{snr}=10$}

\end{axis}
\end{tikzpicture}
\vspace{-7mm}
\caption{Binary input $\sfX\in\{-1,1\}$ equiprobable.}
\label{fig:VersusAlphaBinary}
\end{subfigure}
\caption{$I_{\alpha}(\sfX;\mathrm{snr})$ and the corresponding $\alpha$-MMSE ${\rm mmse}_\alpha(\sfX ; {\rm{snr}})$ versus $\alpha$ and different values of ${\rm{snr}}$.}
\label{fig:AlphaComparison}
\end{figure}

\subsection{On Generalized de Bruijn's Identity}
\label{sec:DeBruijnIdentity}
The classical I-MMSE relationship is known to be equivalent to de Bruijn's identity~\cite{stam1959some}, which relates the derivative of Shannon differential entropy to Fisher information. This connection is typically established via Brown's identity. 

Equipped with our generalization of the I-MMSE relationship in~\eqref{eq:Deriativenewch3}, the representation in~\eqref{eq:alpha_mi_diff_entropy_Rep_Ver1}, and the generalized Brown identity in~\eqref{eq:FisherInfo}, we obtain the following generalization of de Bruijn's identity:
\begin{equation}
    \frac{\rmd}{\rmd \, {\rm snr}} 
    h_{\frac{1}{\alpha}}
    \left(
        \sfX + \frac{1}{\sqrt{\alpha\,{\rm snr}}}\sfZ
    \right)
    =
    -\frac{J(\sfY_\alpha)}{2\,{{\rm snr^2}}}.
    \label{eq:first_de_brujin_identit}
\end{equation}
We note that generalized de Bruijn's identities have been studied previously. In particular, the authors in~\cite{wu2025entropic} introduced the generalization of Fisher information in Definition~\ref{def:Renyi-Fisher},
and established the following generalized de Bruijn's identity:
\begin{equation}
    \frac{\rmd}{\rmd \, {\rm snr}}
    h_{\frac{1}{\alpha}}
    \left(
        \sfX + \frac{1}{\sqrt{\alpha\,{\rm snr}}}\sfZ
    \right)
    =
    -\frac{1}{2 \,\alpha\,{\rm snr}^2}
    J_{\frac{1}{\alpha}}
    \left(
        \sfX + \frac{1}{\sqrt{\alpha\,{\rm snr}}}\sfZ
    \right).
    \label{eq:second_version_de_brujin}
\end{equation}
The equivalence between the de Bruijn's identity in~\eqref{eq:first_de_brujin_identit} and the formulation involving R\'enyi-Fisher information in~\eqref{eq:second_version_de_brujin} follows directly from the score-function identity in~\eqref{eq:score_function_identity}. In particular, this identity shows that the classical Fisher information of $\sfY_\alpha$ can be related to the R\'enyi-Fisher information of
$
\sfY = \sfX + \frac{1}{\sqrt{\alpha\,{\rm snr}}}\sfZ
$
as
\begin{equation}
  \alpha \, J(\sfY_\alpha)
    =
    J_{\frac{1}{\alpha}}
    \left(
        \sfX + \frac{1}{\sqrt{\alpha\,{\rm snr}}}\sfZ
    \right).
\end{equation}

\subsection{Continuity of MMSE at ${\rm snr} =0^+$}
\label{sec:continutity_MMSE_zero_snr}
Continuity of MMSE at $\mathrm{snr}=0^+$ is a subtle but essential property. 
Notably, even in the classical I--MMSE framework for $\alpha=1$, the behavior at 
$\mathrm{snr}=0^+$ is not immediate and requires a careful analysis, which was 
subsequently established in~\cite{wu2011derivative}. We now show continuity of ${\rm mmse}(\sfX_\alpha | \sfY_\alpha)$ at ${\rm snr} =0^+$.
\begin{prop}\label{prop:continutity_mmse_zero_snr}
Let $\alpha \in (0,1) \cup (1,\infty)$ and assume that
\begin{equation}
\label{eq:MomentConstrContinuityMMSE}
   \left \{  \begin{array}{cc}
  \bbE[\sfX^2] <\infty , & \alpha \in (0,1),\\ 
    \bbE[|\sfX|^{4\alpha} ] <\infty , & \alpha \in (1,\infty). \\ 
    \end{array} \right.
\end{equation}
Then, it holds that
\begin{equation}
\lim_{{\rm{snr}} \to 0 } {\rm mmse}(\sfX_\alpha | \sfY_\alpha) = {\rm{Var}}(\sfX).
\end{equation}
\end{prop}
\begin{IEEEproof}
We start by noting that, under the assumption in~\eqref{eq:MomentConstrContinuityMMSE}, Proposition~\ref{prop:moment_conditions} ensures that $f_{\sfY_\alpha}$ always exists.
We now prove that, for each fixed $y\in\mathbb{R}$, it holds that
\begin{align}
\lim_{{\rm{snr}} \to 0 } {\rm{Var}}(\sfX_\alpha | \sfY_{\alpha}=y) ={\rm{Var}}(\sfX).
\label{eq:SNRToZeroVariance}
\end{align}
From~\eqref{eq:conditional_expression}, we have that $\rmd P_{\sfX_\alpha|\sfY_\alpha=y}(x) = \rmd P_{\sfX|\sfY=y ; \alpha \, {\rm snr}} (x)$, which implies 
\begin{align}
\bbE[\xi(\sfX_\alpha) | \sfY_{\alpha}=y]
= \frac{\bbE\left[\xi(\sfX) \; \phi(\sqrt{\alpha \; \rm{snr}}(y-\sfX))\right]}{\bbE\left[\phi(\sqrt{\alpha \; \rm{snr}}(y-\sfX))\right]},
\end{align}
and hence,
\begin{align}
{\rm{Var}}(\sfX_\alpha | \sfY_{\alpha}=y) = \frac{\bbE\left[ \sfX^2 \; \phi(\sqrt{\alpha \; \rm{snr}}(y-\sfX))\right]}{\bbE\left[\phi(\sqrt{\alpha \; \rm{snr}}(y-\sfX))\right]} - \left( \frac{\bbE\left[ \sfX \; \phi(\sqrt{\alpha \; \rm{snr}}(y-\sfX))\right]}{\bbE\left[\phi(\sqrt{\alpha \; \rm{snr}}(y-\sfX))\right]} \right )^2.
\end{align}
Now, note that
\begin{align}
\lim_{{\rm{snr}} \to 0 } \bbE\left[\phi(\sqrt{\alpha \; \rm{snr}}(y-\sfX))\right] = \frac{1}{\sqrt{2 \pi}},
\end{align}
and, for $\xi(\sfX) \in \{\sfX,\sfX^2 \}$, we have that
\begin{align}
\lim_{{\rm{snr}} \to 0 } \bbE\left[\xi(\sfX) \; \phi(\sqrt{\alpha \; \rm{snr}}(y-\sfX))\right] = \frac{1}{\sqrt{2 \pi}}  \bbE\left[\xi(\sfX)\right],
\end{align}
where we have used the dominated convergence theorem which can be applied since 
$\bbE[\sfX^2] < \infty$ by assumption (and hence, also $\bbE[\sfX] < \infty$). 
Thus, we arrive at~\eqref{eq:SNRToZeroVariance}.
Now, from the definition of MMSE, we have that
\begin{align}
\liminf_{{\rm{snr}} \to 0 }  {\rm mmse}(\sfX_\alpha | \sfY_\alpha) &= \liminf_{{\rm{snr}} \to 0 }  \bbE \left[  {\rm{Var}}(\sfX_\alpha | \sfY_\alpha) \right]
\\& \geq  \bbE \left[  \liminf_{{\rm{snr}} \to 0 } {\rm{Var}}(\sfX_\alpha | \sfY_\alpha) \right] \label{eq:FatouLemma}
\\& = {\rm{Var}}(\sfX), \label{eq:LBMMSESNTZero}
\end{align}
where~\eqref{eq:FatouLemma} follows by applying Fatou's lemma and~\eqref{eq:LBMMSESNTZero} is due to~\eqref{eq:SNRToZeroVariance}. 
By the law of total variance, we also have that ${\rm mmse}(\sfX_\alpha | \sfY_\alpha) \leq {\rm{Var}}(\sfX_\alpha)$ and hence, we arrive at
\begin{equation}
\limsup_{{\rm{snr}} \to 0 } \, {\rm mmse}(\sfX_\alpha | \sfY_\alpha)  \leq \limsup_{{\rm{snr}} \to 0 } \, {\rm{Var}}(\sfX_\alpha).
\end{equation}
In Appendix~\ref{app:VarXalphaToVarX}, we show that
\begin{equation}
\label{eq:VarXalphaToVarX}
\lim_{{\rm{snr}} \to 0 } {\rm{Var}}(\sfX_\alpha) = {\rm{Var}}(\sfX).
\end{equation}
This concludes the proof of Proposition~\ref{prop:continutity_mmse_zero_snr}.
\end{IEEEproof}

\subsection{Low-SNR Behavior of $\alpha$-Mutual Information}
\label{sec:LowSNRAlphaMI}
We now prove the low-SNR behavior of $\alpha$-mutual information. As for the case of $\alpha=1$ in~\cite{GuoIT2005} (and remarked in~\cite{VerduIT1990,LapidothIT2002,VerduIT2002}), also for arbitrary values of $\alpha>0$ the next proposition shows that the low-SNR $\alpha$-mutual information is `insensitive' to the input distribution, i.e., it depends on the input distribution only through its variance.
\begin{prop}
\label{prop:LowSNRAlphaMI}
Let $\alpha \in (0,1) \cup (1,\infty)$ and suppose that the assumptions of Proposition~\ref{prop:continutity_mmse_zero_snr} hold.  Then, it holds that
    \begin{align}
        \lim_{\rm snr \to 0^{+}}  \frac{1}{\rm snr }I_\alpha(\sfX; {\rm snr}) = \frac{\alpha}{2}{\rm Var}(\sfX). 
    \end{align}
\end{prop}
\begin{IEEEproof}
From the integral $\alpha$-I-MMSE relationship in~\eqref{eq:FTCLowSNR}, we have
\begin{equation}
    \frac{I_\alpha(\sf X;{\rm snr})}{{\rm snr}} = \frac{\alpha}{2 \, {\rm snr}} \int_0^{{\rm snr}}{\rm mmse}_\alpha(\sf X;\gamma)\,\rmd\gamma.
\end{equation}
By Proposition~\ref{prop:continutity_mmse_zero_snr},
\begin{equation}
    \lim_{\gamma  \to 0^+ } {\rm mmse}_\alpha(\sf X;\gamma)= {\rm Var}(\sf X).
\end{equation}
Therefore, by the continuous analogue of the Ces\'aro averaging theorem,
\begin{equation}
    \lim_{{\rm snr} \to 0^+ }
    \frac{1}{{\rm snr}}
    \int_0^{{\rm snr}}
    {\rm mmse}_\alpha(\sf X;\gamma)\,\rmd\gamma
    =
  {\rm Var}(\sf X).
\end{equation}
Consequently,
\begin{equation}
    \lim_{{\rm snr} \to 0^+  }
    \frac{I_\alpha(\sf X;{\rm snr})}{{\rm snr}} =\frac{\alpha}{2}\operatorname{Var}(\sf X).
\end{equation}
This concludes the proof of Proposition~\ref{prop:LowSNRAlphaMI}.
\end{IEEEproof}

\section{High-SNR Behavior of $\alpha$-Mutual Information}
\label{sec:High_SNR_Behaviour}
In this section, we study the high-SNR behavior of $\alpha$-mutual information. In the classical setting $\alpha = 1$, the high-SNR behavior of the mutual information has found important applications in analog compression~\cite{WuAnlogCompr}, in understanding interference alignment~\cite{wu2014information}, and in establishing fundamental entropic inequalities~\cite{GuoNow2012}.

We first characterize the behavior of $\alpha$-mutual information for discrete random variables. We then show how these high-SNR limits can be used to derive new representations of R\'enyi entropies. Finally, we establish a fundamental connection between the high-SNR behavior of $\alpha$-mutual information and R\'enyi information dimension. 

\subsection{High-SNR Behavior for Discrete Input Distributions}
\label{sec:DiscreteHighSNR}

The mutual information of a discrete random variable in the high-SNR regime converges to the entropy of the input~\cite{GuoIT2005}. We now show a similar result for $\alpha$-mutual information. 

\begin{prop}\label{prop:m-point}
Suppose that $\sfX$ is a discrete random variable. Then,
for every $\alpha\in(0,1)\cup(1,\infty)$, it holds that
\begin{equation}
\label{mptsmass}
\lim_{\mathrm{snr}\to\infty}
I_\alpha(\sfX;\mathrm{snr})
=
H_{\frac1\alpha}(\sfX),
\end{equation}
where the equality is understood in the extended real sense.
\end{prop}
\begin{IEEEproof}
The proof is provided in Appendix~\ref{app:ProofPropm-point}.
\end{IEEEproof}

\subsection{New Entropic Identities }
\label{sec:new_entropic_identities}
The original I-MMSE relationship led to a number of new representations of  entropy and differential entropy in terms of MMSE~\cite{GuoIT2005}. We now show similar results for the R\'enyi  entropy in Definition~\ref{def:ReyniEntropy} and R\'enyi  differential entropy in Definition~\ref{def:diff_renyi_entropy}. 
It would be interesting to investigate if the identities below could be used to show entropic inequalities for $\alpha$-measures as was done in   ~\cite{verdu2006simple,tulino2006monotonic,guo2006proof}.
\begin{prop} 
\label{prop:NewEntrDiscr}
Assume that $\sfX$ is a discrete random variable and that the assumptions of Proposition~\ref{prop:m-point} hold. Then, for any $\alpha \in (0,1) \cup (1,\infty)$, it holds that
    \begin{equation}
        H_{ \frac{1}{\alpha}}(\sfX) = \frac{\alpha}{2} \int_0^\infty {\rm mmse}_\alpha(\sfX, \gamma) \, \rmd \gamma . \label{eq:entropy_i-mmse}
    \end{equation}
\end{prop}
\begin{IEEEproof}
The proof follows by combining~\eqref{eq:FTCLowSNR} and~\eqref{mptsmass} from Proposition~\ref{prop:m-point}. 
\end{IEEEproof}
For differential entropy we have the following result. 
\begin{prop}
\label{prop:NewEntrCont}
Let $\alpha \in (0,1) \cup (1,\infty)$ and $\sfX$ be a continuous random variable such that 
\begin{equation}
   \lim_{t \to 0} h_\frac{1}{\alpha}(\sfX + t\sfZ) = h_\frac{1}{\alpha}(\sfX ). \label{eq:assumption_for_diff_entropy}
\end{equation}
Assume that $f_{\sfY_\alpha}$ exists and that $\bbE[\sfX^2] < \infty$. Then, it holds that
\begin{equation}
  h_\frac{1}{\alpha}(\sfX ) =  \frac{\alpha}{2} \int_0^{\infty} \left( {\rm mmse}_\alpha(\sfX; \gamma) -  \frac{1}{ 2 \pi \alpha^{ \frac{\alpha}{\alpha-1}  } + \alpha \gamma} \right ) \rmd \gamma.  
\end{equation}
\end{prop}
\begin{IEEEproof} Let $\beta = \alpha^{\frac{1}{\alpha-1}} 2 \pi$ and note that 
    \begin{align}
    h_{ \frac{1}{\alpha} } \left(\sfX  + \frac{1}{\sqrt{ \alpha \, \rm snr}} \sfZ \right) &= I_\alpha(\sfX;{\rm snr})+ h_\alpha(\sfZ) -\frac{1}{2} \log({\rm snr}) \label{eq:using_diff_ent_decomp}  \\
    &= \frac{\alpha}{2} \int_0^{\rm snr} {\rm mmse}_\alpha(\sfX; \gamma) \, \rmd \gamma +  \frac{1}{2 (\alpha-1)} \log (\alpha) + \frac{1}{2} \log(2\pi) -\frac{1}{2} \log({\rm snr}) \label{eq:using_diff_entropy_gauiss} \\
    &= \frac{\alpha}{2} \int_0^{\rm snr} {\rm mmse}_\alpha(\sfX; \gamma) \, \rmd \gamma  - \frac{1}{2} \log  \left( \frac{{\rm snr}}{\beta} \right) \\
    &= \frac{\alpha}{2} \left( \int_0^{\rm snr} \left( {\rm mmse}_\alpha(\sfX; \gamma) -  \frac{1}{\alpha}\frac{1}{\beta +\gamma} \right )\rmd \gamma   \right ) + \frac{1}{2}  \log \left(1 +\frac{\rm snr}{\beta}  \right) - \frac{1}{2} \log  \left( \frac{{\rm snr}}{\beta} \right),
    \end{align}
    where in~\eqref{eq:using_diff_ent_decomp} we have used~\eqref{eq:alpha_mi_diff_entropy_Rep_Ver1} and~\eqref{eq:using_diff_entropy_gauiss} follows from~\eqref{eq:FTCLowSNR} and~\eqref{eq:diff_renyi_entropy_Gaussian}.  Taking the limit as ${\rm snr} \to \infty$ and using the assumption in~\eqref{eq:assumption_for_diff_entropy} conclude the proof of Proposition~\ref{prop:NewEntrCont}.   
\end{IEEEproof}

\subsection{$\alpha$-Information Dimension}
\label{sec:alpah_info_dim}
We now characterize the $\alpha$-information dimension in Definition~\ref{def:AlphaInfDim} as a high-SNR limit of $\alpha$-mutual information. Toward this end, we will leverage the two following lemmas, which provide bounds for $\alpha$-mutual information of the additive uniform noise channel $\sfY_{\varepsilon}$ defined in~\eqref{eq:UnifoNoiseChannel}.
\begin{lemma}
\label{lemma:BoundsUnifEpsWith3Eps}
 Let $\alpha \in (0,1) \cup (1,\infty)$. For any $\varepsilon > 0$, it holds that
 \begin{equation}
 \label{eq:BoundUnifEpsWith3Eps}
 I_\alpha \left (\sfX; \sfY_{3 \varepsilon} \right )\leq I_\alpha(\sfX;\sfY_{\varepsilon}) \leq I_\alpha \left (\sfX; \sfY_{3 \varepsilon} \right ) + \max \left \{ 1, \frac{\alpha}{\alpha-1}\right \} \log(3).
 \end{equation}
\end{lemma}
\begin{proof}
The proof is provided in Appendix~\ref{app:BoundsUnifEpsWith3Eps}.
\end{proof}

\begin{lemma}
\label{lemma:UBIAlphaUniWithG}
Let $\alpha \in (0,1) \cup (1,\infty)$. For any $\varepsilon > 0$, it holds that
\begin{equation}
g_\alpha(\sfX,\varepsilon) \leq I_\alpha(\sfX;\sfY_\varepsilon)
\le
g_\alpha(\sfX,\varepsilon)
+\alpha\log (3)
-(\alpha-1)\left(I_\alpha(\sfX;\sfY_\varepsilon)-I_\alpha(\sfX;\sfY_{3\varepsilon})\right),
\label{eq:UBIAlphaUniWithG}
\end{equation}
where
   \begin{equation}
    g_{\alpha}(\sfX,\varepsilon)=\frac{\alpha}{\alpha-1} \log \left( \bbE\left[\left(\frac{1}{P_{\sfX}(\cB_\infty(\sfX,\varepsilon))}\right)^{\frac{\alpha-1}{\alpha}}\right] \right ).
    \label{def:g-alpha}
    \end{equation}
\end{lemma}
\begin{proof}
The proof is provided in Appendix~\ref{app:UBIAlphaUniWithG}.
\end{proof}

The next result uses the two lemmas above to bound $\alpha$-mutual information for the Gaussian noise channel as a function of $g_{\alpha}(\sfX,\varepsilon)$ in~\eqref{def:g-alpha}, which generalizes the $\alpha=1$ result in~\cite[Lemma~6]{wu2014information}. 
\begin{prop}
\label{prop:BoundsGaussian}
Let $\alpha \in (0,1) \cup (1,\infty)$. Assume that $H_{\frac{1}{\alpha}}( \lfloor\sfX \rfloor)<\infty$. Then, for any $\varepsilon >0$, it holds that
\begin{equation}
g_{\alpha}(\sfX,\varepsilon)
+c^\prime_\alpha \leq I_\alpha(\sfX;\sfX+\varepsilon \sfZ)
\leq g_{\alpha}(\sfX,\varepsilon)
+C^\prime_\alpha,
\label{eq:alpha_quant_lb}
\end{equation}
where $c^\prime_\alpha$ and $C^\prime_\alpha$ are constants that only depend on $\alpha$ and where $g_{\alpha}(\sfX,\varepsilon)$ is defined in~\eqref{def:g-alpha}.
\end{prop}
\begin{IEEEproof}
We let $\sfY_\sfZ=\sfX+\varepsilon \sfZ$ and consider the cases of $\alpha<1$ and $\alpha>1$ separately.
\begin{enumerate}
\item {\bf{Case~1:} $\alpha \in (1,+\infty)$.} 
From the upper bound in Lemma~\ref{lemma:RelationAlphaMIGaussianUniform}, we have that
\begin{align}
I_\alpha(\sfX;\sfY_\sfZ)  & \leq  \frac{\alpha}{\alpha-1} \log \left( C_{\alpha} \right ) +  I_\alpha(\sfX;\sfY_{\varepsilon}) 
\\& \leq \frac{\alpha}{\alpha-1} \log \left( C_{\alpha} \right ) + g_\alpha(\sfX,\varepsilon)
+\alpha\log (3)
-(\alpha-1)\left(I_\alpha(\sfX;\sfY_\varepsilon)-I_\alpha(\sfX;\sfY_{3\varepsilon})\right)
\label{eq:IneqFinalLBAlphaGreaterOneStep1}
\\& \leq g_\alpha(\sfX,\varepsilon) + \frac{\alpha}{\alpha-1} \log \left( C_{\alpha}\right ) +\alpha\log (3), \label{eq:IneqFinalLBAlphaGreaterOneStep2}
\end{align}
where~\eqref{eq:IneqFinalLBAlphaGreaterOneStep1} follows from the upper bound in Lemma~\ref{lemma:UBIAlphaUniWithG} and~\eqref{eq:IneqFinalLBAlphaGreaterOneStep2} is due to the lower bound in Lemma~\ref{lemma:BoundsUnifEpsWith3Eps}.

From the lower bound in Lemma~\ref{lemma:RelationAlphaMIGaussianUniform}, we have that
\begin{align}
I_\alpha(\sfX;\sfY_\sfZ)  & \geq \frac{\alpha}{\alpha-1} \log \left( c_{\alpha}  \right ) + I_\alpha(\sfX;\sfY_{\varepsilon}) 
\\& \geq g_\alpha(\sfX,\varepsilon) + \frac{\alpha}{\alpha-1} \log \left( c_{\alpha}  \right ),
\end{align}
where the last inequality is due to the lower bound in Lemma~\ref{lemma:UBIAlphaUniWithG}.
This concludes the proof of~\eqref{eq:alpha_quant_lb} for $\alpha>1$.

\item {\bf{Case~2:} $\alpha \in (0,1)$.} From the lower bound in Lemma~\ref{lemma:RelationAlphaMIGaussianUniform}, we have that
\begin{align}
I_\alpha(\sfX;\sfY_\sfZ)  & \leq \frac{\alpha}{\alpha-1} \log \left( c_{\alpha} \right ) + I_\alpha(\sfX;\sfY_{\varepsilon})
\\& \leq  \frac{\alpha}{\alpha-1} \log \left( c_{\alpha} \right )+ g_\alpha(\sfX,\varepsilon)
+\alpha\log (3)
+(1 - \alpha)\left(I_\alpha(\sfX;\sfY_\varepsilon)-I_\alpha(\sfX;\sfY_{3\varepsilon})\right) \label{eq:FinalUBAlpha01Step1}
\\& \leq g_\alpha(\sfX,\varepsilon) + \frac{\alpha}{\alpha-1} \log \left( c_{\alpha} \right ) + \log(3), \label{eq:FinalUBAlpha01Step2}
\end{align}
where in~\eqref{eq:FinalUBAlpha01Step1} we have used the upper bound in Lemma~\ref{lemma:UBIAlphaUniWithG} and~\eqref{eq:FinalUBAlpha01Step2} follows from the upper bound in Lemma~\ref{lemma:BoundsUnifEpsWith3Eps}.

From the upper bound in Lemma~\ref{lemma:RelationAlphaMIGaussianUniform}, we have that
\begin{align}
I_\alpha(\sfX;\sfY_\sfZ)  & \geq \frac{\alpha}{\alpha-1} \log \left ( C_{\alpha} \right ) + I_\alpha(\sfX;\sfY_{\varepsilon})
\\& \geq g_\alpha(\sfX,\varepsilon) + \frac{\alpha}{\alpha-1} \log \left ( C_{\alpha} \right ),
\end{align}
where the last inequality is due to the lower bound in Lemma~\ref{lemma:UBIAlphaUniWithG}.
This concludes the proof of~\eqref{eq:alpha_quant_lb} for $\alpha<1$. 
\end{enumerate}
This concludes the proof of Proposition~\ref{prop:BoundsGaussian}.
\end{IEEEproof}
An immediate consequence of the result in Proposition~\ref{prop:BoundsGaussian} is given by the following corollary, which provides a new relationship between the high-SNR behavior of  $\alpha$-mutual information and the R\'enyi information dimension,   and generalizes the result for $\alpha=1$ shown in~\cite{guionnet2007classical,wu2014information}. 
\begin{corollary}
\label{cor:InfoDim}
Let $\alpha \in (0,1) \cup (1,\infty)$. Assume that $H_{\frac{1}{\alpha}}( \lfloor\sfX \rfloor)<\infty$. It holds that~\footnote{Whenever $d_{1/\alpha}(\sfX)$ exists.}
\begin{equation}
\lim_{\mathrm{snr} \to \infty }\frac{I_\alpha(\sfX;\mathrm{snr})}{ \frac{1}{2} \log\left ({\mathrm{snr}} \right )} = d_{\frac{1}{\alpha}}(\sfX),
\end{equation}
where $d_{\frac{1}{\alpha}}(\sfX)$ is the information dimension of order $1/\alpha$ of $\sfX$ in Definition~\ref{def:AlphaInfDim}.
\end{corollary}
\begin{IEEEproof}
By using the result in Proposition~\ref{prop:BoundsGaussian}, we have that
    \begin{align}
\lim_{\varepsilon\downarrow 0}\frac{I_\alpha(\sfX;\sfX+\varepsilon \sfZ)}{\log(  \frac{1}{\varepsilon} )} &= \lim_{\varepsilon\downarrow 0} \frac{g_{\alpha}(\sfX, \varepsilon)}{  \log(  \frac{1}{\varepsilon} ) }
\\& = \lim_{\varepsilon\downarrow 0} \frac{\frac{\alpha}{\alpha-1} \log \left( \bbE\left[\left(\frac{1}{P_{\sfX}(B_\infty(\sfX,\varepsilon))}\right)^{\frac{\alpha-1}{\alpha}}\right] \right )}{\log(  \frac{1}{\varepsilon} )}
\\&=d_{\frac{1}{\alpha}}(\sfX),
\end{align}
where the last equality follows from Definition~\ref{def:AlphaInfDim}.
This concludes the proof of Corollary~\ref{cor:InfoDim}.
\end{IEEEproof}
The above result effectively shows that, in the high-SNR regime, we have that
\begin{equation}
    I_\alpha(\sfX; \mathrm{snr}) \approx \frac{d_{\frac{1}{\alpha}}(\sfX)}{2} \log(\mathrm{snr}).
\end{equation}
This can be paired with the following characterization of Rényi information dimension~\cite{csiszar1962dimension,smieja2014renyi}. 

\begin{prop}\label{prop:beh_renyin_dimenion}
The following properties hold:
\begin{itemize} 
    \item Suppose that $\sfX_D$ is a discrete random variable and assume that $H_{ \alpha} (\sfX_D)<\infty$\footnote{We require $H_\alpha(\sfX_D)<\infty$ because  $d_\alpha(\sfX_D)=0$ holds only under this condition  for $0<\alpha<1$. Without it, the conclusion can fail: there exist purely atomic probability measures with $H_\alpha(\sfX_D)=\infty$ for which $d_\alpha(\sfX_D)>0$ (e.g., B\'artfa\"i’s example discussed by Csisz\'ar in~\cite{csiszar1962dimension}).}. Then, 
    \begin{equation}
        d_\alpha(\sfX_D) = 0, \qquad \alpha \in (0,1) \cup (1,\infty).
    \end{equation}
    \item Suppose that $\sfX_C$ is a random variable with absolutely continuous distribution and assume that $H_{ \alpha} ( \lfloor \sfX_C \rfloor )<\infty$. Then, 
    \begin{equation}
        d_\alpha(\sfX_C) = 1, \qquad \alpha \in (0,1) \cup (1,\infty).
    \end{equation}
    \item Suppose that $\sfX$ has a mixed distribution, i.e.,
    \begin{equation}
        P_{\sfX} = (1-\rho) P_{\sfX_1} + \rho P_{\sfX_2}, \qquad \rho \in (0,1).
    \end{equation}
    Moreover, assume that $\sfX_1$ and $\sfX_2$ have finite R\'enyi dimension.
    Then, 
    \begin{equation}
        d_\alpha(\sfX) =
        \begin{cases}
        \max \big( d_\alpha(\sfX_1), d_\alpha(\sfX_2) \big), & \alpha \in (0,1), \\
        \min \big( d_\alpha(\sfX_1), d_\alpha(\sfX_2) \big), & \alpha \in (1,\infty).
        \end{cases}
    \end{equation}
    Furthermore, if $\sfX_1 = \sfX_C$ and $\sfX_2 = \sfX_D$, then
    \begin{equation}
        d_1(\sfX) = 1 - \rho.
    \end{equation}
\end{itemize}
\end{prop}
Combining Corollary~\ref{cor:InfoDim} with Proposition~\ref{prop:beh_renyin_dimenion}, we obtain the following characterization. 
\begin{corollary}
\label{corollary:HIghSNRCsizar}
Assume that the assumptions in Proposition~\ref{prop:beh_renyin_dimenion} hold.
Let
\begin{equation}
P_{\sfX} = (1-\rho) P_{\sfX_C} + \rho P_{\sfX_D}, \qquad \rho \in [0,1],
\end{equation}
where $\sfX_C$ is a random variable with absolutely continuous distribution and $\sfX_D$ is a discrete random variable.
Then, it holds that
\begin{equation}
d_{\frac{1}{\alpha}}(\sfX) = \lim_{ {\rm snr} \to \infty}
\frac{I_\alpha(\sfX;{\rm snr} )}{\tfrac{1}{2} \log({\rm snr} )}
=
\begin{cases}
\mathbbm 1_{\{\rho = 0\}}, & \alpha \in (0,1), \\
1 - \rho, & \alpha = 1, \\
\mathbbm 1_{\{\rho < 1\}}, & \alpha \in (1,\infty).
\end{cases}
\end{equation}
\end{corollary}
Corollary~\ref{corollary:HIghSNRCsizar} shows that the high-SNR behavior of $\alpha$-mutual information exhibits a sharp phase transition at $\alpha = 1$. For $\alpha \in (0,1)$, the presence of any discrete component suppresses the $\log({\rm snr} )$ growth, whereas for $\alpha \in (1,\infty)$, any continuous component, no matter how small, ensures $\log({\rm snr} )$ scaling.

\begin{exa}
\label{exa:Corollary3-mixed}
Consider the mixed distribution $P_{\sfX}
=
\frac{1}{2}P_{\sfX_C}
+
\frac{1}{2}P_{\sfX_D}$,
where $\sfX_C\sim\mathcal N(0,1)$ and $\sfX_D$ is equiprobable on $\{-1,1\}$. Thus, in the notation of
Corollary~\ref{corollary:HIghSNRCsizar}, we have $\rho=1/2$ and
\begin{equation}
\lim_{{\rm snr}\to\infty}
\frac{
I_\alpha(\sfX;{\rm snr})
}{
\frac{1}{2}\log({\rm snr})
}
= \frac{1}{2} \mathbbm 1_{\{\alpha=1\}} + \mathbbm 1_{\{\alpha>1\}}.
\end{equation}
Fig.~\ref{fig:Corollary3-mixed} plots $\frac{
I_\alpha(\sfX;{\rm snr})
}{\frac{1}{2}\log({\rm snr})}$ as a function of $\log_{10}({\rm snr})$ for $\alpha\in\{0.5,1,2\}$ using Monte Carlo simulation with $5 \times 10^5$ iterations. The three curves converge to $0$, $1/2$, and $1$, respectively, in agreement with Corollary~\ref{corollary:HIghSNRCsizar}. In particular, this example illustrates the phase transition at $\alpha=1$: although the discrete and absolutely continuous components have equal mixture weights, the
discrete component determines the high-SNR scaling for $\alpha<1$, whereas the absolutely continuous component determines it for $\alpha>1$; at $\alpha=1$, the limiting coefficient is the weight $1-\rho=1/2$ of the absolutely continuous component.
\begin{figure}[t]
\centering
\begin{tikzpicture}
\begin{axis}[
    width=9cm,
    height=7.2cm,
    grid=major,
    xmin=0.7,
    xmax=40,
    ymin=0,
    ymax=1.45,
    xtick={1,5,10,15,20,25,30,35,40},
    xlabel={$\log_{10}({\rm snr})$},
    ylabel={$\frac{I_\alpha(\sfX;{\rm snr})}{\frac{1}{2}\log({\rm snr})}$},
    legend style={at={(0.97,0.97)},anchor=north east,draw=black,fill=white},
    legend cell align=left,
]
\addplot[red,thick,mark=*,mark repeat=50]
table[x=log10snr,y=ratio05,col sep=comma]
{Corollary3_mixed_highsnr_MC.csv};
\addlegendentry{$\alpha=0.5$}
\addplot[blue,thick,mark=o,mark repeat=50]
table[x=log10snr,y=ratio1,col sep=comma]
{Corollary3_mixed_highsnr_MC.csv};
\addlegendentry{$\alpha=1$}
\addplot[green!60!black,thick,mark=triangle*,mark repeat=50]
table[x=log10snr,y=ratio2,col sep=comma]
{Corollary3_mixed_highsnr_MC.csv};
\addlegendentry{$\alpha=2$}
\addplot[black,dotted,forget plot] coordinates {(0.7,0) (40,0)};
\addplot[black,dotted,forget plot] coordinates {(0.7,0.5) (40,0.5)};
\addplot[black,dotted,forget plot] coordinates {(0.7,1) (40,1)};
\end{axis}
\end{tikzpicture}
\caption{High-SNR behavior of $\alpha$-mutual information for the mixed
input in Example~\ref{exa:Corollary3-mixed}, with $\rho=1/2$,
$\sfX_C\sim\mathcal N(0,1)$, and $\sfX_D$ equiprobable on $\{-1,1\}$.
}
\label{fig:Corollary3-mixed}
\end{figure}
\end{exa}

\section{Conclusions and Future Directions}
\label{sec:Conclusion}
In this paper, we studied Sibson’s $\alpha$-mutual information in the context of the additive Gaussian noise channel and developed a number of structural results that parallel those known for the classical case $\alpha = 1$. In particular, we established several regularity properties, derived an $\alpha$-I-MMSE relationship, and showed how this identity leads to generalized de Bruijn-type relations and new estimation-theoretic representations of R\'enyi entropy and differential R\'enyi entropy. We further characterized both the low- and high-SNR behaviors of $\alpha$-mutual information, including its connections to $\alpha$-information dimension.

Several interesting directions for future work emerge from our results. First, the $\alpha$-I-MMSE relationship suggests the possibility of deriving new functional and information-theoretic inequalities, in analogy with the classical I-MMSE identity, which has been used to obtain entropy power inequalities~\cite{verdu2006simple,tulino2006monotonic,guo2006proof}, log-Sobolev inequalities~\cite{raginsky2013concentration,matrixSNR}, and related bounds~\cite{dytso2017view,bustin2009mmse,lozano2006optimum,GuoIT2011} to name a few. It would be interesting to understand whether similar inequalities can be obtained in the $\alpha$-setting and what new phenomena may arise.

Second, while our analysis focuses on finite values of $\alpha$, the limiting case $\alpha = \infty$ remains largely unexplored in this context. This regime is closely connected to notions of maximal leakage and privacy~\cite{esposito2021generalization}, and it would be of interest to investigate whether the techniques developed here can be extended to provide new insights into privacy-constrained inference problems.

Third, our results on the high-SNR behavior and $\alpha$-information dimension suggest a deeper connection with $\alpha$-rate-distortion theory. In particular, it would be interesting to explore whether analogues of the relationships established in~\cite{kawabata1994rate} between information dimension and rate-distortion functions can be extended to the $\alpha$-R\'enyi setting.

Fourth, the strict concavity (and convexity) properties established in this paper open the door to studying optimization problems involving $\alpha$-mutual information. For example, it would be natural to investigate whether, under amplitude constraints, the maximizing input distributions are discrete and to characterize their structure, in analogy with classical results for Shannon mutual information~\cite{smith1971information,dytso2019capacity}.

Finally, it would be of interest to extend the present analysis to the vector case, see for example~\cite{stotz2016degrees}, and also beyond the Gaussian noise channel. In particular, studying $\alpha$-mutual information for other canonical models, such as the Poisson channel, may reveal new connections between information measures and estimation quantities. In this direction, an intriguing question is whether generalized I-MMSE-type relationships can be established that relate $\alpha$-mutual information to estimation errors defined via more general divergences, such as Bregman divergences~\cite{guo2008mutual,atar2012mutual,jiao2017relations}.

\appendices

\section{Proof Lemma~\ref{lem:bounds_on_renyi_entorpy}}
\label{app:proof_entropy_moments}
Let $p_j = \Pr \left(  \sfW   = j\right )$, where $j \in \{0,1,2,\ldots\}$ and note that
\begin{align}
\sum_{j=0}^{+\infty} p_j^{\alpha} &= p_0^{\alpha} + \sum_{j = 1}^{+\infty} p_j^{\alpha}
\\& = p_0^{\alpha} + \sum_{j = 1}^{+\infty} \left(|j|^k p_j \right )^{\alpha} |j|^{-k \alpha}
\\& \leq p_0^{\alpha} + \left( \sum_{j =1}^{+\infty} \left [ \left(|j|^k p_j \right )^{\alpha} \right ]^{\frac{1}{\alpha}}\right )^{\alpha} \left( \sum_{j =1}^{+ \infty} \left( |j|^{-k \alpha}\right )^{\frac{1}{1-\alpha}} \right )^{1-\alpha} \label{eq:HolderFloorX}
\\& = p_0^{\alpha} + \left( \sum_{j =1}^{+ \infty} |j|^k p_j\right )^{\alpha} \left( \sum_{j =1}^{+ \infty} |j|^{-\frac{k \alpha}{1-\alpha}} \right )^{1 - \alpha}
\\& = p_0^{\alpha} + \bbE^{\alpha}\left[  \sfW^k \right] \left( \sum_{j=1}^{+\infty} \left( \frac{1}{|j|} \right )^{\frac{k \alpha }{1-\alpha}} \right )^{ 1-\alpha}
\\& < \infty, \label{eq:FinitenessFloorX}
\end{align}
where~\eqref{eq:HolderFloorX} follows from using H\"older's inequality, and~\eqref{eq:FinitenessFloorX} is due to the assumption $\bbE[ \sfW^k ]<\infty$
and the fact that the second term in the multiplication is a p-series that is finite whenever $k>\frac{1-\alpha}{\alpha}$.
This proves the first property.

For the second property, consider a random variable $\sfW$ such that, for all $j \in \{2,3,\ldots\}$,
\begin{equation}
p_j = \frac{C}{j^{\frac{1}{\alpha}} (\log j)^{\frac{1}{\alpha}}},
\end{equation}
where $C$ is the normalization constant.
We have that
\begin{equation}
\bbE \left [ \sfW^k\right ] = C \sum_{j=2}^{+\infty} \frac{j^{k-\frac{1}{\alpha}}}{ (\log j)^{\frac{1}{\alpha}}},
\end{equation}
which using $k = \frac{1-\alpha}{\alpha}$ leads to
\begin{equation}
\bbE \left [ \sfW^{\frac{1-\alpha}{\alpha}}\right ] = C \sum_{j=2}^{+\infty} \frac{1}{ j (\log j)^{\frac{1}{\alpha}}},
\end{equation}
which is a convergent Bertrand series since $\frac{1}{\alpha}>1$. Moreover, from Definition~\ref{def:ReyniEntropy}, we have that
\begin{align}
H_{\alpha}(\sfW) & = \frac{1}{1-\alpha} \log \left( \sum_{j=2}^{+\infty }  p_j^\alpha  \right ) 
\\& = \frac{1}{1-\alpha} \log \left( \sum_{j=2}^{+\infty } \frac{C^\alpha}{j (\log j)} \right )
\\& = \infty,
\end{align}
where the last equality follows since the series is divergent. This proves the second property and concludes the proof of Lemma~\ref{lem:bounds_on_renyi_entorpy}.

\section{Proof of Lemma~\ref{lemma:UBLBHalpha}}
\label{app:lemma:UBLBHalpha}
We start by observing that $\lfloor \langle\sfX\rangle_m \rfloor = \lfloor \sfX \rfloor$, which follows from the property of nested divisions of the floor function.     Thus, $\lfloor \sfX \rfloor$ is a deterministic function of $\langle\sfX\rangle_m$, which implies that 
\begin{equation}
H_{\alpha} \left( \langle\sfX\rangle_m \right ) = H_{\alpha} \left( \langle\sfX\rangle_m,\lfloor \sfX \rfloor \right ) \geq H_{\alpha} \left(\lfloor \sfX \rfloor  \right ).
\end{equation}
This proves the lower bound in~\eqref{eq:LBUBQuantizedX}. 
To prove the upper bound in~\eqref{eq:LBUBQuantizedX}, for short of notation, we let $\sfK = \lfloor \sfX \rfloor$ and $\sfQ  = \langle\sfX\rangle_m$.
We have that
\begin{align}
H_{\alpha}\left(\sfQ  \right ) & = H_{\alpha}\left(\sfQ, \sfK  \right ) \label{eq:DefReyniEntr}
\\& =  \frac{1}{1-\alpha} \log \left( \sum_{k} P_{\sf{K}}^{\alpha}(k) \sum_q P^{\alpha}_{\sfQ|\sfK}(q|k) \right )
\\& = \frac{1}{1-\alpha} \log \left( \sum_{k} P_{\sf{K}}^{\alpha}(k) \sum_{j=0}^{m-1} \bbP^{\alpha} \left( \left . \sfQ = k+\frac{j}{m} \right | \sfK = k \right )\right ) \label{eq:PQGivenK}
\\& \leq \frac{1}{1-\alpha}  \log \left( \sum_{k} P_{\sf{K}}^{\alpha}(k) \left( \sum_{j=0}^{m-1} \frac{\bbP \left( \left . \sfQ = k+\frac{j}{m} \right | \sfK = k \right )}{m} \right )^{\alpha} m \right ) \label{eq:JensenUBHalpha}
\\& = \frac{1}{1-\alpha}  \log \left(\sum_{k} P_{\sf{K}}^{\alpha}(k) \,  m^{1-\alpha}  \right )
\\& = H_{\alpha}\left(\sfK  \right ) + \log(m),
\end{align}
where~\eqref{eq:PQGivenK} follows by noting that
when $\sfK=k$, we have that 
\begin{align}
k \leq \sfX < k+1 & \implies m k \leq m \sfX < m (k+1) 
\\&\implies \lfloor m \sfX \rfloor \in \left\{ mk, mk+1, \ldots, mk +m-1\right \}
\\& \implies \sfQ \in \left\{ k, k+ \frac{1}{m}, \ldots, k+\frac{m-1}{m}\right \},
\end{align}
i.e., for a given $k$, $\sfQ$ can get $m$ different values; and in~\eqref{eq:JensenUBHalpha} we have used Jensen's inequality.
This concludes the proof of Lemma~\ref{lemma:UBLBHalpha}.

\section{Proof of~\eqref{eq:AlphaMI} in Lemma~\ref{lemma:GaussianAlphaMI}}
\label{app:GaussianAlphaMI}
By evaluating~\eqref{eq:MinimizerAlphaMI} for the Gaussian noise channel in~\eqref{eq:GaussianChannel}, we have that
\begin{equation}
f_{\sfY_\alpha}(y) = \kappa_\alpha \; \sqrt{\rm{snr}}  \; \bbE^{ \frac{1}{\alpha} } \left[   \phi^\alpha \left (\sqrt{\rm snr} \left( y-  \sfX  \right ) \right ) \right],
\end{equation}
with 
\begin{equation}
    \kappa_\alpha^{-1}=  \sqrt{\rm snr} \int_\bbR \bbE^{ \frac{1}{\alpha} } \left[   \phi^\alpha \left (\sqrt{\rm snr} \left( y-  \sfX  \right ) \right ) \right] \rmd y.
\end{equation}
Then, from Definition~\ref{def:MIalpha}, we have that
\begin{align}
I_{\alpha}(\sfX;\mathrm{snr})&= D_{\alpha}\left( P_{\sfX,\sfY} \| P_{\sfX} \times P_{\sfY_\alpha} \right ) \\
&= \frac{1}{\alpha-1}  \log \left( \int_{\bbR^2} (f_{\sfY| \sfX ;{\rm snr} }(y|x))^\alpha   \   (f_{\sfY_\alpha}(y))^{1-\alpha}  \ \rmd P_{\sfX}(x) \ {\rmd}y \right ) \\
&= \frac{1}{\alpha-1}  \log \left( \int_{\bbR^2} \left( \sqrt{\rm{snr}} \, \phi \left(\sqrt{\rm{snr}}(y-x)\right) \right)^\alpha (  f_{\sfY_\alpha}(y) )^{1-\alpha} \ {\rmd} P_\sfX(x) \ \rmd y \right ) \\
&= \frac{1}{\alpha-1} \log \left(   \int_\bbR  \bbE \left[ \left( \sqrt{\rm{snr}} \, \phi \left(\sqrt{\rm{snr}}(y-\sfX)\right) \right)^\alpha\right] f_{\sfY_\alpha}^{1-\alpha}(y) \ \rmd y   \right)\\
&= \frac{1}{\alpha-1} \log \left(   \int_\bbR  \frac{f_{\sfY_\alpha}^{\alpha}(y)}{ \kappa_\alpha^{\alpha}}  f_{\sfY_\alpha}^{1-\alpha}(y) \ \rmd y   \right)\\
&= \frac{\alpha}{\alpha-1} \log \left(   \frac{1}{\kappa_\alpha}  \right) +  \frac{1}{\alpha-1} \log \left(   \int_\bbR    f_{\sfY_\alpha}(y) \ \rmd y   \right)\\
&= \frac{\alpha}{\alpha-1} \log \left(   \frac{1}{\kappa_\alpha}  \right) \\
&= \frac{\alpha}{\alpha-1} \log \left(   \int_\bbR \bbE^{ \frac{1}{\alpha} } \left[    \left( \sqrt{\rm snr} \; \phi \left (\sqrt{\rm snr} \left( y-  \sfX  \right ) \right ) \right)^\alpha \right] \rmd y  \right). 
\label{eq:ProofAlphaMI}
\end{align}
This concludes the proof of~\eqref{eq:AlphaMI} in Lemma~\ref{lemma:GaussianAlphaMI}.

\section{Proof of Lemma~\ref{lemma:UniformMIProb}}
\label{app:proofUniformNoiseAlphaMI}
By applying the definition of $\alpha$-mutual information, we have that
\begin{align}
I_\alpha(\sfX;\sfY_{\varepsilon}) &= \frac{\alpha}{\alpha-1} \log \left ( \int_\bbR  \bbE^{\frac{1}{\alpha}} \left [ f^\alpha_{\sfY_{\varepsilon}|\sfX} (y| \sfX)\right ] \, \rmd y\right )
\\& = \frac{\alpha}{\alpha-1} \log \left ( \int_\bbR \bbE^{\frac{1}{\alpha}} \left [ \left( \frac{1}{\varepsilon} \right )^\alpha \mathbbm 1_{\left \{ |y -\sfX| \leq \frac{\varepsilon}{2}  \right \}} \right ] \, \rmd y \right )
\\& =  - \frac{\alpha}{\alpha-1} \log(\varepsilon) + \frac{\alpha}{\alpha-1} \log \left( \int_\bbR P_{\sfX}^{ \frac{1}{\alpha}  } \left(\cB_\infty \left (y,\frac{\varepsilon}{2} \right )\right) \, \rmd y \right ).
\end{align}
This proves~\eqref{eq:ExpIAlpha1}. Using~\eqref{eq:ExpIAlpha1}, we have that
\begin{align}
I_\alpha(\sfX;\sfY_{\varepsilon}) & = - \frac{\alpha}{\alpha-1} \log(\varepsilon) + \frac{\alpha}{\alpha-1} \log \left( \int_\bbR P_{\sfX}^{\frac{1}{\alpha}} \left(\cB_\infty \left (y,\frac{\varepsilon}{2} \right )\right) \, \rmd y \right )
\\& = - \frac{\alpha}{\alpha-1} \log(\varepsilon) + \frac{\alpha}{\alpha-1} \log \left( \varepsilon^{\frac{1}{\alpha}} \int_\bbR f^{\frac{1}{\alpha}}_{\sfY_{\varepsilon}}(y) \, \rmd y \right ) 
\label{eq:LBAlphaMIYEpsStep2}
\\& = - \frac{\alpha}{\alpha-1} \log(\varepsilon) + \frac{\alpha}{\alpha-1} \log \left( \varepsilon^{\frac{1}{\alpha}} \bbE \left [ \left( \frac{1}{f_{\sfY_{\varepsilon}}(\sfY_{\varepsilon})}\right )^{\frac{\alpha-1}{\alpha}}\right ] \right ) 
\\& =  \frac{\alpha}{\alpha-1} \log \left( \bbE \left [ \left( \frac{1}{P_{\sfX}\left(\cB_\infty \left (\sfY_{\varepsilon}, \frac{\varepsilon}{2}  \right )\right)}\right )^{\frac{\alpha-1}{\alpha}}\right ] \right ),\label{eq:LBAlphaMIYEpsStep3}
\end{align}
where~\eqref{eq:LBAlphaMIYEpsStep2} and~\eqref{eq:LBAlphaMIYEpsStep3} follow since
\begin{align}
f_{\sfY_{\varepsilon}}(y)
&=\bbE\left[f_{\sfY_{\varepsilon}|\sfX}(y|\sfX)\right] \\
&=\bbE\left[\frac{1}{\varepsilon} \mathbbm 1_{\{|y-\sfX|\leq \varepsilon/2\}}\right]\\
&=\frac{1}{\varepsilon} P_{\sfX}\left(\cB_\infty \left (y, \frac{\varepsilon}{2} \right )\right).
\label{eq:pY_ball}
\end{align}
This proves~\eqref{eq:ExpIAlpha2} and concludes the proof of Lemma~\ref{lemma:UniformMIProb}.

{
\section{Proof of Lemma~\ref{lemma:RelationAlphaMIGaussianUniform}}
\label{app:RelationAlphaMIGaussianUniform}
We let $\sfY_\sfZ=\sfX+\varepsilon \sfZ$, where $\sfZ=\sfT+\sfW$ and
\begin{align}
\sfT &=  \left\lfloor \sfZ+\frac12\right\rfloor\in\mathbb Z,
\\& \sfW = \sfZ - \sfT \in\left[-\frac12,\frac12 \right ).
\end{align}
Moreover, for $t\in\mathbb Z$, we let
\begin{align}
&p_t = \mathbb P(\sfT=t)=\int_{t-\frac12}^{t+\frac12}\phi(z) \, \rmd z,
\\& f_t(w) = f_{\sfW|\sfT}(w|t)
=
\frac{\phi(t+w) \mathbbm 1_{\{|w|\le \frac12\}}}{p_t}.
\label{eq:pt_ft}
\end{align}
We start by noting that
\begin{align}
\exp \left( \frac{\alpha-1}{\alpha} I_\alpha(\sfX;\sfY_\sfZ) \right )&= \int_{\mathbb R}
\bbE^\frac{1}{\alpha}[f^\alpha_{\sfY_\sfZ|\sfX}(y|\sfX)] \, \rmd y 
\\& =  \int_{\mathbb R} \frac{1}{\varepsilon} \, \bbE^\frac{1}{\alpha} \left [ \phi^\alpha \left(\frac{y-\sfX}{\varepsilon}\right) \right ] \, \rmd y 
\\& =  \int_{\mathbb R} \frac{1}{\varepsilon} \, \bbE^\frac{1}{\alpha} \left [ \left( \sum_{t\in\mathbb Z} p_t \, f_t \left (\frac{y-\sfX}{\varepsilon} -t\right ) \right )^\alpha \right ] \, \rmd y,  \label{eq:{eq:kernel_mixture_ft}}
\end{align}
where~\eqref{eq:{eq:kernel_mixture_ft}} follows since
\begin{equation}
\phi(u)
=
\sum_{t\in\mathbb Z}\phi(u)  \, \mathbbm 1_{\left \{ u\in \left [t-\tfrac12,t+\tfrac12 \right ) \right \} } =\sum_{t\in\mathbb Z} p_t \, f_t(u-t),
\end{equation}
where the last equality is due to~\eqref{eq:pt_ft} since on the set $\{u\in[t-\tfrac12,t+\tfrac12)\}$, we have that $u=t+w$ with $w\in[-\tfrac12,\tfrac12)$.
In our bounds, we will define $\widetilde{\sfY}_t = \sfX + \varepsilon (t+\sfW_t)$ with $\sfW_t \sim f_t$ independent of $\sfX$ for which
\begin{equation}
\label{eq:PDFYTildetGivenX}
f_{\widetilde{\sfY}_t|\sfX}(y|x)=\frac{1}{\varepsilon}\,f_t \left(\frac{y-x}{\varepsilon}-t\right),
\end{equation}
and $\sfY_t = \sfX +\varepsilon \sfW_t$.
We note that $I_\alpha(\sfX;\widetilde{\sfY}_t) = I_\alpha(\sfX;\sfY_t)$
since  $\widetilde{\sfY}_t$ differs from $\sfY_t$ only by an additive constant and $\alpha$-mutual information is invariant under bijective measurable transformations of the observation variable~\cite{esposito2025sibson}. In our bounds, we will also use the facts that
\begin{equation}
f_{\sfY_t|\sfX}(y|x)
=
\frac{1}{\varepsilon}f_t \left(\frac{y-x}{\varepsilon}\right)
\ge
\frac{1}{\varepsilon} \rme^{-|t|}  \mathbbm 1_{\left\{\left|\frac{y-x}{\varepsilon}\right|\le \frac12\right\}}
=
\frac{\rme^{-|t|}}{\varepsilon}  \mathbbm 1_{\{|y-x|\le \frac{\varepsilon}{2}\}},
\label{eq:kernel_lower_pointwise1}
\end{equation}
and
\begin{equation}
f_{\sfY_t|\sfX}(y|x)
=
\frac{1}{\varepsilon}f_t \left(\frac{y-x}{\varepsilon}\right)
\le
\frac{1}{\varepsilon} \rme^{|t|}  \mathbbm 1_{\left\{\left|\frac{y-x}{\varepsilon}\right|\le \frac12\right\}}
=
\frac{\rme^{|t|}}{\varepsilon}  \mathbbm 1_{\{|y-x|\le \frac{\varepsilon}{2}\}},
\label{eq:kernel_lower_pointwise2}
\end{equation}
where the bounds follow from~\cite[eq.(264)]{wu2014information}, that is,
\begin{equation}
\rme^{-|t|}\,\mathbbm 1_{\{|w|\le 1/2\}} \le f_t(w)\le \rme^{|t|} \,\mathbbm 1_{\{|w|\le 1/2\}},
\qquad \forall w\in\mathbb R.
\label{eq:ft_compare_uniform}
\end{equation}
We start by proving the {\em upper bound} in~\eqref{eq:RelationAlphaMIGaussianUniform}. We begin by defining 
\begin{align}
    C_\alpha(\sfT)  &=  \left\{  \begin{array}{cc}
    1 & \alpha >1, \\ 
    \sum_{t \in \bbZ} p_t^\alpha  & \alpha <1,
    \end{array} \right. \\
    \tilde{p}_t &=  \left\{  \begin{array}{cc}
    p_t & \alpha >1, \\ 
     \frac{p_t^\alpha}{C_\alpha(\sfT)}  & \alpha <1,
    \end{array} \right. 
\end{align}
for the now we claim that $C_\alpha(\sfT)<\infty$ for all $\alpha \in (0,1)$ and, hence, $\tilde{p}_t$ is proper pmf. 
Next, note that for $\alpha \in (1,\infty)$ by Minkowski's inequality, we have that
\begin{equation}
\bbE^\frac{1}{\alpha} \left [ \left( \sum_{t\in\mathbb Z} p_t \, f_t \left (\frac{y-\sfX}{\varepsilon} -t\right ) \right )^\alpha \right ] \le \sum_{t\in\mathbb Z} p_t \, \bbE^\frac{1}{\alpha} \left [  f_t^\alpha \left (\frac{y-\sfX}{\varepsilon} -t\right )  \right ], \label{eq:inequlaity_on_sum_for_a>1}
\end{equation}
and for $\alpha \in (0,1)$, we have that
\begin{align}
\bbE^\frac{1}{\alpha} \left [ \left( \sum_{t\in\mathbb Z} p_t \, f_t \left (\frac{y-\sfX}{\varepsilon} -t\right ) \right )^\alpha \right ] &\le  \left( \sum_{t\in\mathbb Z} p_t^\alpha \,\bbE  \left [   f_t^\alpha \left (\frac{y-\sfX}{\varepsilon} -t\right ) \right ] \right)^\frac{1}{\alpha}  \label{eq:sum_bound} \\
& =  C_\alpha^\frac{1}{\alpha}( \sfT) \left( \sum_{t\in\mathbb Z} \tilde{p}_t  \,\bbE  \left [   f_t^\alpha \left (\frac{y-\sfX}{\varepsilon} -t\right ) \right ] \right)^\frac{1}{\alpha}  \\
&  \le  C_\alpha^\frac{1}{\alpha}( \sfT)    \sum_{t\in\mathbb Z} \tilde{p}_t  \,\bbE^\frac{1}{\alpha}  \left [   f_t^\alpha \left (\frac{y-\sfX}{\varepsilon} -t\right ) \right ] ,\label{eq:using_jensens_inequality}  
\end{align}
where~\eqref{eq:sum_bound} follows from the inequality $\left(\sum_{t\in\mathbb Z}b_t\right)^k
\le
\sum_{t\in\mathbb Z}b_t^k$ for positive $b_t$ and $k \in (0,1)$ and~\eqref{eq:using_jensens_inequality} follows by using Jensen's inequality. 
Combining~\eqref{eq:inequlaity_on_sum_for_a>1} and~\eqref{eq:using_jensens_inequality}, we have the following generic inequality for $\alpha \in (0,1) \cup (1,\infty)$
\begin{equation}
\bbE^\frac{1}{\alpha} \left [ \left( \sum_{t\in\mathbb Z} p_t \, f_t \left (\frac{y-\sfX}{\varepsilon} -t\right ) \right )^\alpha \right ]   \le  C_\alpha^\frac{1}{\alpha}( \sfT)  \sum_{t\in\mathbb Z} \tilde{p}_t  \,\bbE^\frac{1}{\alpha}  \left [   f_t^\alpha \left (\frac{y-\sfX}{\varepsilon} -t\right ) \right ].    \label{eq:sum_bound_v2}
\end{equation}
From~\eqref{eq:{eq:kernel_mixture_ft}}, we have that
\begin{align}
\exp \left( \frac{\alpha-1}{\alpha} I_\alpha(\sfX;\sfY_\sfZ) \right )&  \leq C_\alpha^\frac{1}{\alpha}( \sfT)  \int_{\mathbb R} \frac{1}{\varepsilon} \, \sum_{t\in\mathbb Z} \tilde{p}_t \, \bbE^\frac{1}{\alpha} \left [  f_t^\alpha \left (\frac{y-\sfX}{\varepsilon} -t\right )  \right ] \, \rmd y \label{eq:MinkUB}
\\& = C_\alpha^\frac{1}{\alpha}( \sfT)  \sum_{t\in\mathbb Z} \tilde{p}_t \int_{\mathbb R} \frac{1}{\varepsilon} \,  \bbE^\frac{1}{\alpha} \left [  f_t^\alpha \left (\frac{y-\sfX}{\varepsilon} -t\right )  \right ] \, \rmd y \label{eq:FubiniTonelliUB}
\\& = C_\alpha^\frac{1}{\alpha}( \sfT)  \sum_{t\in\mathbb Z} \tilde{p}_t \int_{\mathbb R}   \bbE^\frac{1}{\alpha} \left [  f^\alpha_{\widetilde{\sfY}_t|\sfX}(y|\sfX)  \right ] \, \rmd y \label{eq:UBIalphaPDFYTildetGivenX}
\\& = C_\alpha^\frac{1}{\alpha}( \sfT) \sum_{t\in\mathbb Z} \tilde{p}_t \int_{\mathbb R}   \bbE^\frac{1}{\alpha} \left [  f^\alpha_{\sfY_t|\sfX}(y|\sfX)  \right ] \, \rmd y \label{eq:EqualityMIYtildeY}
\\& \leq C_\alpha^\frac{1}{\alpha}( \sfT) \sum_{t\in\mathbb Z}  \frac{1}{\varepsilon} \tilde{p}_t \rme^{|t|} \int_{\mathbb R} \bbE^\frac{1}{\alpha} \left [  \mathbbm 1_{\{|y-\sfX|\le \frac{\varepsilon}{2}\}} \right ] \, \rmd y \label{eq:WuUBGeneral}
\\&= C_\alpha^\frac{1}{\alpha}( \sfT) \sum_{t\in\mathbb Z}  \frac{1}{\varepsilon} \tilde{p}_t \rme^{|t|} \int_{\mathbb R}
P_\sfX^{\frac{1}{\alpha}} \left(\cB_\infty \left (y,\frac{\varepsilon}{2} \right ) \right) \, \rmd y 
\\& = C_\alpha^\frac{1}{\alpha}( \sfT) \left( \sum_{t\in\mathbb Z}   \tilde{p}_t \rme^{|t|} \right ) \exp \left( \frac{\alpha-1}{\alpha} I_\alpha(\sfX;\sfY_{\varepsilon}) \right ),   \label{eq:FinalUBBeforeBranch}
\end{align}
where:~\eqref{eq:MinkUB} follows from~\eqref{eq:sum_bound_v2};~\eqref{eq:FubiniTonelliUB} follows from Fubini-Tonelli's theorem; in~\eqref{eq:UBIalphaPDFYTildetGivenX}, we have used~\eqref{eq:PDFYTildetGivenX};~\eqref{eq:EqualityMIYtildeY} is due to the fact that  $I_\alpha(\sfX;\widetilde{\sfY}_t) = I_\alpha(\sfX;\sfY_t)$ (as justified above); in~\eqref{eq:WuUBGeneral}, we have used~\eqref{eq:kernel_lower_pointwise2}; and ~\eqref{eq:FinalUBBeforeBranch} follows from Lemma~\ref{lemma:UniformMIProb}.
It remains to show that $C_\alpha(\sfT)$ and $ \sum_{t\in\mathbb Z}   \tilde{p}_t \rme^{|t|}$ are bounded for $0< \alpha<1$ (the case $\alpha>1$ is trivial).
First, note that for $t \ge 1$, we have
\begin{equation}
p_t = \int_{t-\frac12}^{t+\frac12}\phi(z) \, \rmd z \le \phi \left(t -\frac{1}{2} \right) .\label{Eq:bound_pt}
\end{equation}
Second, note that for $\alpha \in (0,1)$ we have that
\begin{align}
    C_\alpha(\sfT) &= \sum_{t\in\mathbb Z}   p_t^\alpha  \\
    &= p_0^\alpha + 2 \sum_{t\in\mathbb Z , \, t \ge 1}   p_t^\alpha\\
    &\le p_0^\alpha + 2 \sum_{t\in\mathbb Z , \, t \ge 1}   \phi^\alpha \left(t -\frac{1}{2} \right)\\
    &<\infty,
\end{align}
where the first inequality follows from \eqref{Eq:bound_pt} and the last inequality follows since Gaussian tails are summable. Third, we have that  \begin{align}
    \sum_{t\in\mathbb Z}   \tilde{p}_t \rme^{|t|} & = \frac{1}{C_\alpha(\sfT)} \sum_{t\in\mathbb Z}   p_t^\alpha \rme^{|t|} \\
    & =  \frac{1}{C_\alpha(T)}  \left( p^{\alpha}_0 +2 \sum_{t\in\mathbb Z, t \ge 1}   p_t^\alpha \rme^{|t|} \right) \\
    & \le   \frac{1}{C_\alpha(\sfT)}  \left( p_0^\alpha +2  \left( \frac{1}{\sqrt{2\pi}} \right)^\alpha\sum_{t\in\mathbb Z, t \ge 1}    \exp \left( - \frac{ \alpha (t -1/2)^2}{2} +t \right) \right) \\
    &<\infty
\end{align}
where the first inequality follows from~\eqref{Eq:bound_pt} and the second inequality follows since the dominant term in the exponent is $-t^2$, which ensures that the sum converges. 
This proves the upper bound in~\eqref{eq:RelationAlphaMIGaussianUniform}.

We now derive the {\em lower bound} in~\eqref{eq:RelationAlphaMIGaussianUniform}. From~\eqref{eq:{eq:kernel_mixture_ft}}, we have that
\begin{align}
\exp \left( \frac{\alpha-1}{\alpha} I_\alpha(\sfX;\sfY_\sfZ) \right )& \geq \int_{\mathbb R} \frac{1}{\varepsilon} \, \bbE^\frac{1}{\alpha} \left [ \left(  p_0 \, f_0 \left (\frac{y-\sfX}{\varepsilon} \right ) \right )^\alpha \right ]  \, \rmd y  \label{eq:PickOneT}
\\&= \left( 2\Phi \left(\tfrac12\right)-1 \right )  \int_{\mathbb R}
\bbE^\frac{1}{\alpha}[f^\alpha_{\sfY_0|\sfX}(y|\sfX)] \, \rmd y \label{eq:MIWithYBijectiveMeasurableTransformation}
\\& \geq  \left( 2\Phi \left(\tfrac12\right)-1 \right ) \frac{1}{\varepsilon} \int_{\mathbb R}
\bbE^\frac{1}{\alpha} \left[
 \mathbbm 1_{\left \{|y-\sfX|\le \frac{\varepsilon}{2} \right \}}
\right] \, \rmd y \label{eq:BoundsWu}
\\&= \left( 2\Phi \left(\tfrac12\right)-1 \right ) \frac{1}{\varepsilon} \int_{\mathbb R}
P_\sfX^{\frac{1}{\alpha}} \left(\cB_\infty \left (y,\frac{\varepsilon}{2} \right ) \right) \, \rmd y 
\\& = \left( 2\Phi \left(\tfrac12\right)-1 \right ) \exp \left( \frac{\alpha-1}{\alpha} I_\alpha(\sfX;\sfY_{\varepsilon}) \right ), \label{eq:UnifBoundLemma}
\end{align}
where:~\eqref{eq:PickOneT} follows by just considering $t =0$ for which $ p_0=
2\Phi \left(\tfrac12\right)-1$;
\eqref{eq:MIWithYBijectiveMeasurableTransformation} and~\eqref{eq:BoundsWu} follow from~\eqref{eq:kernel_lower_pointwise1}; and~\eqref{eq:UnifBoundLemma} follows from Lemma~\ref{lemma:UniformMIProb}. This proves the lower bound in~\eqref{eq:RelationAlphaMIGaussianUniform} and concludes the proof of Lemma~\ref{lemma:RelationAlphaMIGaussianUniform}.
}

\section{Proof of Theorem~\ref{thm:FinitenessAlphaMI}}
\label{app:FinitenessAlphaMI}
To see that $\alpha$-mutual information is always finite for $\alpha \in  (0,1)$, we claim that for every distribution on $\sfX$ and every ${\rm snr}>0$
\begin{equation}
0<  \|  f_\sfY(\cdot; \alpha \; {\rm snr})  \|_{ \frac{1}{\alpha}} <\infty,
\end{equation}
which together with~\eqref{eq:norm_representation} implies that $I_{\alpha}(\sfX;\mathrm{snr})<\infty$. The left inequality follows trivially since $f_\sfY$ is a proper density for every $\sfX$ and cannot have a zero norm. 
To see the right inequality, note that
\begin{align}
     \|  f_\sfY(\cdot; \alpha \; {\rm snr})  \|_{ \frac{1}{\alpha}}^{ \frac{1}{\alpha}} & = \int_{\bbR} f_\sfY^{ \frac{1}{\alpha}}(t; \alpha \; {\rm snr})  \; \rmd t \\
     &=\int_{\bbR} \bbE^{ \frac{1}{\alpha} } \left[     \sqrt{\rm  \alpha \, snr} \; \phi \left (\sqrt{\rm \alpha \, snr} \left( y-  \sfX  \right ) \right )  \right] \rmd y \\
     & \le \int_{\bbR} \bbE \left[     (\sqrt{\rm  \alpha \, snr})^{ \frac{1}{\alpha} } \; \phi \left ( \sqrt{\rm   snr} \left( y-  \sfX  \right ) \right )  \right] \rmd y  \label{eq:Jensen's_inequlaity_for_finitness} \\
      & =  \bbE \left[     (\sqrt{\rm  \alpha \, snr})^{ \frac{1}{\alpha} } \; \int_{\bbR} \phi \left ( \sqrt{\rm   snr} \left( y-  \sfX  \right ) \right ) \rmd y \right]  \label{eq:fub_ton_for_finitness}\\
       & =     (\sqrt{\rm  \alpha \, snr})^{ \frac{1}{\alpha} }  \frac{1}{\sqrt{\rm snr}} <\infty,
\end{align}
where~\eqref{eq:Jensen's_inequlaity_for_finitness} follows from Jensen's inequality since $x \mapsto x^{\frac{1}{\alpha}}$ is convex on $(0,+\infty)$ for $\alpha \in (0,1)$; and~\eqref{eq:fub_ton_for_finitness} follows from Fubini-Tonelli's theorem.   
The fact that $0 \leq d_{\frac{1}{\alpha}}(\sfX) \leq 1$ follows from~\eqref{eq:def_dalpha}, Lemma~\ref{lemma:UBLBHalpha}, and the fact that, for $0<\alpha<1$, we have that $H_{\frac{1}{\alpha}}(\lfloor \sfX \rfloor ) < \infty$ (see Remark~\ref{rem:HalphaFinitealphaGreater1}).
This concludes the proof of the first case. 

For Case 2, from Lemma~\ref{lemma:GaussianAlphaMI}, we have the following trivial equivalences ${\rm (a)} \Leftrightarrow {\rm (c)}  \Leftrightarrow {\rm (e)} $ and ${\rm (b)} \Leftrightarrow {\rm (d)}  \Leftrightarrow {\rm (f)}$. We now show that ${\rm (c) \Leftrightarrow (d)}$.
The implication ${\rm (c) \Leftarrow (d)}$ is trivial. We therefore focus on proving ${\rm (c) \Rightarrow (d)}$. Fix some $0<{\rm snr}_0 <\infty$ for which $ \|  f_\sfY(\cdot; \alpha \; {\rm snr}_0)  \|_{ \frac{1}{\alpha}} <\infty$. We study two cases separately, namely ${\rm snr}<{\rm snr}_0$ and ${\rm snr}>{\rm snr}_0$.
For every $0< {\rm snr} < {\rm snr}_0$, by the data-processing inequality for $\alpha$-mutual information, we have that  
\begin{equation}
  I_{\alpha}(\sfX;\mathrm{snr})  <  I_{\alpha}(\sfX;\mathrm{snr}_0)<\infty ,
\end{equation}
which by \eqref{eq:norm_representation} and the monotonicity of the logarithm implies that
\begin{equation}
\|  f_\sfY(\cdot; \alpha \; {\rm snr})  \|_{ \frac{1}{\alpha}}< \infty.
\end{equation}
Now, for the case ${\rm snr}>{\rm snr}_0$, we have that
\begin{align}
\|  f_\sfY(\cdot; \alpha \; {\rm snr})  \|_{ \frac{1}{\alpha}}^{\frac{1}{\alpha}}  &=\int_{\bbR} \bbE^{ \frac{1}{\alpha} } \left[     \sqrt{\rm  \alpha \, snr} \; \frac{1}{\sqrt{2 \pi}} \exp \left (- \frac{{\rm \alpha \, snr} \left( y-  \sfX  \right )^2}{2} \right )  \right] \rmd y \\
& \le   \int_{\bbR} \bbE^{ \frac{1}{\alpha} } \left[     \sqrt{\rm  \alpha \, snr} \; \frac{1}{\sqrt{2 \pi}} \exp \left (- \frac{ {\rm \alpha \, snr_0} \left( y-  \sfX  \right )^2}{2} \right ) \right] \rmd y \label{eq:montonicity_of_exp} \\
& =  \left( \frac{\rm snr}{\rm snr_0} \right)^{ \frac{1}{2 \alpha} } \int_{\bbR} \bbE^{ \frac{1}{\alpha} } \left[     \sqrt{\rm  \alpha \, snr_0} \; \frac{1}{\sqrt{2 \pi}} \exp \left (- \frac{{\rm \alpha \, snr_0} \left( y-  \sfX  \right )^2}{2} \right ) \right] \rmd y \\
& =  \left( \frac{\rm snr}{\rm snr_0} \right)^{ \frac{1}{2 \alpha} } \|  f_\sfY(\cdot; \alpha \; {\rm snr}_0)  \|_{ \frac{1}{\alpha}}^{\frac{1}{\alpha}}<\infty,
\end{align}
where~\eqref{eq:montonicity_of_exp} follows from  the fact that $\exp \left (- \frac{ {\rm \alpha \, snr} \left( y-  \sfX  \right )^2}{2} \right ) < \exp \left (- \frac{ {\rm \alpha \, snr_0} \left( y-  \sfX  \right )^2}{2} \right ) $  for ${\rm snr} > {\rm snr}_0$. 

In the above, we have shown that if $\|  f_\sfY(\cdot; \alpha \; {\rm snr}_0)  \|_{ \frac{1}{\alpha}}<\infty$ for some ${\rm snr}_0$, then  $\|  f_\sfY(\cdot; \alpha \; {\rm snr})  \|_{ \frac{1}{\alpha}}<\infty$ for every ${\rm snr}$, which in other words means that  ${\rm (c) \Rightarrow (d)}$.
Consequently, we have the following chain of implications: ${\rm (a)} \Leftrightarrow {\rm (c)}  \Leftrightarrow {\rm (e)} \Leftrightarrow {\rm (b)} \Leftrightarrow {\rm (d)}  \Leftrightarrow {\rm (f)} $. 

We now show the equivalence to ${\rm (g)}$. We start by showing that ${\rm (g) }\Rightarrow {\rm (a)}$ where, without loss of generality, we assume ${\rm snr} =1/\alpha$.
We let  $\sfX = \sfK + \sfU$ where $\sfK =\lfloor \sfX \rfloor \in \bbZ$ is the integer part and $\sfU \in [0,1)$ is the fractional part.  Next, from~\eqref{eq:fYsnr}, note that 
\begin{equation}
    \label{eq:fysumk}
f_\sfY (y; \alpha \, {\rm{snr}}) = f_\sfY (y; 1) =  \sum_{k \in \bbZ} P_\sfK(k) \, g_k (y-k),
\end{equation}
where
\begin{equation}
    g_k (t)= \int_{0}^1 \phi(t-u) \, \rmd P_{\sfU|\sfK} (u|k).
    \label{eq:gkt}
\end{equation}
Now, note that for every $t \in \bbR$, we have that
\begin{align}
    g_k (t)&= \int_{0}^1 \phi(t-u) \, \rmd P_{\sfU|\sfK} (u|k) \\
    &=  \frac{1}{\sqrt{2 \pi}} \int_{0}^1 \rme^{- \frac{(t-u)^2}{2} } \rmd P_{\sfU|\sfK} (u|k) \\
    &\le   \frac{1}{\sqrt{2 \pi}} \int_{0}^1 \rme^{- \frac{ \frac{1}{2} t^2 -u^2 }{2} } \, \rmd P_{\sfU|\sfK} (u|k) \\
    & = \frac{1}{\sqrt{2 \pi}} \rme^{- \frac{  t^2  }{4} } \int_{0}^1 \rme^{\frac{u^2 }{2} } \, \rmd P_{\sfU|\sfK} (u|k) \\
    &\le   \frac{1}{\sqrt{2 \pi}} \rme^{- \frac{  t^2  }{4} } \rme^{\frac{1}{2}}, \label{eq:using_boundness_of_U}
\end{align}
where we have used the inequality $(a-b)^2 \geq \frac{1}{2} a^2 - b^2$ and~\eqref{eq:using_boundness_of_U} follows from the assumption that $0 \leq \sfU < 1$.

Next, recall the inequality $\left( \sum_i a_i  \right)^{ \frac{1}{\alpha} } \le \left( \sum_i a_i^{ \frac{1}{\alpha} }  \right) $ where $a_i>0$ and $\alpha>1$, which used in~\eqref{eq:fysumk} implies that
\begin{align}
   \int_{-\infty}^\infty f_\sfY^{ \frac{1}{\alpha} }   (y; 1)  \, \rmd y &\le  \int_{-\infty}^\infty \sum_{k \in \bbZ} P_\sfK^{ \frac{1}{\alpha} }(k) \, g_k^{ \frac{1}{\alpha} } (y-k)  \, \rmd y \\
   &=  \sum_{k \in \bbZ} P_\sfK^{ \frac{1}{\alpha} }(k) \int_{-\infty}^\infty  g_k^{ \frac{1}{\alpha} } (y-k)  \, \rmd y \label{eq:FubiniTonelliReyniEntropy}\\
   &=  \sum_{k \in \bbZ} P_\sfK^{ \frac{1}{\alpha} }(k) \int_{-\infty}^\infty  g_k^{ \frac{1}{\alpha} } (y)  \, \rmd y    \label{eq:change_of_variable}\\
    & \le   \sum_{k \in \bbZ} P_\sfK^{ \frac{1}{\alpha} }(k) \int_{-\infty}^\infty  \left(  \frac{1}{\sqrt{2 \pi}} \rme^{- \frac{  y^2  }{4} } \rme^{\frac{1}{2}}  \right)^{  \frac{1}{\alpha} }   \rmd y \label{eq:using_gaussian_bound}\\
    &= C_\alpha  \sum_{k \in \bbZ} P_\sfK^{ \frac{1}{\alpha} }(k),
\end{align}
where~\eqref{eq:FubiniTonelliReyniEntropy} follows from Fubini-Tonelli's theorem;~\eqref{eq:change_of_variable} is due to a  change of variable; 
and \eqref{eq:using_gaussian_bound} follows from the bound in \eqref{eq:using_boundness_of_U}. In the last step $C_\alpha$ is a quantity that only depends on $\alpha$.

Combining everything and using the representation of $\alpha$-mutual information in~\eqref{eq:AlphaMISecondRepresentation} and Definition~\ref{def:ReyniEntropy}, we arrive at the conclusion that 
\begin{align}
I_\alpha(\sfX;1/\alpha) &= \frac{\alpha}{\alpha-1} \log \left( \int_{\bbR} f_\sfY^{ \frac{1}{\alpha}}(t; 1)  \; \rmd t  \right)  + \Delta(\alpha, 1/\alpha)
\\& \le  \frac{\alpha}{\alpha-1} \log(C_{\alpha}) + H_{\frac{1}{\alpha}}(\sfK)  + \Delta(\alpha, 1/\alpha) 
\\& = \frac{\alpha}{\alpha-1} \log(C_{\alpha}) + H_{\frac{1}{\alpha}}(\lfloor \sfX \rfloor) + \Delta(\alpha, 1/\alpha).
\end{align}
This shows that ${\rm (g) }\Rightarrow {\rm (a)}$.  We now show that ${\rm (g) }\Leftarrow {\rm (a)}$. We start by noting that, by using the inequality $(a-b)^2 \le 2 (a^2+b^2)$ and following very similar steps as in~\eqref{eq:using_boundness_of_U}, one can lower bound $g_k (t)$ in~\eqref{eq:gkt} as follows,
\begin{equation}
     g_k (t) \ge   \frac{1}{\sqrt{2 \pi}} \rme^{- t^2 } \rme^{-1}, \, t\in \bbR.  \label{eq:lower_bound_on_g_k}
\end{equation}
Let $\cI_k = [k, k+1]$, where $k \in \bbZ$, and note that
\begin{align}
    \int_{-\infty}^\infty f_\sfY^{ \frac{1}{\alpha} }   (y; 1)  \, \rmd y &= \sum_{k \in \bbZ}  \int_{\cI_k} f_\sfY^{ \frac{1}{\alpha} }   (y; 1)  \, \rmd y\\
    & \ge  \sum_{k \in \bbZ}  \int_{\cI_k} \left( P_\sfK(k) \, g_k(y-k) \right)^{ \frac{1}{\alpha} }   \, \rmd y \label{eq:Uisng_lower_bound_on_g_k} \\
    & \ge  \sum_{k \in \bbZ}  \int_{I_k} \left( P_\sfK(k)  \frac{1}{\sqrt{2 \pi}} \rme^{- (y-k)^2 } \rme^{-1} \right)^{ \frac{1}{\alpha} }   \, \rmd y \label{eq:lower_bound_sum_by_one_term}  \\
     & =  \sum_{k \in \bbZ} P_\sfK^{ \frac{1}{\alpha} }(k)  \int_{\cI_k} \left(  \frac{1}{\sqrt{2 \pi}} \rme^{- (y-k)^2 } \rme^{-1} \right)^{ \frac{1}{\alpha} }   \, \rmd y \label{eq:FubiniTonelliReyniEntropy2} \\
     & \ge  \sum_{k \in \bbZ} P_\sfK^{ \frac{1}{\alpha} }(k)   \left(  \frac{1}{\sqrt{2 \pi}} \rme^{- 1 } \rme^{-1} \right)^{ \frac{1}{\alpha} }  , \label{eq:last_bound_in_series} 
\end{align}
where~\eqref{eq:Uisng_lower_bound_on_g_k} follows from~\eqref{eq:fysumk} using the fact that $ f_\sfY (y; 1) =  \sum_{k \in \bbZ} P_\sfK(k) \, g_k (y-k) \ge   P_\sfK(k) \, g_k (y-k)$;~\eqref{eq:lower_bound_sum_by_one_term} follows from using the bound in~\eqref{eq:lower_bound_on_g_k}; 
and~\eqref{eq:last_bound_in_series} follows by noting that  $\rme^{- (y-k)^2 } \ge \rme^{-1}$ for $y \in I_k$. 

Combining the bound in~\eqref{eq:last_bound_in_series} with the representation of $\alpha$-mutual information in~\eqref{eq:AlphaMISecondRepresentation}, we arrive at
\begin{align}
I_\alpha(\sfX;1/\alpha) &= \frac{\alpha}{\alpha-1} \log \left( \int_{\bbR} f_\sfY^{ \frac{1}{\alpha}}(t; 1)  \; \rmd t  \right)  + \Delta(\alpha, 1/\alpha)\\
& \ge \frac{1}{\alpha-1}\log \left (  \frac{1}{\sqrt{2 \pi}} \rme^{- 1 } \rme^{-1} \right) +  H_{\frac{1}{\alpha}}(\lfloor \sfX \rfloor) + \Delta(\alpha, 1/\alpha),
\end{align}
which implies ${\rm (g) }\Leftarrow {\rm (a)}$.

It remains to show the equivalence in ${\rm (h) }$. The fact that ${\rm (g) }\Rightarrow {\rm (h)}$ follows from~\eqref{eq:def_dalpha} and Lemma~\ref{lemma:UBLBHalpha}. To prove ${\rm (h) }\Rightarrow {\rm (g)}$, we prove the contrapositive, i.e., we show that whenever ${\rm (g)}$ is not satisfied, then ${\rm (h)}$ is also not satisfied. Assume that $H_{\frac{1}{\alpha}}(\lfloor \sfX \rfloor) = \infty$. Then, from Lemma~\ref{lemma:UBLBHalpha}, we have that $H_{\alpha}(\langle\sfX\rangle_m) = \infty$ for all $m \in \bbN$ and ${\rm (h)}$ fails. Hence, ${\rm (h) }\Rightarrow {\rm (g)}$.
This concludes the proof of Theorem~\ref{thm:FinitenessAlphaMI}.

\section{Proof of Proposition~\ref{prop:LimitSNRToZeroAlphaMI}}
\label{app:LimitSNRToZeroAlphaMI}
We use the equivalent channel model $\sfY = \sqrt{{\rm{snr}}} \, \sfX +\sfZ$, for which
\begin{align}
I_\alpha(\sfX;\mathrm{snr})&=\frac{\alpha}{\alpha-1} \log \left(   \int_\bbR \bbE^{ \frac{1}{\alpha} } \left[  \phi^\alpha(y- \sqrt{\rm snr} \, \sfX ) \right] \rmd y  \right). \label{eq:AlphaMIEquivalentChannel}
\end{align}
The case of $\alpha=1$ was shown in~\cite{GuoIT2005,FunctionalMMSE}. We now consider the cases of $\alpha<1$ and $\alpha>1$ separately.
\begin{enumerate}
\item {\bf{Case~1:} $\alpha \in (1,+\infty)$.}
Let $\sfX_M= \sfX \; \mathbbm 1_{\{|\sfX| \le M\}}$. In addition, let 
\begin{align}
    C({\rm snr};M) &= \int_\bbR \bbE^{\frac{1}{\alpha}} \left[  \phi^\alpha(y- \sqrt{\rm snr} \; \sfX ) \mathbbm 1_{\{|\sfX|>M \}} \right] \,  \rmd y \\
    &= \int_\bbR \bbE^{\frac{1}{\alpha}} \left[  \phi^{\alpha}(y- \sqrt{\rm snr} \; \sfX )  \Big| |\sfX|>M   \right]  \bbP^{\frac{1}{\alpha}}[ |\sfX|>M]  \,  \rmd y  \\
    &= \rme^{ \frac{\alpha-1}{\alpha} I_\alpha \left (\sfX^{M};{\rm snr} \right )  +\frac{1}{\alpha}\log(\bbP[ |\sfX|>M]  ) }, 
\end{align}
where $\sfX^M$ is a random variable with distribution $P_{\sfX | |\sfX|>M}$.

Now, without loss of generality, assume that ${\rm snr} \leq \delta$, for some $\delta>0$ such that $I_{\alpha}(\sfX;\delta)<\infty$. With this, we can write
\begin{align}
   & \int_\bbR \left |\bbE^{ \frac{1}{\alpha} } \left[  \phi^\alpha(y- \sqrt{\rm snr} \; \sfX ) \right] -  \bbE^{ \frac{1}{\alpha} } \left[  \phi^\alpha(y- \sqrt{\rm snr} \; \sfX_M ) \right] \right |  \rmd y  \label{eq:diff_equations_for_snr+0}\\
    &\le \int_{\bbR} \left |\bbE \left[  \phi^\alpha(y- \sqrt{\rm snr} \; \sfX ) \right] -  \bbE \left[  \phi^\alpha(y- \sqrt{\rm snr} \; \sfX_M ) \right] \right |^{\frac{1}{\alpha}}  \rmd y  \label{eq:First bound_via_riangle} \\
    &= \int_\bbR \left |\bbE \left[  \phi^\alpha(y- \sqrt{\rm snr} \; \sfX )- \phi^\alpha(y- \sqrt{\rm snr} \; \sfX_M ) \right] \right |^{\frac{1}{\alpha}}  \rmd y\\
    &= \int_\bbR \left |\bbE \left[  \phi^\alpha(y- \sqrt{\rm snr} \; \sfX ) \mathbbm 1_{\{|\sfX|>M\}}- \phi^\alpha(y- \sqrt{\rm snr} \; \sfX_M ) \mathbbm 1_{\{|\sfX|>M\}} \right] \right |^{\frac{1}{\alpha}}  \rmd y \\
    &= \int_\bbR \left |\bbE \left[  \phi^\alpha(y- \sqrt{\rm snr} \; \sfX ) \mathbbm 1_{\{|\sfX|>M\}} \right]- \phi^\alpha(y) \; \bbE\left[ \mathbbm 1_{\{|\sfX|>M\}} \right] \right |^{\frac{1}{\alpha}}  \rmd y \\
    &\le \int_\bbR \bbE^{\frac{1}{\alpha}} \left[  \phi^\alpha(y- \sqrt{\rm snr} \; \sfX ) \mathbbm 1_{\{|\sfX|>M\}} \right]  + \phi(y) \, \bbE^{\frac{1}{\alpha}}\left[ \mathbbm 1_{\{|\sfX|>M\}} \right]   \rmd y \label{eq:Bound_on_triangle_for_p>1}\\
    & =   C({\rm snr};M) + \bbE^{\frac{1}{\alpha}}\left[ \mathbbm 1_{\{|\sfX|>M\}} \right]  \\
    & \le   C(\delta;M) + \bbE^{\frac{1}{\alpha}}\left[ \mathbbm 1_{\{|\sfX|>M\}} \right]  ,\label{eq:C_M_definition}
\end{align}
where in~\eqref{eq:First bound_via_riangle} we have used the fact that $| a^{\frac{1}{\alpha}} - b^{\frac{1}{\alpha}}| \le  | a-b|^{\frac{1}{\alpha}}, \alpha>1, a>0,b>0 $;~\eqref{eq:Bound_on_triangle_for_p>1} follows from the inequality $| a-b|^p \le |a|^p + |b|^p, \, p<1$; and~\eqref{eq:C_M_definition} follows from the fact that, since $ {\rm snr} \mapsto I_\alpha(\sfX^{M};{\rm snr})$ is non-decreasing (from the data processing theorem~\cite[Theorem~2]{verdu2015alpha}), so is ${\rm snr} \mapsto  C({\rm snr};M)$ for $\alpha>1$.

Thus, for every $M>0$, we arrive at
\begin{align}
\lim_{{\rm{snr}} \to 0} I_{\alpha}(\sfX;\,\mathrm{snr}) &= \frac{\alpha}{\alpha-1}  \log \left(  \lim_{{\rm{snr}} \to 0} \int_\bbR \bbE^{ \frac{1}{\alpha} } \left[  \phi^\alpha(y- \sqrt{\rm snr} \; \sfX ) \right] \rmd y  \right)
\\& \leq \frac{\alpha}{\alpha-1} \!\log \left( \lim_{{\rm{snr}} \to 0}\int_{\bbR} \bbE^{ \frac{1}{\alpha} } \left[  \phi^\alpha(y\!-\! \sqrt{\rm snr} \; \sfX_M ) \right]  \rmd y  \!+\!   C(\delta;M)  + \bbE^{\frac{1}{\alpha}}\left[ \mathbbm 1_{\{|\sfX|>M \}} \right]\right )
\\& = \frac{\alpha}{\alpha-1} \log \left( \int_{\bbR} \lim_{{\rm{snr}} \to 0} \bbE^{ \frac{1}{\alpha} } \left[  \phi^\alpha(y\!-\! \sqrt{\rm snr} \; \sfX_M ) \right]  \rmd y +  C(\delta;M) +\! \bbE^{\frac{1}{\alpha}}\left[ \mathbbm 1_{\{|\sfX|>M\}} \right]\right ) \label{eq:DCTAlphaGreat1MISNRtoZero}
\\& = \frac{\alpha}{\alpha-1} \log \left(\int_{\bbR} \phi(y) \; \rmd y +   C(\delta;M) + \bbE^{\frac{1}{\alpha}}\left[ \mathbbm 1_{\{|\sfX|>M\}} \right] \right ), \label{eq:DCTAlphaGreat1MISNRtoZeroFinal}
\\& = \frac{\alpha}{\alpha-1} \log \left(1 +   C(\delta;M) + \bbE^{\frac{1}{\alpha}}\left[ \mathbbm 1_{\{|\sfX|>M\}} \right] \right ), 
\end{align}
where~\eqref{eq:DCTAlphaGreat1MISNRtoZero} follows by using the dominated convergence theorem, which holds since
\begin{align}
\bbE^{\frac{1}{\alpha}} \left [ \phi^\alpha(y- \sqrt{\rm snr} \; \sfX_M ) \right ] &= \bbE^{\frac{1}{\alpha}} \left [ \left( \frac{1}{2 \, \pi} \right )^{\frac{\alpha}{2}} \exp \left( -\frac{\alpha}{2} \left( y -\sqrt{{\rm{snr}}} \; \sfX_M \right )^2\right )  \right ] 
\\& = \bbE^{\frac{1}{\alpha}} \left [ \left( \frac{1}{2 \, \pi} \right )^{\frac{\alpha}{2}} \exp \left(-\frac{\alpha}{2} y^2 \right ) \exp \left( \alpha \, \sqrt{{\rm{snr}}} \, \sfX_M \, y\right ) \exp \left(-\frac{\alpha}{2} {\rm{snr}} \, \sfX_M^2 \right ) \right ]
\\& \leq  \frac{1}{\sqrt{2 \, \pi}} \exp \left(-\frac{1}{2} y^2 \right ) \exp \left(   M \, \sqrt{\delta} \, |y|\right ) \label{eq:FurtherBoundforDCT}
\\& = \frac{1}{\sqrt{2 \, \pi}} \exp \left( \frac{M^2 \delta}{2} \right ) \exp \left( -\frac{1}{2} \left( |y|-M \sqrt{\delta}\right )^2\right ),
\label{eq:IntegrableFunctionLimitSNRtozeroMI}
\end{align}
where~\eqref{eq:FurtherBoundforDCT} follows since $\sfX_M \leq |\sfX_M|\leq M$, ${\rm{snr}} \leq \delta$, and by bounding the last exponential term by one. Note that~\eqref{eq:IntegrableFunctionLimitSNRtozeroMI} is an integrable function.
It now remains to show that $  C(\delta;M) $ converges to zero as $M \to \infty$. This follows since 
\begin{align}
 \lim_{M \to \infty}   C(\delta;M) &=  \int_\bbR  \lim_{M \to \infty}  \bbE^{\frac{1}{\alpha}} \left[  \phi^\alpha(y- \sqrt{\delta} \; \sfX ) \mathbbm 1_{\{|\sfX|>M\}} \right] \, \rmd y\label{eq:DCT_v1}\\
&= \int_\bbR    \bbE^{\frac{1}{\alpha}} \left[  \lim_{M \to \infty} \phi^\alpha(y- \sqrt{\delta} \; \sfX ) \mathbbm 1_{\{|\sfX|>M\}} \right ] \, \rmd y \\
&= 0, \label{eq:limit_M_C}
\end{align}
where~\eqref{eq:DCT_v1} follows from using the dominated convergence theorem with 
\begin{equation}
    \int_\bbR    \bbE^{\frac{1}{\alpha}} \left[  \phi^\alpha(y- \sqrt{\delta} \; \sfX ) \mathbbm 1_{\{|\sfX|>M\}} \right] \, \rmd y \le \int_\bbR    \bbE^{\frac{1}{\alpha}} \left[   \phi^\alpha(y- \sqrt{\delta} \; \sfX ) \right] \, \rmd y =  \rme^{ \frac{\alpha-1}{\alpha} I_\alpha(\sfX;\delta)} <\infty,
\end{equation}
where we have used the assumption that $I_\alpha(\sfX;\delta)<\infty$ for some $\delta>0$. This concludes the proof for $\alpha>1$.
\item {\bf{Case~2:} $\alpha \in (0,1)$.}\footnote{For this case the assumption that $I_\alpha(\sfX;\delta)<\infty$ is redundant.}From the definition in~\eqref{eq:Def_I}, we have that evaluating $D_{\alpha}\left( P_{\sfX,\sfY} \| P_{\sfX} \times Q_{\sfY} \right )$ in any $Q_{\sfY}$ provides an upper bound on $I_{\alpha}(\sfX;\,\mathrm{snr})$. We choose $\rmd Q_\sfY = \phi$. With this, we obtain
\begin{align}
I_{\alpha}(\sfX;\,\mathrm{snr}) \leq \frac{1}{\alpha-1} \log \left( \int_\bbR   \bbE [  \phi^\alpha (y -\sqrt{{\rm{snr}}} \; \sfX) ]   \; ( \phi(y) )^{1-\alpha} \; \rmd y \right ),
\end{align}
which leads to
\begin{align}
\lim_{{\rm{snr}} \to 0} I_{\alpha}(\sfX;\,\mathrm{snr}) &\leq \lim_{{\rm{snr}} \to 0} \frac{1}{\alpha-1} \log \left( \int_\bbR   \bbE [  \phi^\alpha (y -\sqrt{{\rm{snr}}} \; \sfX) ]   \; ( \phi(y) )^{1-\alpha} \; \rmd y \right )
\\& = \frac{1}{\alpha-1} \log \left( \lim_{{\rm{snr}} \to 0} \int_\bbR   \bbE [  \phi^\alpha (y -\sqrt{{\rm{snr}}} \; \sfX) ]   \; ( \phi(y) )^{1-\alpha} \; \rmd y \right )
\\& = \frac{1}{\alpha-1} \log \left(  \int_\bbR  \lim_{{\rm{snr}} \to 0} \bbE [  \phi^\alpha (y -\sqrt{{\rm{snr}}} \; \sfX) ]   \; ( \phi(y) )^{1-\alpha} \; \rmd y \right ) \label{eq:DCTalphaMISNR0}
\\& = \frac{1}{\alpha-1} \log \left(  \int_\bbR \phi(y) \; \rmd y \right )
\\& = 0,
\end{align}
where~\eqref{eq:DCTalphaMISNR0} follows by using the dominated convergence theorem, which holds since 
\begin{align}
\bbE [  \phi^\alpha (y -\sqrt{{\rm{snr}}} \; \sfX) ]   \; ( \phi(y) )^{1-\alpha} \leq \left( \frac{1}{ \sqrt{2 \, \pi}} \right )^\alpha ( \phi(y) )^{1-\alpha},
\end{align}
and since $( \phi(y) )^{1-\alpha}$ is integrable for $\alpha \in (0,1)$. This concludes the proof for $\alpha<1$ and of Proposition~\ref{prop:LimitSNRToZeroAlphaMI}.
\end{enumerate}

\section{Proof of Theorem~\ref{thm:Continuity}}
\label{app:Continuity}
First, without loss of generality, we assume ${\rm snr}=1$.
By~\eqref{eq:AlphaMISecondRepresentation}, the continuity of $I_{\alpha}(\sfX; 1)$ depends solely on the convergence 
\begin{equation}
    \int_{\bbR} f_{\sfY_n}^{ \frac{1}{\alpha}}(y;\alpha )\, \rmd y
\to
\int_{\bbR} f_\sfY^{ \frac{1}{\alpha}}(y;\alpha )\, \rmd y,
\end{equation}
 where from~\eqref{eq:fYsnr} we recall that 
\begin{equation}
\label{eq:fYnProof}
    f_{\sfY_n}(y; \alpha ) =\bbE[ \sqrt{\alpha} \, \phi( \sqrt{\alpha} \, (y-\sfX_n))].
\end{equation}
By standard arguments via the Portmanteau theorem~\cite{dudley2018real}, since $x \mapsto \phi( \sqrt{\alpha} (y-x))$ is bounded and continuous for fixed $y$, the convergence of $\sfX_n \to \sfX$ in distribution implies that 
\begin{align}
\label{eq:PointWisefYntofY}
    \lim_{n \to \infty} f_{\sfY_n}(y; \alpha ) = f_{\sfY}(y; \alpha ) \quad \forall y \in \mathbb{R}.
\end{align}
Now that we have established pointwise convergence, we need to upgrade it to convergence in the norm. We consider the cases $\alpha \in (0,1)$ and $\alpha \in (1, \infty) $ separately.
\begin{itemize}
    \item Let $\alpha \in (0,1)$. Fix some $A>0$ and write:
\begin{equation}
    \label{decomp}
    \int_{\mathbb{R}} f_{\sfY_n}^{ \frac{1}{\alpha}}(y; \alpha) \, \rmd y = \int_{|y| \le A} f_{\sfY_n}^{ \frac{1}{\alpha}}(y; \alpha) \, \rmd y + \int_{|y| > A} f_{\sfY_n}^{ \frac{1}{\alpha}}(y; \alpha) \, \rmd y.
\end{equation}
We now consider each of the above two integrals separately.  
For the first integral, we apply the dominated convergence theorem, i.e.,
\begin{equation}
    \lim_{n \to \infty} \int_{|y| \le A} f_{\sfY_n}^{ \frac{1}{\alpha}}(y; \alpha) \, \rmd y = \int_{|y| \le A} f_{\sfY}^{ \frac{1}{\alpha}}(y; \alpha) \, \rmd y, 
    \label{eq:bounded_part_alpha<1}
\end{equation}
which can be applied since $f_{\sfY_n}(y; \alpha) \to f_{\sfY}(y; \alpha)$ pointwise from~\eqref{eq:PointWisefYntofY} and since, from~\eqref{eq:fYnProof}, we have that $f_{\sfY_n}(y;\alpha) \leq \sqrt{\frac{\alpha}{2 \pi}} \leq 1$, which implies that for all $n$ and $y$, we have that 
\begin{equation}
0\le f_{\sfY_n}^{ \frac{1}{\alpha}}(y;\alpha) \mathbbm 1_{\{|y|\le A\}}\le \mathbbm 1_{\{|y|\le A\}}. 
\end{equation}
For the second integral, first note that  by using that  $t\mapsto t^{ \frac{1}{\alpha}}$ is convex,  Jensen's inequality implies that, for all $y\in\bbR$, we can write
\begin{align}
f_{\sfY_n}^{ \frac{1}{\alpha}}(y;\alpha)
=
\bbE^{ \frac{1}{\alpha}} \left[ \sqrt{\alpha} \, \phi( \sqrt{\alpha}(y-\sfX_n)) \right]
\le
\bbE\left[  \alpha^{  \frac{1}{2 \alpha}} {\left( \frac{1}{\sqrt{2\pi}} \right)^{\frac{1}{\alpha}-1}}\phi(y-\sfX_n)\right]  = C_\alpha  f_{\sfY_n}(y; \alpha=1),  
\end{align}
where  $C_\alpha= \alpha^{  \frac{1}{2 \alpha}} \left( \frac{1}{\sqrt{2\pi}} \right)^{\frac{1}{\alpha}-1}$
and, therefore, we arrive at
\begin{align}
\int_{|y|>A} f_{\sfY_n}^{ \frac{1}{\alpha}}(y;\alpha)\,\rmd y
&\le C_\alpha  \int_{|y|>A} f_{ {\sfY_n}}(y; \alpha=1) \, \rmd y\\
&=  C_\alpha \, \bbP [ |\sfY_n|>A ; \alpha=1].
\label{eq:tail_alpha<1}
\end{align}
Consequently, we have that
\begin{align}
\lim_{n\to\infty}\int_{|y|>A} f_{\sfY_n}^{ \frac{1}{\alpha}}(y;\alpha)\,\rmd y
\le C_\alpha\lim_{n\to\infty}\bbP(|\sfY_n|>A; \alpha=1)
= C_\alpha \, \bbP(|\sfY|>A ;\alpha=1), \label{eq;Tail_bound_A}
\end{align}
where the last equality follows since the pointwise convergence in~\eqref{eq:PointWisefYntofY} implies convergence in distribution.
Now, combining everything together, we arrive at 
\begin{align}
 \lim_{n\to\infty}\int_{\bbR} f_{\sfY_n}^{ \frac{1}{\alpha}}(y;\alpha)\,\rmd y &= \lim_{A \to \infty} \lim_{n\to\infty} 
\int_{|y|\le A} f_{\sfY_n}^{ \frac{1}{\alpha}}(y;\alpha)\,\rmd y
+ \lim_{A \to \infty} \lim_{n\to\infty} \int_{|y|>A} f_{\sfY_n}^{ \frac{1}{\alpha}}(y;\alpha)\,\rmd y\\
&= \lim_{A \to \infty} \lim_{n\to\infty} 
\int_{|y|\le A} f_{\sfY_n}^{ \frac{1}{\alpha}}(y;\alpha)\,\rmd y \label{eq:using_bound_tail_bound}  \\
&= \lim_{A \to \infty}  
\int_{|y|\le A} f_{\sfY}^{ \frac{1}{\alpha}}(y;\alpha)\,\rmd y \label{eq:continuity_trancated}  \\
&=  
\int_{\bbR} f_{\sfY}^{ \frac{1}{\alpha}}(y;\alpha)\,\rmd y  ,
\label{eq:limsup_chain_alpha<1}
\end{align}
where~\eqref{eq:using_bound_tail_bound} follows by using the tail bound in~\eqref{eq;Tail_bound_A};~\eqref{eq:continuity_trancated}  follows from the limit in~\eqref{eq:bounded_part_alpha<1}; and~\eqref{eq:limsup_chain_alpha<1} follows from the monotone convergence theorem.  This completes the proof of continuity for $\alpha \in (0,1)$.

\item Let $\alpha\in(1,\infty)$.  
Also for this case, we will apply the dominated convergence theorem to bring the limit inside the integral. 
Recall from~\eqref{eq:PointWisefYntofY} that $f_{\sfY_n}(y;\alpha)\to f_{\sfY}(y;\alpha )$ pointwise and hence, it remains to prove that $f^{ \frac{1}{\alpha}}_{\sfY_n}(y;\alpha)$ is dominated by an integrable function.
We start by noting that
\begin{align}
f_{\sfY_n}(y;\alpha)
&=\int_\bbR  \sqrt{\alpha} \, \phi \left(\sqrt{\alpha}(y-x) \right)\,\rmd P_{\sfX_n}(x)\\
&=\int_{|x|\le |y|/2} \sqrt{\alpha} \, \phi \left(\sqrt{\alpha}(y-x) \right)\,\rmd P_{\sfX_n}(x)
+\int_{|x|>|y|/2} \sqrt{\alpha} \, \phi \left(\sqrt{\alpha}(y-x) \right) \,\rmd P_{\sfX_n}(x)\\
&\le \max_{|x|\le |y|/2}\sqrt{\alpha} \, \phi \left (\sqrt{\alpha}(y-x) \right )+ \sqrt{\alpha} \, \bbP \left [|\sfX_n|>\frac{|y|}{2} \right ]\\
&\le \sqrt{\alpha} \, \phi \left (\sqrt{\alpha} \, \frac{|y|}{2} \right )+ \sqrt{\alpha}\,\bbP \left [|\sfX_n|>\frac{|y|}{2} \right ]
\label{eq:split_alpha>1} \\
&\le\sqrt{\alpha} \, \phi \left (\sqrt{\alpha} \, \frac{|y|}{2} \right )+ \sqrt{\alpha}  \min \left(1, \frac{\bbE[|\sfX_n|^k]}{(|y|/2)^k} \right) \label{eq:marko_use}
\\
&\le \sqrt{\alpha} \, \phi \left (\sqrt{\alpha} \, \frac{|y|}{2} \right )+ \sqrt{\alpha}  \min \left(1, \frac{C_k}{(|y|/2)^k} \right), \label{eq:LastStepCk}
\end{align}
where~\eqref{eq:split_alpha>1} follows from the fact that $|x|\le |y|/2$ implies $|y-x|\ge |y|/2$ and $\phi$ is decreasing in $|\cdot|$;~\eqref{eq:marko_use} follows from the Markov's inequality; and~\eqref{eq:LastStepCk} follows by letting $C_k = \sup_n\bbE[|\sfX_n|^k]$ with $C_k < \infty$.

Now, using the fact that $(a+b)^p\le a^p+b^p$ for $a,b\ge 0$ and $p \in (0,1)$, we arrive at 
\begin{equation}
f_{\sfY_n}^{ \frac{1}{\alpha}}(y;\alpha)
\le \alpha^{ \frac{1}{2\alpha}} \phi^{ \frac{1}{\alpha}}\left (\sqrt{\alpha}\frac{|y|}{2} \right )+ \alpha^{ \frac{1}{2\alpha}}\,\left( \min \left(1, \frac{C_k}{(|y|/2)^k} \right) \right )^{ \frac{1}{\alpha}},
\label{eq:fp_bd_alpha>1}
\end{equation}
which is integrable provided that $ \frac{k}{\alpha}>1$ or $k>\alpha$ as it is assumed.  
Thus, we can use the dominated convergence theorem and write
\begin{equation}
\lim_{n \to \infty} \int_{\bbR} f_{\sfY_n}^{ \frac{1}{\alpha}}(y;\alpha )\,\rmd y=
\int_{\bbR} f_{\sfY}^{ \frac{1}{\alpha}}(y;\alpha)\,\rmd y.
\end{equation} 
This completes the proof of continuity for $\alpha>1$.
\end{itemize}
This concludes the proof of Theorem~\ref{thm:Continuity}.

\section{Proof of Proposition~\ref{prop:ConcConv}}
\label{app:ConcavityConvexity}
Let $P_{\sfX} = \lambda P_{\sfX_1} + (1-\lambda) P_{\sfX_2}$, where $\lambda \in (0,1)$.
From the representation of $\alpha$-mutual information in~\eqref{eq:norm_representation}, for a fixed $\rm{snr}>0$ and  up to a multiplicative constant independent of $P_{\sfX}$, we have that
\begin{equation}
\label{eq:gAlpha}
\zeta_{\alpha}(\sfX; {\rm snr}) = \|  f_\sfY(\cdot; \alpha \; {\rm snr})  \|_{ \frac{1}{\alpha}} = \|  \lambda f_{\sfY_1}(\cdot; \alpha \; {\rm snr}) +(1-\lambda) f_{\sfY_2}(\cdot; \alpha \; {\rm snr})  \|_{ \frac{1}{\alpha}},
\end{equation}
where $f_{\sfY_i}, i \in \{1,2\}$ is the density induced by the input distribution $P_{\sfX_i}$.
Now, the following holds:
\begin{itemize}
\item For $\alpha \in (1,\infty)$, since by assumption we have that both $ \|   f_{\sfY_1}(\cdot; \alpha \; {\rm snr})   \|_{ \frac{1}{\alpha}} \|$ and $\|f_{\sfY_2}(\cdot; \alpha \; {\rm snr})  \|_{ \frac{1}{\alpha}} $ are finite, we can apply the reverse Minkowski inequality and obtain
\begin{equation}
\label{eq:ConcMIPX}
 \|  \lambda f_{\sfY_1}(\cdot; \alpha \; {\rm snr}) +(1-\lambda) f_{\sfY_2}(\cdot; \alpha \; {\rm snr})  \|_{ \frac{1}{\alpha}}  \geq  \lambda \|   f_{\sfY_1}(\cdot; \alpha \; {\rm snr})   \|_{ \frac{1}{\alpha}}+ (1-\lambda) \|   f_{\sfY_2}(\cdot; \alpha \; {\rm snr})  \|_{ \frac{1}{\alpha}} ,
\end{equation}
which implies that $\zeta_{\alpha}(\sfX; {\rm snr})$ is concave. 
\item For $\alpha \in (0,1)$, we can apply the Minkowski inequality to~\eqref{eq:gAlpha} and write~\eqref{eq:ConcMIPX} with the sign of the inequality swapped (i.e., $\leq$ instead of $\geq$), which implies that $\zeta_{\alpha}(\sfX; {\rm snr})$ is convex.
\end{itemize}
It, therefore, follows that proving {\em strict} concavity of $\zeta_{\alpha}(\sfX; {\rm snr})$ for $\alpha \in (1,\infty)$ and {\em strict} convexity of $\zeta_{\alpha}(\sfX; {\rm snr})$ for $\alpha \in (0,1)$  reduces to showing that the inequality in~\eqref{eq:ConcMIPX} (and the inequality swapped) is {\em strict}. Toward this end, we note that the reverse Minkowski inequality and the Minkowski inequality hold with equality if and only if 
\begin{equation}
\label{eq:MinkEquality}
f_{\sfY_1}(y; \alpha \; {\rm snr}) = \kappa \, f_{\sfY_2}(y; \alpha \; {\rm snr}) \quad \text{for almost every} \,  y.
\end{equation}
Since, for $i \in \{1,2\}$, we have that
\begin{equation}
\int_\bbR f_{\sfY_i}(y; \alpha \; {\rm snr}) \, \rmd y = 1,
\end{equation}
then $\kappa$ in~\eqref{eq:MinkEquality} must be one, that is,
\begin{equation}
\label{equalityrule}
f_{\sfY_1}(y; \alpha \; {\rm snr}) =  f_{\sfY_2}(y; \alpha \; {\rm snr}) \quad \text{for almost every } y.
\end{equation}
We now show that, given $P_{\sfX_{1}} \neq P_{\sfX_{2}}$, the above cannot happen, i.e., $f_{\sfY_1}(y; \alpha \; {\rm snr}) \neq  f_{\sfY_2}(y; \alpha \; {\rm snr})$ for almost every $y$. In the remaining of this proof, we let $\phi_\sfW$ be the characteristic function of $\sfW$. With this, for $i \in \{1,2\}$, we have that
\begin{equation}
\phi_{\sfY_i}(t) = \phi_{\sfX_i}(t) \, \phi_{\sfN}(t),
\end{equation}
where $\sfX_i \sim P_{\sfX_i}$ and $\sfN = \frac{1}{\sqrt{\mathrm{snr}}}\sfZ$. 
Since $\phi_{\sfN}(t) >0$ for all $t \in \bbR$, we can {\em uniquely} recover the input distribution from the output density by taking the inverse of
\begin{equation}
\phi_{\sfX_i}(t) = \frac{\phi_{\sfY_i}(t)}{\phi_{\sfN}(t)}.
\end{equation}
 Since $P_{\sfX_{1}} \neq P_{\sfX_{2}}$, it follows that $f_{\sfY_1}(y; \alpha \; {\rm snr}) \neq  f_{\sfY_2}(y; \alpha \; {\rm snr})$ almost surely which contradicts~\eqref{equalityrule}. Thus, the inequalities in the reverse Minkowski inequality and in the Minkowski inequality are {\em strict}.

The proof of Proposition~\ref{prop:ConcConv} is concluded by noting that for $\alpha \in (1,\infty)$, $I_\alpha(\sfX;{\rm snr})$ is {\em strictly} concave  since, from~\eqref{eq:norm_representation}, we have $I_\alpha(\sfX;{\rm snr}) =  \frac{1}{\alpha-1} \log(\zeta_{\alpha}(\sfX; {\rm snr}))$ and logarithm preserves strict concavity.

\section{Proof of Theorem~\ref{thm:IMMSE}}
\label{app:DerivativeMI}

For $\snr>0$, define
\begin{equation}
    f_{\snr}(y) = f_{\sfY}(y;\alpha \; \snr),
    \qquad
    q= \frac{1}{\alpha},
    \qquad
    \Psi_\alpha(\snr) = \int_{\bbR}f_{\snr}^q(y) \, \rmd y. \label{eq:psi_def}
\end{equation}
Before presenting a full proof of Theorem~\ref{thm:IMMSE}, we present a sequence of helper lemmas. 
We start with the following posterior averaging lemma. 
\begin{lemma}
\label{lem:residual}
Assume that $f_{\sfY_\alpha}$ exists. Define
\begin{equation}
    R_\alpha(y,\snr) =\bbE\left[(\sfX-\sfY)^2\mid \sfY=y;\alpha \; \snr\right].
\end{equation}
Then, 
\begin{equation}
\label{eq:tilted-residual}
    \int_{\bbR} f_{\sfY_\alpha}(y) \, R_\alpha(y,\snr)\,\rmd y
    =
    \frac{1}{\snr}
    +(1-\alpha ) \,  {\rm mmse}(\sfX_\alpha\mid\sfY_\alpha).
\end{equation}
Moreover, the integral on the left-hand side is finite.
\end{lemma}
\begin{IEEEproof}
From the definition in~\eqref{eq:conditional_expression}, we have that $P_{\sfX_\alpha\mid\sfY_\alpha=y}
    =P_{\sfX\mid\sfY=y;\alpha \; \snr}$, which implies that
    \begin{equation}
        R_\alpha(y,\snr)={\rm Var}(\sfX_\alpha\mid\sfY_\alpha=y)
 +\left(\bbE[\sfX\mid\sfY=y;\alpha \; \snr]-y\right)^2.
 \label{eq:RAlpha}
    \end{equation}
   Next,  Tweedie's formula~\cite{robbins1956empirical,dytso2022conditional} gives
\begin{equation}
    \bbE[\sfX\mid\sfY=y;\alpha \; \snr]-y
    =\frac{1}{\alpha \; \snr}\rho_{\sfY;\alpha \; \snr}(y)= \frac{1}{\snr}\rho_{\sfY_\alpha}(y), \label{eq:tilted_tweedy}
\end{equation}
where in the last equality we have used the score identity in~\eqref{eq:score_function_identity}.   Averaging~\eqref{eq:RAlpha} with respect to $f_{\sfY_\alpha}$, therefore, yields
\begin{equation}
\int_{\bbR}f_{\sfY_\alpha}(y) \, R_\alpha(y,\snr)\,\rmd y=
{\rm mmse}(\sfX_\alpha\mid\sfY_\alpha)
+\frac{1}{\snr^2}J(\sfY_\alpha). \label{eq:integral_expression_for_R}
\end{equation}
The generalized Brown's identity from Theorem~\ref{theorem:GeneralizationBrown} gives
\begin{equation}
    J(\sfY_\alpha)
    ={\rm snr} -  \alpha \; {\rm snr}^2 \; {\rm mmse}(\sfX_\alpha | \sfY_\alpha),
\end{equation}
which proves~\eqref{eq:tilted-residual}.
Finally, the integral in
 \eqref{eq:tilted-residual} is finite since by using~\eqref{eq:UBMMSE}, we have that 
 \begin{align}
 \left| \frac{1}{\snr} 
    +(1-\alpha ) \,  {\rm mmse}(\sfX_\alpha\mid\sfY_\alpha) \right |  \le \frac{1}{\snr} \left(1 +\frac{|1-\alpha|}{\alpha} \right),
\end{align}
which completes the proof of Lemma~\ref{lem:residual}. 
 \end{IEEEproof}
\begin{lemma}
\label{lem:residual-monotonicity}
For every fixed $y\in\bbR$, the mapping
$\snr\mapsto R_\alpha(y,\snr)$ is nonincreasing on $(0,\infty)$.
\end{lemma}

\begin{IEEEproof}
Fix $y\in\bbR$ and define
\begin{align}
D_y(\snr)
&=
\bbE\left[
\exp\left(
-\frac{\alpha \; \snr}{2}(y-\sfX)^2
\right)
\right],\\
N_y(\snr)
&=
\bbE\left[
(y-\sfX)^2
\exp\left(
-\frac{\alpha \; \snr}{2}(y-\sfX)^2
\right)
\right].
\end{align}
By Bayes' rule,
\begin{equation}
\label{eq:R_ratio}
R_\alpha(y,\snr)
=
\frac{N_y(\snr)}{D_y(\snr)}.
\end{equation}
Differentiating with respect to $\snr$ gives
\begin{align}
D_y'(\snr)
&=
-\frac{\alpha}{2}N_y(\snr),\\
N_y'(\snr)
&=
-\frac{\alpha}{2}
\bbE\left[
(y-\sfX)^4
\exp\left(
-\frac{\alpha \; \snr}{2}(y-\sfX)^2
\right)
\right].
\end{align}
These differentiations are justified locally on $(0,\infty)$ since,
for $k\in\{1,2\}$ and every $c>0$,
\begin{equation}
\sup_{u\geq0}u^k \rme^{-cu}<\infty.
\end{equation}
Applying the quotient rule to \eqref{eq:R_ratio}, we obtain
\begin{align}
\frac{\partial}{\partial\snr}
R_\alpha(y,\snr)
&=
-\frac{\alpha}{2}
\left(
\frac{
\bbE\left[
(y-\sfX)^4
\rme^{-\frac{\alpha \; \snr}{2}(y-\sfX)^2}
\right]
}{
D_y(\snr)
}
-
\left(
\frac{N_y(\snr)}{D_y(\snr)}
\right)^2
\right)
\nonumber\\
&=
-\frac{\alpha}{2}
\left(
\bbE\left[
(y-\sfX)^4
\mid
\sfY=y;\alpha \; \snr
\right]
-
R^2_\alpha(y,\snr)
\right)
\nonumber\\
&=
-\frac{\alpha}{2}
{\rm Var}\left(
(y-\sfX)^2
\mid
\sfY=y;\alpha \; \snr
\right)
\leq 0,
\label{eq:R_derivative}
\end{align}
where we have used the expression of $R_\alpha(y,\snr)$ in Lemma~\ref{lem:residual}.
This concludes the proof of Lemma~\ref{lem:residual-monotonicity}.
\end{IEEEproof}
We can now justify the differentiation of $\Psi_\alpha$ along the Gaussian
precision flow.
\begin{lemma}
\label{lem:heat-diff}
Fix $\snr>0$. Assume that $f_{\sfY_\alpha}$ exists at $\snr$. Then,
\begin{align}
\label{eq:heat-diff}
\Psi_\alpha'(\snr)
&=
-\frac{1}{2 \; \alpha^2 \; \snr^2}
\int_{\bbR}
f^{\frac{1}{\alpha}-1}_{\snr}(y) \,
f_{\snr}''(y)\,\rmd y
\\
&=
-\frac{1}{2}
\int_{\bbR}
f^{\frac{1}{\alpha}}_{\snr}(y) \,
R_\alpha(y,\snr)\,\rmd y
+
\frac{1}{2 \; \alpha \; \snr}
\Psi_\alpha(\snr).
\label{eq:*Derivative_of_psi}
\end{align}
The integral in~\eqref{eq:heat-diff} is absolutely convergent.
\end{lemma}
\begin{IEEEproof}
We first establish the pointwise derivative. Recall that
\begin{equation}
\label{eq:f-explicit}
f_{\snr}(y)
=
\sqrt{\alpha \; \snr}\,
\bbE\left[
\phi \left(\sqrt{\alpha \; \snr}(y-\sfX)\right)
\right].
\end{equation}
For a fixed $y$, differentiation through the expectation gives
\begin{equation}
\label{eq:f-snr-derivative}
\frac{\partial}{\partial\snr}f_{\snr}(y)
=
\frac{1}{2 \; \snr}f_{\snr}(y)
-\frac{\alpha}{2} \, f_{\snr}(y) \, R_\alpha(y,\snr).
\end{equation}
Indeed, for $\gamma$ in any compact neighborhood of $\snr>0$,
\begin{align}
&\left|
\frac{\partial}{\partial\gamma}
\left[
\sqrt{\alpha \; \gamma}\,
\phi \left(\sqrt{\alpha \; \gamma}(y-x)\right)
\right]
\right|
\leq
C_{\alpha,\snr}
\left(1+(y-x)^2\right)
\exp\left(-c_{\alpha,\snr}(y-x)^2\right)
\leq C'_{\alpha,\snr},
\end{align}
uniformly in $x$. Hence, differentiation through the expectation
follows from the dominated convergence theorem.
Moreover, direct differentiation with respect to $y$ gives
\begin{equation}
\label{eq:curvature}
\frac{f_{\snr}''(y)}{f_{\snr}(y)}
=
\alpha^2 \; \snr^2 \; R_\alpha(y,\snr)-\alpha \; \snr,
\end{equation}
where the differentiation can be interchanged with the expectation in~\eqref{eq:f-explicit} since, for $k\in\{1,2\}$,
\begin{equation}
    \sup_{u\in\bbR}
    |u|^k \rme^{-c u^2}<\infty,
    \qquad c>0.
\end{equation}
Combining~\eqref{eq:f-snr-derivative} and~\eqref{eq:curvature} yields
the identity\footnote{The identity in~\eqref{eq:heat_equation} is a form of well-known heat equations~\cite{evans2010pde}.}
\begin{equation}
\label{eq:heat_equation}
\frac{\partial}{\partial\snr}f_{\snr}(y)
=
-\frac{1}{2 \; \alpha \; \snr^2}f_{\snr}''(y).
\end{equation}
It remains to justify differentiation under the $y$-integral. For
this purpose, let $\gamma$ satisfy
\begin{equation}
\label{eq:gamma_neighborhood}
\frac{\snr}{2}
\leq
\gamma
\leq
\frac{3 \, \snr}{2}.
\end{equation}
Since $\gamma\geq\snr/2$, the Gaussian-convolution representation
implies
\begin{align}
f_{\sfY}(y;\alpha \; \gamma)
&=
\sqrt{\frac{\alpha \; \gamma}{2\pi}}
\bbE\left[
\exp\left(
-\frac{\alpha \; \gamma}{2}(y-\sfX)^2
\right)
\right]
\nonumber\\
&\leq
\sqrt{\frac{2\gamma}{\snr}}\,
f_{\sfY}\left(y;\frac{\alpha \; \snr}{2}\right)
\nonumber\\
&\leq
\sqrt{3}\,
f_{\sfY}\left(y;\frac{\alpha \; \snr}{2}\right).
\label{eq:f_domination}
\end{align}
Consequently,
\begin{equation}
\label{eq:fq_domination}
f^q_{\sfY}(y;\alpha \; \gamma)
\leq
3^{q/2}
f^q_{\sfY}\left(y;\frac{\alpha \; \snr}{2}\right).
\end{equation}
Moreover, by Lemma~\ref{lem:residual-monotonicity},
\begin{equation}
\label{eq:R_domination}
R_\alpha(y,\gamma)
\leq
R_\alpha\left(y,\frac{\snr}{2}\right).
\end{equation}
From~\eqref{eq:f-snr-derivative} and $q=1/\alpha$, we have
\begin{equation}
\frac{\partial}{\partial\gamma}
f^q_{\sfY}(y;\alpha \; \gamma)
=
f^q_{\sfY}(y;\alpha \; \gamma)
\left(
\frac{1}{2\alpha \; \gamma}
-
\frac{1}{2}R_\alpha(y,\gamma)
\right).
\end{equation}
Therefore,~\eqref{eq:gamma_neighborhood},~\eqref{eq:fq_domination}, and~\eqref{eq:R_domination} yield
\begin{align}
\left|
\frac{\partial}{\partial\gamma}
f^q_{\sfY}(y;\alpha\gamma)
\right|
&\leq
3^{q/2}
f^q_{\sfY}\left(
y;\frac{\alpha \; \snr}{2}
\right)
\left(
\frac{1}{\alpha \; \snr}
+
\frac{1}{2}
R_\alpha\left(
y,\frac{\snr}{2}
\right)
\right).
\label{eq:dominating_function}
\end{align}
The right-hand side is integrable since   $f_{\sfY_\alpha}$ is assumed to exist and 
from Lemma~\ref{lem:residual}, we have that
\begin{align}
\int_{\bbR}
f^q_{\sfY}\left(y;\frac{\alpha\snr}{2}\right)
R_\alpha\left(y,\frac{\snr}{2}\right)\,\rmd y =
\Psi_\alpha\left(\frac{\snr}{2}\right)
\left(
\frac{2}{\snr}
+
(1-\alpha) \,
{\rm mmse}_\alpha\left(\sfX;\frac{\snr}{2}\right)
\right)
<\infty.
\label{eq:dominating_integrable}
\end{align}
Thus,~\eqref{eq:dominating_function} provides an integrable dominating
function independent of $\gamma$ in the neighborhood~\eqref{eq:gamma_neighborhood}. The dominated convergence theorem,
therefore, gives
\begin{align}
\Psi_\alpha'(\snr)
&=
\int_{\bbR}
\left.
\frac{\partial}{\partial\gamma}
f^q_{\sfY}(y;\alpha \; \gamma)
\right|_{\gamma=\snr}
\,\rmd y
\nonumber\\
&=
q\int_{\bbR}
f^{q-1}_{\snr}(y) \,
\frac{\partial}{\partial\snr}f_{\snr}(y)
\,\rmd y
\nonumber\\
&=
-\frac{1}{2 \; \alpha^2 \; \snr^2}
\int_{\bbR}
f^{\frac{1}{\alpha}-1}_{\snr}(y) \,
f_{\snr}''(y)\,\rmd y,
\end{align}
where the last equality follows from~\eqref{eq:heat_equation}. This
proves the first equality in~\eqref{eq:heat-diff}.

Next, using~\eqref{eq:curvature},
\begin{align}
&-\frac{1}{2 \; \alpha^2 \; \snr^2}
\int_{\bbR}
f^{\frac{1}{\alpha}-1}_{\snr}(y)
\, f_{\snr}''(y)\,\rmd y
=
-\frac{1}{2}
\int_{\bbR}
f^q_{\snr}(y)
\, R_\alpha(y,\snr)\,\rmd y
+
\frac{1}{2 \; \alpha \; \snr}
\int_{\bbR}
f^q_{\snr}(y)\,\rmd y,
\end{align}
which proves~\eqref{eq:*Derivative_of_psi}.

Finally,~\eqref{eq:curvature} also gives
\begin{align}
\int_{\bbR}
\left|
f^{\frac{1}{\alpha}-1}_{\snr}(y)
f_{\snr}''(y)
\right|\,\rmd y
\leq
\alpha^2\snr^2
\int_{\bbR}
f^q_{\snr}(y) \, R_\alpha(y,\snr)\,\rmd y
+
\alpha \; \snr\,\Psi_\alpha(\snr).
\label{eq:heat_absolute_convergence}
\end{align}
Since
\begin{equation}
\int_{\bbR}
f^q_{\snr}(y) \, R_\alpha(y,\snr)\,\rmd y
=
\Psi_\alpha(\snr)
\int_{\bbR}
f_{\sfY_\alpha}(y) \, R_\alpha(y,\snr)\,\rmd y
<\infty
\end{equation}
by Lemma~\ref{lem:residual}, and
$\Psi_\alpha(\snr)<\infty$ by assumption, the right-hand side of~\eqref{eq:heat_absolute_convergence} is finite. This proves the
absolute convergence in~\eqref{eq:heat-diff} and concludes the proof of Lemma~\ref{lem:heat-diff}.
\end{IEEEproof}
We are now ready to prove Theorem~\ref{thm:IMMSE}. We have that
\begin{align}
\frac{\rmd}{\rmd\snr}
I_\alpha(\sfX;\snr)
&= \frac{\rmd}{ \rmd \snr}  \left(  \frac{\alpha}{\alpha-1}
\log \Psi_\alpha(\snr) 
+
\Delta(\alpha,\snr)\right) \label{eq:using_eq:AlphaMISecondRepresentation} \\
&=
\frac{\alpha}{\alpha-1}
\frac{\Psi_\alpha'(\snr)}{\Psi_\alpha(\snr)}
+
\frac{1}{2 \; \snr}  \label{eq:using_eq:_Derivative_of_psi}\\
&=
  \frac{\alpha}{\alpha-1} \left( -\frac{1}{2}
\int_{\bbR}
 f_{\sfY_\alpha}(y) \,
 R_\alpha(y,\snr)
 \,\rmd y +\frac{1}{2 \; \alpha \; \snr} \right)  +
\frac{1}{2 \; \snr} \label{eq:usinTmP} \\
&=
\frac{\alpha}{2}
{\rm mmse}(\sfX_\alpha\mid\sfY_\alpha)
-
\frac{1}{2 \; \snr}
+
\frac{1}{2 \; \snr} \label{eq:using_eq:tilted-residual}
\\
&=
\frac{\alpha}{2}
{\rm mmse}(\sfX_\alpha\mid\sfY_\alpha),
\end{align}
where~\eqref{eq:using_eq:AlphaMISecondRepresentation} follows from~\eqref{eq:AlphaMISecondRepresentation} and the definition of $\Psi_\alpha(\snr)$  in~\eqref{eq:psi_def};~\eqref{eq:usinTmP} is due to Lemma~\ref{lem:heat-diff} and~\eqref{eq:PDFYalpha}; and~\eqref{eq:using_eq:tilted-residual} follows from~\eqref{eq:tilted-residual}. This concludes the proof of Theorem~\ref{thm:IMMSE}.

\section{Proof of~\eqref{eq:VarXalphaToVarX}}
\label{app:VarXalphaToVarX}
We consider the cases of $\alpha<1$ and $\alpha>1$ separately.
\begin{enumerate}
\item {\bf{Case~1:} $\alpha \in (0,1)$.}
By using the law of total expectation with $k \in \{1,2\}$, we have that
\begin{align}
\bbE[\sfX_\alpha^k] &= \bbE \left [\bbE \left [\sfX_\alpha^k|\sfY_\alpha \right ] \right ] 
\\& = \int_{\bbR} \bbE \left [\sfX_\alpha^k|\sfY_\alpha = y \right ] \, f_{\sfY_\alpha}(y) \, \rmd y
\\& = \int_{\bbR} \bbE \left [\sfX_\alpha^k|\sfY_\alpha = y \right ] \, \frac{f_\sfY^{ \frac{1}{\alpha}}(y; \alpha \; {\rm snr}) }{ \int_{\bbR} f_\sfY^{ \frac{1}{\alpha}}(t; \alpha \; {\rm snr}) \; \rmd t} \, \rmd y \label{eq:FirstStepKthExpectedValue}
\\& = \frac{1}{\rme^{\frac{\alpha-1}{\alpha} I_{\alpha}(\sfX;{\rm{snr}})}} \int_\bbR  \bbE \left [ \sfX^k_\alpha | \sfY_\alpha=y  \right ] \, \sqrt{ \rm{snr}} \; \bbE^{\frac{1}{\alpha}}[  \phi^\alpha(\sqrt{ {\rm{snr}}}(y-\sfX))] \; \rmd y \label{eq:SecondStepKthExpectedValue}
\\& = \frac{1}{\rme^{\frac{\alpha-1}{\alpha} I_{\alpha}(\sfX;{\rm{snr}})}} \int_\bbR  \bbE \left [  \sfX^k_\alpha \left | \sfY_\alpha= \frac{y}{\sqrt{{\rm{snr}}}} \right . \right ] \, \bbE^{\frac{1}{\alpha}}[  \phi^\alpha(y-\sqrt{ {\rm{snr}}} \, \sfX)] \; \rmd y
\label{eq:SecondMomentXalpha}
\\& = \frac{1}{\rme^{\frac{\alpha-1}{\alpha} I_{\alpha}(\sfX;{\rm{snr}})}} \int_\bbR  \frac{\bbE [\sfX^k \, \phi^\alpha(y-\sqrt{ {\rm{snr}}} \, \sfX)] }{\bbE [\phi^\alpha(y-\sqrt{ {\rm{snr}}} \, \sfX)]} \, \bbE^{\frac{1}{\alpha}}[  \phi^\alpha(y-\sqrt{ {\rm{snr}}} \, \sfX)] \; \rmd y, \label{eq:LastStepKthExpectedValue}
\end{align}
where~\eqref{eq:FirstStepKthExpectedValue} follows from using~\eqref{eq:PDFYalpha};~\eqref{eq:SecondStepKthExpectedValue} follows from the representation of $\alpha$-mutual information in~\eqref{eq:AlphaMISecondRepresentation}; and~\eqref{eq:LastStepKthExpectedValue} follows from~\eqref{eq:conditional_expression}.  
We now show that
\begin{equation}
\int_\bbR  \frac{\bbE \left [\sfX^k \, \phi^\alpha(y-\sqrt{ {\rm{snr}}} \, \sfX) \right ] }{\bbE \left [\phi^\alpha(y-\sqrt{ {\rm{snr}}} \, \sfX) \right ]} \, \bbE^{\frac{1}{\alpha}}[  \phi^\alpha(y-\sqrt{ {\rm{snr}}} \, \sfX)] \; \rmd y \notag= \bbE[  \sfX^k  \xi(\sfX, {\rm snr})],
\end{equation}
where
\begin{equation}
    \xi(\sfX, {\rm snr}) =  \int_{\bbR}  \frac{ \, \phi^\alpha(y-\sqrt{ {\rm{snr}}} \, \sfX)  }{\bbE \left [\phi^\alpha(y-\sqrt{ {\rm{snr}}} \, \sfX) \right ]} \, \bbE^{\frac{1}{\alpha}}[  \phi^\alpha(y-\sqrt{ {\rm{snr}}} \, \sfX)] \; \rmd y.
\end{equation}
This follows from Fubini-Tonelli's theorem since, for $\alpha \in(0,1)$, we have that $\xi(\sfX,{\rm snr})$ is bounded by a constant, that is,
\begin{align}
\xi(\sfX, {\rm snr}) &=  \int_{\bbR}  \frac{ \, \phi^\alpha(y-\sqrt{ {\rm{snr}}} \, \sfX)  }{\bbE \left [\phi^\alpha(y-\sqrt{ {\rm{snr}}} \, \sfX) \right ]} \, \bbE^{\frac{1}{\alpha}}[  \phi^\alpha(y-\sqrt{ {\rm{snr}}} \, \sfX)] \; \rmd y \label{eq:gFunctionDCT}
\\& \leq \int_\bbR  \frac{ \ \phi^\alpha(y-\sqrt{ {\rm{snr}}} \, \sfX) }{\bbE \left [\phi^\alpha(y-\sqrt{ {\rm{snr}}} \, \sfX) \right ]} \, \bbE \left [  \phi(y-\sqrt{ {\rm{snr}}} \, \sfX) \right ] \; \rmd y \label{eq:JensenLowSnrMMSE}
\\& \leq \int_\bbR \phi^\alpha(y-\sqrt{ {\rm{snr}}} \, \sfX) \; \rmd  y\label{eq:PowerAlphaGreaterThanPowerOne} \\
&= C_\alpha, \label{eq:CAlpha}
\end{align}
where~\eqref{eq:JensenLowSnrMMSE} follows from Jensen's inequality;~\eqref{eq:PowerAlphaGreaterThanPowerOne} is due to the fact that $x^\alpha > x$ for $x \in (0,1)$ and $\alpha \in (0,1)$; and in~\eqref{eq:CAlpha} we let $C_\alpha = \frac{\sqrt{2 \pi  (1/\alpha)}}{(2 \pi)^{\alpha/2}}$.
Consequently, by the dominated convergence theorem and the fact that we assume that $\bbE[\sfX^2] <\infty$, from~\eqref{eq:LastStepKthExpectedValue}, we can write
\begin{align}
\lim_{\rm snr \to 0} \bbE[\sfX_\alpha^k] &= \lim_{\rm snr \to 0}  \frac{1}{\rme^{\frac{\alpha-1}{\alpha} I_{\alpha}(\sfX;{\rm{snr}})}}  \bbE [\sfX^k  \xi (\sfX,{\rm snr}) ]\\
&=   \frac{1}{\rme^{\frac{\alpha-1}{\alpha} \lim_{\rm snr \to 0}I_{\alpha}(\sfX;{\rm{snr}})}}  \bbE [\sfX^k   \lim_{\rm snr \to 0} \xi (\sfX,{\rm snr}) ]
\\&=    \bbE \left [\sfX^k   \lim_{\rm snr \to 0} \xi(\sfX,{\rm snr}) \right ] \label{eq:Prop1LimitSNRToZeroIalpha}
\\& = \bbE \left [\sfX^k \right ], \label{eq:LimitInsideIntegralIng}
\end{align}
where in~\eqref{eq:Prop1LimitSNRToZeroIalpha}, we have used the result in Proposition~\ref{prop:LimitSNRToZeroAlphaMI} and in~\eqref{eq:LimitInsideIntegralIng}, we have applied the dominated convergence theorem to bring the limit inside the integral in~\eqref{eq:gFunctionDCT}, which  is possible since
\begin{align}
\xi(\sfX, {\rm snr}) &=  \int_{\bbR}   \phi^\alpha(y-\sqrt{ {\rm{snr}}} \, \sfX)   \, \bbE^{\frac{1}{\alpha}-1}[  \phi^\alpha(y-\sqrt{ {\rm{snr}}} \, \sfX)] \; \rmd y
\\& \leq \left( \frac{1}{\sqrt{2 \, \pi}} \right )^{1-\alpha} \int_{\bbR}   \phi^\alpha(y-\sqrt{ {\rm{snr}}} \, \sfX)   \,  \rmd y,
\end{align}
and since
\begin{align}
\phi^\alpha(y-\sqrt{ {\rm{snr}}} \, \sfX)  &= \left( \frac{1}{\sqrt{2 \, \pi}} \right )^\alpha \exp \left( - \alpha \frac{(y-\sqrt{{\rm{snr}}} \, \sfX)^2}{2} \right )
\\& \leq \left( \frac{1}{\sqrt{2 \, \pi}} \right )^\alpha \exp \left( -\frac{\alpha}{2} \left(\frac{1}{2}y^2 - {\rm{snr}} \, \sfX^2 \right ) \right )
\\& = \left( \frac{1}{\sqrt{2 \, \pi}} \right )^\alpha \exp \left( \frac{\alpha \, {\rm{snr}} \, \sfX^2}{2} \right ) \exp \left( -\frac{\alpha}{4}y^2\right )
\\& \leq \left( \frac{1}{\sqrt{2 \, \pi}} \right )^\alpha \exp \left( \frac{\alpha \, \delta \, \sfX^2}{2} \right ) \exp \left( -\frac{\alpha}{4}y^2\right ),
\end{align}
which is an integrable function of $y$ for a fixed $\sfX$. Note that we have used the inequality $(a-b)^2 \geq \frac{1}{2} a^2 - b^2$ 
and the assumption that, without loss of generality, $\rm{snr} \leq \delta$ for some $\delta>0$. This concludes the proof for $\alpha \in (0,1)$.

\item {\bf{Case~2:} $\alpha \in (1,+\infty)$.} 
By using the law of total expectation with $k \in \{1,2\}$, we have that
\begin{align}
    \bbE[\sfX_{\alpha}^k ] & = \bbE \left [\bbE \left [\sfX_\alpha^k|\sfY_\alpha \right ] \right ]  
    \\&=  \bbE \left[ \bbE[\sfX^k| \sfY =\sfY_\alpha; \alpha \, {\rm snr}] \right] \label{eq:CondModelEqualOriginal}\\
    &=\bbE \left[ \bbE \left[\sfX^k \mathbbm 1_{\{|\sfX| \le M \}} | \sfY=\sfY_\alpha; \alpha \, {\rm snr} \right] \right] + \bbE \left[ \bbE \left[\sfX^k \mathbbm 1_{\{|\sfX| > M \}} | \sfY=\sfY_\alpha; \alpha \, {\rm snr} \right] \right], \label{eq:two_cond_expecatations}
\end{align}
where~\eqref{eq:CondModelEqualOriginal} follows from~\eqref{eq:conditional_expression}. 
We now study the two expectations in~\eqref{eq:two_cond_expecatations} separately. First, we claim that for every $M >0$, we have that
\begin{equation}
    \lim_{{\rm snr} \to 0 } \bbE \left[ \bbE \left[\sfX^k \mathbbm 1_{\{|\sfX| \le M\} } | \sfY=\sfY_\alpha; \alpha \, {\rm snr} \right] \right] =  \bbE\left[\sfX^k \mathbbm 1_{\{|\sfX| \le M \} } \right]. \label{eq:first_truncated_moment}
\end{equation}
To see this,  let 
\begin{align}
g_k(y; {\rm snr}) &=  \bbE \left[\sfX^k \mathbbm 1_{\{|\sfX| \le M \}} | \sfY=y; \alpha \, {\rm snr} \right], \\
m_k &=\bbE\left[\sfX^k \mathbbm 1_{\{|\sfX| \le M \}} \right],
\end{align}
and note that  $|g_k(y; {\rm snr})|= | \bbE \left[\sfX^k \mathbbm 1_{\{|\sfX| \le M \}} | \sfY=y; \alpha \, {\rm snr} \right] | \le M^k$.  Next, note that 
\begin{align}
\lim_{{\rm snr} \to 0 } g_k(y; {\rm snr}) &= \lim_{{\rm snr} \to 0 } \bbE \left[\sfX^k \mathbbm 1_{\{|\sfX| \le M \}} | \sfY=y; \alpha \, {\rm snr} \right] \\
 &=  \lim_{{\rm snr} \to 0 } \frac{ \bbE \left[ \sfX^k \mathbbm 1_{\{|\sfX| \le M \}} \, \rme^{ - \frac{\alpha \, {\rm snr}}{2} (y -\sfX)^2} \right ]}{ \bbE \left[\rme^{ - \frac{\alpha \, {\rm snr}}{2} (y -\sfX)^2}  \right ]} \\  
 &=  \frac{ \lim_{{\rm snr} \to 0 }  \bbE \left [ \sfX^k \mathbbm 1_{\{|\sfX| \le M \}} \, \rme^{ - \frac{\alpha \, {\rm snr}}{2} (y -\sfX)^2} \right ]}{  \lim_{{\rm snr} \to 0 }  \bbE \left [\rme^{ - \frac{\alpha \, {\rm snr}}{2} (y -\sfX)^2}  \right ]} \\ 
 &= \bbE\left[\sfX^k \mathbbm 1_{\{|\sfX| \le M \}} \right] =m_k, \label{eq:limit_of momennt}
\end{align}
with the last step following by the dominated convergence theorem  since both the functions in the numerator and denominator are bounded. 

Next, by using a similar change of variable leading to~\eqref{eq:SecondMomentXalpha}, we have that
\begin{align}
\bbE \left[ \bbE \left[\sfX^k \mathbbm 1_{\{|\sfX| \le M \}} | \sfY=\sfY_\alpha; \alpha \, {\rm snr} \right] \right] &=\int_\bbR  g_{k} \left( y; {\rm snr} \right) f_{\sfY_\alpha}(y) \, \rmd y  \\
&=  \frac{1}{\rme^{\frac{\alpha-1}{\alpha} I_{\alpha}(\sfX;{\rm{snr}})}} \int_\bbR  g_{k} \left( \frac{y}{\sqrt{\rm snr}}; {\rm snr} \right)\bbE^{\frac{1}{\alpha}}[  \phi^\alpha(y-\sqrt{ {\rm{snr}}} \, \sfX)] \; \rmd y. \label{eq:expresson_inters_of_I_alpha}
\end{align}
Now, using the triangle inequality, we obtain 
\begin{align}
&\left| \int_\bbR  g_{k} \left( \frac{y}{\sqrt{\rm snr}}; {\rm snr} \right)\bbE^{\frac{1}{\alpha}}[  \phi^\alpha(y-\sqrt{ {\rm{snr}}} \, \sfX)] \; \rmd y -m_k \right| \notag
\\& \leq \int_\bbR   \left| g_{k} \left( \frac{y}{\sqrt{\rm snr}}; {\rm snr} \right) -m_k \right| \phi(y) \; \rmd y +M^k  \int_\bbR \left|\bbE^{\frac{1}{\alpha}}[  \phi^\alpha(y-\sqrt{ {\rm{snr}}} \, \sfX)] -\phi(y) \right | \, \rmd y
\\& = \int_\bbR \left| g_{k} \left( \frac{y}{\sqrt{\rm snr}}; {\rm snr} \right) -m_k \right| \phi(y) \; \rmd y \notag
\\& \quad + M^k  \int_\bbR \left | \left( \bbE[  \phi^\alpha(y-\sqrt{ {\rm{snr}}} \, \sfX)\, \mathbbm 1_{\{|\sfX| \le M_2\}} ] +\bbE[  \phi^\alpha(y-\sqrt{ {\rm{snr}}} \, \sfX)\, \mathbbm 1_{\{|\sfX| > M_2\}} ] \right )^{\frac{1}{\alpha}} -\phi(y)\right |\, \rmd y \label{eq:IntermTrangleIneqM2}
\\& \leq \int_\bbR \left| g_{k} \left( \frac{y}{\sqrt{\rm snr}}; {\rm snr} \right) -m_k \right| \phi(y) \; \rmd y +  M^k  \int_\bbR   \left|  \phi(y)- \bbE^{\frac{1}{\alpha}}[  \phi^\alpha(y-\sqrt{ {\rm{snr}}} \, \sfX) \mathbbm 1_{\{|\sfX| \le M_2\}}] \right|  \; \rmd y \notag
\\& \quad + M^k  \int_\bbR    \bbE^{\frac{1}{\alpha}}[  \phi^\alpha(y-\sqrt{ {\rm{snr}}} \, \sfX) \mathbbm 1_{\{|\sfX| > M_2\}}]   \; \rmd y, 
\label{eq:diff_bound}
\end{align}
where in the last step we have again used triangle inequality (by adding an subtracting the term $\bbE^{\frac{1}{\alpha}}[  \phi^\alpha(y-\sqrt{ {\rm{snr}}} \, \sfX) \mathbbm 1_{\{|\sfX| \le M_2\}}$ inside the absolute value of the second integral in~\eqref{eq:IntermTrangleIneqM2}) and the fact that $|(a+b)^r -a^r| \leq |b^r|$ for $a,b \geq 0$ and $r \in (0,1)$.

The first integral  in~\eqref{eq:diff_bound} converges to zero by applying the dominated convergence theorem, with the pointwise convergence established using arguments similar to those in~\eqref{eq:limit_of momennt}. 
The second integral converges to   
\begin{equation}
\lim_{{\rm snr}\to 0}\int_\bbR   \left|  \phi(y)- \bbE^{\frac{1}{\alpha}}[  \phi^\alpha(y-\sqrt{ {\rm{snr}}} \, \sfX) \mathbbm 1_{\{|\sfX| \le M_2 \}}] \right|  \; \rmd y = 1 - \left( \Pr(|\sfX| \le M_2) \right )^{1/\alpha}.  \label{eq:secon_case_var}
\end{equation}To see this, first note that by the dominated convergence theorem,  a.s. $y$
 \begin{equation}
\lim_{{\rm snr}\to 0}
\bbE^{\frac{1}{\alpha}}
\left[
\phi^\alpha\left(y-\sqrt{{\rm snr}} \,\sfX\right)
\mathbbm 1_{\{{|\sfX|\leq M_2}\}}
\right]  =
\bbE^{\frac{1}{\alpha}}
\left[
\phi^\alpha(y)\mathbbm 1_{\{{|\sfX|\leq M_2}\}}
\right]
=\phi(y) \; \left( \Pr(|\sfX| \le M_2) \right )^{1/\alpha}.
 \label{eq:point_Wise_conv}
\end{equation}
Second, note that,  assuming without loss of generality that ${\rm snr}\leq 1$,  for $|x| \le M_2$ we have
\begin{align}
\phi\left(y-\sqrt{{\rm snr}} \, x\right)
&=\phi(y)\exp\left(
\sqrt{{\rm snr}} \, x \, y-\frac{{\rm snr} \, x^2}{2}
\right) \leq \phi(y)\exp\left(M_2|y|\right).
\end{align}
Consequently,
\begin{equation}
\left|
\phi(y)-
\bbE^{\frac{1}{\alpha}}
\left[
\phi^\alpha\left(y-\sqrt{{\rm snr}} \, \sfX\right)
 \mathbbm 1_{\{{|\sfX|\leq M_2}\}}
\right]
\right| \leq
\phi(y)\left(1+\exp\left(M_2|y|\right)\right), 
\end{equation}
where the right-hand side is integrable over $\bbR$. Thus, by the dominated convergence theorem and~\eqref{eq:point_Wise_conv}, we obtain
\begin{align}
& \lim_{{\rm snr}\to 0}
\int_{\bbR}
\left|
\phi(y)-
\bbE^{\frac{1}{\alpha}}
\left[
\phi^\alpha\left(y-\sqrt{{\rm snr}} \, \sfX\right)
 \mathbbm 1_{\{{|\sfX|\leq M_2}\}}
\right]
\right|
\rmd y 
\\& =
\left(1-\left( \Pr( |\sfX| \le M_2 ) \right )^{\frac{1}{\alpha}}\right)
\int_{\bbR}\phi(y) \, \rmd y
\\&=1-\left(\Pr( |\sfX| \le M_2) \right)^{\frac{1}{\alpha}}.
\end{align}

The third integral is already studied in  Appendix~\ref{app:LimitSNRToZeroAlphaMI}, but here we show the specifically needed steps: let $\delta>0$ be such that $\snr \le \delta $, then 
\begin{align}
    \int_\bbR    \bbE^{\frac{1}{\alpha}}[  \phi^\alpha(y-\sqrt{ {\rm{snr}}} \, \sfX) \mathbbm 1_{\{|\sfX| > M_2\}}] \, \rmd y  \le \int_\bbR    \bbE^{\frac{1}{\alpha}}[  \phi^\alpha(y-\sqrt{ \delta} \, \sfX) \mathbbm 1_{\{|\sfX| > M_2\}}]  \, \rmd y, \label{eq:using_eq:C_M_definition}
\end{align}
where in~\eqref{eq:using_eq:C_M_definition} we have used the bound in~\eqref{eq:C_M_definition}.

Combining~\eqref{eq:point_Wise_conv} and~\eqref{eq:using_eq:C_M_definition}, we arrive at the conclusion that for any positive $M, M_2$, and $\delta$ we have that 
\begin{align}
& \lim_{\rm snr \to 0 }   \left| \int_\bbR  g_{k} \left( \frac{y}{\sqrt{\rm snr}}; {\rm snr} \right)\bbE^{\frac{1}{\alpha}}[  \phi^\alpha(y-\sqrt{ {\rm{snr}}} \, \sfX)] \; \rmd y -m_k \right|  \notag\\
 &\le M^k \left(1-\left(\Pr( |\sfX| \le M_2) \right)^{\frac{1}{\alpha}} \right ) + M^k  \int_\bbR    \bbE^{\frac{1}{\alpha}}[  \phi^\alpha(y-\sqrt{ \delta} \, \sfX) \mathbbm 1_{\{|\sfX| > M_2\}}]  \, \rmd y  . 
\end{align}
Now taking $M_2$ to infinity, both of the above expressions go to zero. In particular, the second expression was shown to go to zero in~\eqref{eq:limit_M_C} provided that $I_\alpha(\sfX;\snr)<\infty$ is finite for some $\snr>0$. This concludes the proof of~\eqref{eq:first_truncated_moment}.

We now bound the second expectation in~\eqref{eq:two_cond_expecatations} such that it vanishes with $M>0$. Without loss of generality, we assume that $\rm{snr} \leq \delta$, for some $\delta>0$. We note that 
\begin{align}
     &
    \left|
    \bbE \left[ \bbE \left[\sfX^k \mathbbm 1_{\{|\sfX| > M \}} | \sfY=\sfY_\alpha; \alpha \, {\rm snr} \right] \right]
    \right|
    \\
    &
    \leq
    \bbE \left[ \bbE \left[|\sfX|^k \mathbbm 1_{\{|\sfX| > M \}} | \sfY=\sfY_\alpha; \alpha \, {\rm snr} \right] \right]
    \\
    &= \int_\bbR \bbE[|\sfX|^k \mathbbm 1_{\{|\sfX| > M \}} | \sfY =y; \alpha \, {\rm snr}] \, f_{\sfY_\alpha}(y) \, \rmd y \\
     &= \int_\bbR  \frac{\bbE[|\sfX|^k \mathbbm 1_{\{|\sfX| > M \}} f_{\sfY |\sfX; \alpha \, {\rm snr}}(y|\sfX)]}{f_\sfY(y;\alpha \, {\rm snr})} \, f_{\sfY_\alpha}(y) \, \rmd y \\
     & \le  \int_\bbR  \frac{\bbE^{ \frac{q-1}{q} }[|\sfX|^{k  \frac{q}{q-1} } \mathbbm 1_{\{|\sfX| > M \}}]  \, \bbE^{ \frac{1}{q}  }[ f_{\sfY |\sfX; \alpha \, {\rm snr}}^q(y|\sfX)]}{f_\sfY(y;\alpha \, {\rm snr})} \, f_{\sfY_\alpha}(y) \, \rmd y  \label{eq:using_Holder}\\
     & \le  (\alpha \, \delta)^{\frac{q-1}{2q}} \, \bbE^{ \frac{q-1}{q}}[|\sfX|^{k  \frac{q}{q-1}} \mathbbm 1_{\{|\sfX| > M \}}]    \int_\bbR \frac{  \bbE^{ \frac{1}{q}  }[ f_{\sfY |\sfX; \alpha \, {\rm snr}}(y|\sfX)]}{f_\sfY(y;\alpha \, {\rm snr})} \, f_{\sfY_\alpha}(y) \, \rmd y \label{eq:bounding_noise} \\
     & =   (\alpha \, \delta)^{\frac{q-1}{2q}} \, \bbE^{ \frac{q-1}{q}}[ |\sfX|^{k  \frac{q}{q-1}} \mathbbm 1_{\{|\sfX| > M \}}] \exp \left( - \frac{\alpha-1}{\alpha} (I_\alpha(\sfX; {\rm snr}) -\Delta(\alpha, {\rm snr}  )  ) \right)  \int_\bbR  f_{\sfY}^{ \frac{1}{q} + \frac{1}{\alpha} -1}(y;\alpha \, {\rm snr}) \, \rmd y \label{eq:BoundWithAlphaMI} \\
     & =   (\alpha \, \delta)^{\frac{q-1}{2q}} \, \bbE^{ \frac{q-1}{q}}[ |\sfX|^{k  \frac{q}{q-1}} \mathbbm 1_{\{|\sfX| > M \}}]    \notag
     \\& \qquad \exp \left(  - \frac{\alpha-1}{\alpha} (I_\alpha(\sfX; {\rm snr}) -\Delta(\alpha, {\rm snr}  ) ) + \frac{\beta-1}{\beta} \left ( I_\beta \left (\sfX; \frac{\alpha}{\beta}{\rm snr} \right ) -\Delta \left (\beta,  \frac{\alpha}{\beta}{\rm snr} \right ) \right )  \right), \label{eq:last_MI_term}
\end{align}
where in~\eqref{eq:using_Holder} we have used the H\"older's inequality with $q>1$;~\eqref{eq:bounding_noise} follows from using the facts that $q-1>0$ and  $f_{\sfY |\sfX; \alpha \, {\rm snr}}^{q-1}(y|\sfX) \le (\alpha \, \delta)^{\frac{q-1}{2}}$ since $\rm{snr} \leq \delta$, for some $\delta>0$;~\eqref{eq:BoundWithAlphaMI} follows from using~\eqref{eq:PDFYGivenX} and the fact that  
\begin{equation}
f_{\sfY_\alpha}(y)=  \frac{f_\sfY^{\frac{1}{\alpha} }(y; \alpha \, {\rm snr})}{\exp \left( \frac{\alpha-1}{\alpha} ( I_\alpha(\sfX ;{\rm snr} ) -\Delta(\alpha, {\rm snr}) )\right)},
\end{equation}
which is a consequence of~\eqref{eq:PDFYalpha} and~\eqref{eq:AlphaMISecondRepresentation}; and~\eqref{eq:last_MI_term} follows from letting $\beta  = \frac{1}{  \frac{1}{q} +\frac{1}{\alpha} -1}$ and noting that $\int_\bbR  f_{\sfY}^{ \frac{1}{q} + \frac{1}{\alpha} -1}(y;\alpha \, {\rm snr}) \, \rmd y =  \int_\bbR  f_{\sfY}^{ \frac{1}{\beta} }(y;\alpha \, {\rm snr}) \, \rmd y = \exp \left ( \frac{\beta-1}{\beta} \left( I_\beta \left (\sfX; \frac{\alpha}{\beta}{\rm snr} \right) -\Delta  \left (\beta,  \frac{\alpha}{\beta}{\rm snr} \right ) \right ) \right )$ from~\eqref{eq:AlphaMISecondRepresentation}. 

In order to guarantee that~\eqref{eq:last_MI_term} vanishes as $M \to \infty$, we require the following conditions:
\begin{subequations}
\begin{align}
     \bbE \left [ |\sfX|^{k   \frac{q}{q-1} } \right ] & <\infty, \\
    I_{ \alpha } (\sfX; {\rm snr}) &<\infty, \\
    I_{ \beta } \left (\sfX; \frac{\alpha}{\beta} {\rm snr} \right ) &<\infty, \\
   \beta= \frac{1}{\frac{1}{q} +\frac{1}{\alpha} -1}&>0. 
\end{align}
\label{eq:conditions_for_finitness}
\end{subequations}
To guarantee the above conditions, choose some $\epsilon \in (0, \frac{1}{\alpha})$ and let
\begin{equation}
\frac{1}{q}= \frac{\alpha-1}{\alpha} +\epsilon \Rightarrow  \frac{1}{\beta}=\frac{1}{q} +\frac{1}{\alpha}-1=\epsilon >0.
\end{equation}
This, together with Proposition~\ref{prop:moment_conditions}, leads to the following implications: 
\begin{align}
    \bbE \left[ |\sfX|^{k   \frac{q}{q-1} } \right] &= \bbE \left [  |\sfX|^{   \frac{k}{ \frac{1}{\alpha} -\epsilon} } \right] <\infty,\\
    I_{ \beta } (\sfX; {\rm snr}) &<\infty  \Longleftarrow   \bbE \left [ | \sfX|^{  \frac{1}{\epsilon} +\rho }\right] <\infty, \\
    I_{ \alpha } (\sfX; {\rm snr}) &<\infty  \Longleftarrow  \bbE \left [ | \sfX|^{  \alpha +\rho }\right] <\infty,
\end{align}
for any small $\rho>0$. 
Since $ \epsilon < \frac{1}{\alpha}$ all above moment conditions are satisfied provided that 
\begin{equation}
\label{eq;MaxFinalStepMInfinity}
\max \left \{\bbE \left [  |\sfX|^{   \frac{k}{ \frac{1}{\alpha} -\epsilon} } \right],\bbE \left [ | \sfX|^{  \frac{1}{\epsilon} +\rho }\right] \right \} < \infty. 
\end{equation}
Choosing $\epsilon =\frac{1}{2\alpha}$ leads to  the condition 
\begin{equation}
 \max \left \{   \bbE \left [  |\sfX|^{   2 k \alpha } \right],  \bbE \left[ | \sfX|^{2 \alpha +\rho} \right] \right \} <\infty \Longleftarrow   \bbE \left [  |\sfX|^{   4 \alpha } \right] <\infty,
\end{equation}
where the last implications follows by taking $\rho $ small enough and considering the stringer condition at $k=2$.

Combining \eqref{eq:first_truncated_moment} and \eqref{eq:last_MI_term} and taking $M \to \infty$, we conclude that  for $k \in  \{1,2 \}$
\begin{equation}
\lim_{{\rm snr} \to 0} \bbE [ \sfX_\alpha^k ] = \bbE [ \sfX^k ].
\end{equation}
This concludes the proof of~\eqref{eq:VarXalphaToVarX}. 
\end{enumerate}

\section{Proof of Proposition~\ref{prop:m-point}}
\label{app:ProofPropm-point}
We start by noting that when $|\cX|=1$, then $\sfX$ is deterministic and $I_\alpha(\sfX;\mathrm{snr})=
H_{\frac1\alpha}(\sfX)=0$, so the result is trivial. In what follows we therefore assume that $|\cX|\geq2$, and fix $\alpha\in(0,1)\cup(1,\infty)$. If $|\cX|=N<\infty$, we write
$\cX=\{x_1,\ldots,x_N\}$, otherwise, we write
$\cX=\{x_1,x_2,\ldots\}$, where the $x_j$'s are distinct. We let
$p_j=\Pr(\sfX=x_j)$. Fix a finite $m\geq2$ such that $x_1,\ldots,x_m$ are defined, and let
\begin{equation}
d_m
=
\min_{\substack{1\leq i,j\leq m\\i\neq j}}
|x_i-x_j|
>0.
\label{eq:Separation_dist}
\end{equation}
Following~\cite[Definition~5.1]{yuksel2012optimization}, for every fixed finite $m$, consider the partition
\begin{equation}
B_{m,j}
=
\{x_j\},
\qquad
j\in\{1,\ldots,m\},
\label{eq:qmSingletonCells}
\end{equation}
and
\begin{equation}
B_{m,0}
=
\cX\setminus\{x_1,\ldots,x_m\}.
\label{eq:qmTailCell}
\end{equation}
Let $\sfT_m\in\{0,1,\ldots,m\}$ denote
the output of the deterministic quantizer characterized by
\begin{equation}
P_{\sfT_m|\sfX}(j|x)
=
\mathbbm{1}_{\{x\in B_{m,j}\}},
\qquad
j\in\{0,1,\ldots,m\}.
\label{eq:TmQuantizer}
\end{equation}
Consequently, $\Pr(\sfT_m=j)=p_j, j\in\{1,\ldots,m\}$ and $\Pr(\sfT_m=0)=\sum_{j>m}p_j$, where the latter is zero when $m=|\cX|<\infty$. Since $\sfT_m\to\sfX\to\sfY$ is a Markov chain, by the data-processing inequality~\cite[Theorem~3]{verdu2015alpha} we have,
\begin{equation}
I_\alpha(\sfX;\mathrm{snr})
=
I_\alpha(\sfX;\sfY)
\geq
I_\alpha(\sfT_m;\sfY).
\label{eq:DPIXtoTm}
\end{equation}
For this fixed $m$, define an $\mathrm{snr}$-dependent
decoder
\begin{equation}
\widehat{\sfT}_{m,\mathrm{snr}}
=
\sum_{j=1}^{m}
j\,\mathbbm{1}
\left\{
|\sfY-x_j|
<
\frac{r_{\mathrm{snr}}}{2}
\right\},
\label{eq:TmDecoder}
\end{equation}
where we let
\begin{equation}
r_{\mathrm{snr}}
=
\min\left\{
\frac{d_m}{4},
\mathrm{snr}^{-\frac14}
\right\}.
\label{eq:RadiusChoice}
\end{equation}
We now show that, for every fixed finite $m$,
\begin{equation}
\Pr\left(
\widehat{\sfT}_{m,\mathrm{snr}}\neq\sfT_m
\right)
\to0,
\qquad
\mathrm{snr}\to\infty.
\label{eq:TmErrorToZero}
\end{equation}
By the law of total probability and the definition of $\sfT_m$,
\begin{align}
\Pr\left(
\widehat{\sfT}_{m,\mathrm{snr}}\neq\sfT_m
\right)
=
\sum_{j=1}^{m}
p_j
\Pr\left(
\widehat{\sfT}_{m,\mathrm{snr}}\neq j
\,\middle|\,
\sfX=x_j
\right)
+
\sum_{i>m}
p_i
\Pr\left(
\widehat{\sfT}_{m,\mathrm{snr}}\neq0
\,\middle|\,
\sfX=x_i
\right).
\label{eq:ErrorDecomposition}
\end{align}
For $j\in\{1,\ldots,m\}$, the event $\left\{
|\sfY-x_j|<\frac{r_{\mathrm{snr}}}{2}
\right\}$ implies $\widehat{\sfT}_{m,\mathrm{snr}}=j$. Hence, using
$\sfY-x_j=\sfZ/\sqrt{\mathrm{snr}}$ conditioned on
$\{\sfX=x_j\}$,
\begin{align}
\Pr\left(
\widehat{\sfT}_{m,\mathrm{snr}}\neq j
\,\middle|\,
\sfX=x_j
\right)
&\leq
\Pr\left(
|\sfY-x_j|
\geq
\frac{r_{\mathrm{snr}}}{2}
\,\middle|\,
\sfX=x_j
\right)
\notag\\
&=
2Q\left(
\frac{r_{\mathrm{snr}}\sqrt{\mathrm{snr}}}{2}
\right).
\label{eq:CoreErrorTm}
\end{align}
Next, consider $x_i\in B_{m,0}$ such that
\begin{equation}
\min_{1\leq j\leq m}|x_i-x_j|
\geq
r_{\mathrm{snr}}.
\end{equation}
If $\widehat{\sfT}_{m,\mathrm{snr}}\neq0$, then
$|\sfY-x_j|<r_{\mathrm{snr}}/2$ for some
$j\in\{1,\ldots,m\}$. Therefore,
\begin{align}
|\sfY-x_i|
&\geq
|x_i-x_j|-|\sfY-x_j|
\\
&>
r_{\mathrm{snr}}
-\frac{r_{\mathrm{snr}}}{2}
=
\frac{r_{\mathrm{snr}}}{2},
\end{align}
and consequently
\begin{align}
\Pr\left(
\widehat{\sfT}_{m,\mathrm{snr}}\neq0
\,\middle|\,
\sfX=x_i
\right)
&\leq
\Pr\left(
|\sfY-x_i|
>
\frac{r_{\mathrm{snr}}}{2}
\,\middle|\,
\sfX=x_i
\right)
\notag\\
&=
2Q\left(
\frac{r_{\mathrm{snr}}\sqrt{\mathrm{snr}}}{2}
\right).
\label{eq:TailFarError}
\end{align}
For the remaining $x_i\in B_{m,0}$, i.e., those satisfying $\min_{1\leq j\leq m}|x_i-x_j|
<
r_{\mathrm{snr}}$, we simply bound the conditional error probability by one. Substituting these bounds into~\eqref{eq:ErrorDecomposition} yields
\begin{align}
\Pr\left(
\widehat{\sfT}_{m,\mathrm{snr}}\neq\sfT_m
\right)
\leq
2Q\left(
\frac{r_{\mathrm{snr}}\sqrt{\mathrm{snr}}}{2}
\right)
+
\Pr\left(
\sfX\in B_{m,0},
\;
\min_{1\leq j\leq m}
|\sfX-x_j|
<
r_{\mathrm{snr}}
\right).
\label{eq:TotalErrorBound}
\end{align}
The first term in~\eqref{eq:TotalErrorBound} converges to zero since
$r_{\mathrm{snr}}\sqrt{\mathrm{snr}}\to\infty$.
Since $r_{\mathrm{snr}}\downarrow0$, the events
\begin{equation}
\left\{
\sfX\in B_{m,0},
\;
\min_{1\leq j\leq m}
|\sfX-x_j|
<
r_{\mathrm{snr}}
\right\}
\end{equation}
decrease to the empty set as $\mathrm{snr}\to\infty$ (note that every $x\in B_{m,0}$ satisfies $\min_{1\leq j\leq m}|x-x_j|>0$, since $x$ is distinct from each of the finitely many retained points).
Hence, by continuity from the above of probability measures,
\begin{equation}
\Pr\left(
\sfX\in B_{m,0},
\;
\min_{1\leq j\leq m}
|\sfX-x_j|
<
r_{\mathrm{snr}}
\right)
\to0.
\label{eq:TailNearVanishing}
\end{equation}
Together with~\eqref{eq:TotalErrorBound}, this proves
\eqref{eq:TmErrorToZero}. Since $\widehat{\sfT}_{m,\mathrm{snr}}$ is a deterministic function
of $\sfY$, then we have the following Markov chain $\sfT_m
\to
\sfY
\to
\widehat{\sfT}_{m,\mathrm{snr}}$, and thus again by the data-processing
inequality,
\begin{equation}
I_\alpha(\sfX;\mathrm{snr})
\geq
I_\alpha(\sfT_m;\sfY)
\geq
I_\alpha\left(
\sfT_m;\widehat{\sfT}_{m,\mathrm{snr}}
\right).
\label{eq:UnifiedDPI}
\end{equation}
By using the discrete representation of Sibson's $\alpha$-mutual information~\cite[eq.(30)]{esposito2025sibson}, we can write
\begin{align}
&\exp\left(
\frac{\alpha-1}{\alpha}
I_\alpha\left(
\sfT_m;\widehat{\sfT}_{m,\mathrm{snr}}
\right)
\right)
=
\sum_{k=0}^{m}
\left(
\sum_{j=0}^{m}
\Pr(\sfT_m=j)
\Pr\left(
\widehat{\sfT}_{m,\mathrm{snr}}=k
\,\middle|\,
\sfT_m=j
\right)^\alpha
\right)^{\frac1\alpha}.
\label{eq:DiscreteSibsonTm}
\end{align}
For every $j\in\{0,\ldots,m\}$ with $\Pr(\sfT_m=j)>0$,
\begin{align}
\Pr\left(
\widehat{\sfT}_{m,\mathrm{snr}}\neq j
\,\middle|\,
\sfT_m=j
\right)
&=
\frac{
\Pr\left(
\{\widehat{\sfT}_{m,\mathrm{snr}}\neq\sfT_m\}
\cap
\{\sfT_m=j\}
\right)
}{
\Pr(\sfT_m=j)
}\\
&\leq
\frac{
\Pr\left(
\widehat{\sfT}_{m,\mathrm{snr}}\neq\sfT_m
\right)
}{
\Pr(\sfT_m=j)
}
\to 0 \qquad \mathrm{snr} \to \infty,
\label{eq:ConditionalErrorToZero}
\end{align}
by~\eqref{eq:TmErrorToZero}. Hence
\begin{equation}
\Pr\left(
\widehat{\sfT}_{m,\mathrm{snr}}=k
\,\middle|\,
\sfT_m=j
\right)
\to
\mathbbm{1}_{\{k=j\}},
\end{equation}
while terms for which $\Pr(\sfT_m=j)=0$ vanish identically in
\eqref{eq:DiscreteSibsonTm}. Since the sums are finite,
\begin{align}
\exp\left(
\frac{\alpha-1}{\alpha}
I_\alpha\left(
\sfT_m;\widehat{\sfT}_{m,\mathrm{snr}}
\right)
\right)
\to
\sum_{k=0}^{m}
\Pr(\sfT_m=k)^{\frac1\alpha}.
\end{align}
It follows that
\begin{align}
\lim_{\mathrm{snr}\to\infty}
I_\alpha\left(
\sfT_m;\widehat{\sfT}_{m,\mathrm{snr}}
\right)
&=
H_{\frac1\alpha}(\sfT_m)
=
\frac{\alpha}{\alpha-1}
\log\left(
\sum_{j=1}^{m}
p_j^{\frac1\alpha}
+
\left(
\sum_{j>m}p_j
\right)^{\frac1\alpha}
\right).
\label{eq:TmSibsonLimit}
\end{align}
Consequently, for every fixed finite $m$,
\begin{equation}
\liminf_{\mathrm{snr}\to\infty}
I_\alpha(\sfX;\mathrm{snr})
\geq
H_{\frac1\alpha}(\sfT_m).
\label{eq:FiniteMLowerBound1}
\end{equation}
We now let $m\to\infty$. If $|\cX|=N<\infty$, taking $m=N$ gives $H_{\frac1\alpha}(\sfT_N)
=H_{\frac1\alpha}(\sfX)$. If $\cX$ is countably infinite, then in the extended real sense,
\begin{equation}
H_{\frac1\alpha}(\sfT_m)
\to
H_{\frac1\alpha}(\sfX).
\label{eq:TmEntropyToXEntropy}
\end{equation}
Since the left-hand side of~\eqref{eq:FiniteMLowerBound1} is
independent of $m$, we obtain
\begin{equation}
\liminf_{\mathrm{snr}\to\infty}
I_\alpha(\sfX;\mathrm{snr})
\geq
H_{\frac1\alpha}(\sfX).
\label{eq:FinalLowerBound}
\end{equation}
The proof is completed by using the data-processing inequality for $\alpha$-mutual information~\cite[Theorem~3]{verdu2015alpha} which implies that 
\begin{equation}
    I_\alpha(\sfX;\sfY) \le I_\alpha(\sfX;\sfX) =   H_{ \frac{1}{\alpha}} (\sfX) ,
\end{equation}
where the last equality follows from~\cite[Example~2]{verdu2015alpha}. This concludes the proof of  Proposition~\ref{prop:m-point}.

\section{Proof of Lemma~\ref{lemma:BoundsUnifEpsWith3Eps}}
\label{app:BoundsUnifEpsWith3Eps}
The lower bound in~\eqref{eq:BoundUnifEpsWith3Eps} follows from the data-processing inequality, which holds since we can build a coupling $(\sfX,\sfU_1,\sfU_2)$ such that $\sfX \to \sfX+\sfU_1 \to \sfX +\sfU_2$ is a Markov chain, and $\sfU_1$ has the same distribution as $\varepsilon \sfU$, and $\sfU_2$ has the same distribution as $3\varepsilon \sfU$.
We now focus on the upper bound in~\eqref{eq:BoundUnifEpsWith3Eps} by considering the cases of $\alpha<1$ and $\alpha>1$ separately.
For $\alpha \in (1,+\infty)$, from~\eqref{eq:ExpIAlpha1} in Lemma~\ref{lemma:UniformMIProb}, we have that
\begin{align}
I_\alpha(\sfX;\sfY_\varepsilon) - I_\alpha(\sfX;\sfY_{3\varepsilon}) &= \frac{\alpha}{\alpha-1} \log \left( \frac{\int_\bbR P_{\sfX}^{\frac{1}{\alpha}} \left(\cB_\infty \left (y,\frac{\varepsilon}{2} \right )\right) \, \rmd y}{\int_\bbR P_{\sfX}^{\frac{1}{\alpha}} \left(\cB_\infty \left (y,\frac{3}{2} \varepsilon\right )\right) \, \rmd y} \right ) + \frac{\alpha}{\alpha-1} \log(3)
\\& \leq \frac{\alpha}{\alpha-1} \log(3),
\end{align}
where the inequality follows since  $\cB_\infty \left (y,\frac{\varepsilon}{2} \right ) \subseteq \cB_\infty \left (y,\frac{3}{2} \varepsilon \right )$, which implies that $P_{\sfX} \left(\cB_\infty \left (y,\frac{\varepsilon}{2} \right ) \right ) \leq P_{\sfX} \left(\cB_\infty \left (y,\frac{3}{2}\varepsilon \right ) \right )$.

For $\alpha \in (0,1)$, we have that
\begin{align}
&\frac{\exp \left( \frac{1-\alpha}{\alpha} I_\alpha(\sfX;\sfY_\varepsilon)\right )}{\exp \left( \frac{1-\alpha}{\alpha} I_\alpha(\sfX;\sfY_{3\varepsilon})\right )} 
\\&= \frac{\int_\bbR P_{\sfX}^{\frac{1}{\alpha}} \left(\cB_\infty \left (y,\frac{3}{2}\varepsilon \right )\right) \, \rmd y }{3\int_\bbR P_{\sfX}^{\frac{1}{\alpha}} \left(\cB_\infty \left (y,\frac{1}{2}\varepsilon \right )\right) \, \rmd y} \label{eq:RatioExpMI}
\\& \leq \frac{\int_\bbR \left( P_{\sfX} \left(\cB_{\infty} \left( y - \varepsilon,  \frac{1}{2} \varepsilon \right) \right ) + P_{\sfX} \left(\cB_{\infty} \left( y ,  \frac{1}{2} \varepsilon \right) \right ) + P_{\sfX} \left( \cB_{\infty} \left( y +  \varepsilon, \frac{1}{2} \varepsilon\right) \right ) \right )^{\frac{1}{\alpha}}\, \rmd y}{3\int_\bbR P_{\sfX}^{\frac{1}{\alpha}} \left(\cB_\infty \left (y,\frac{1}{2}\varepsilon \right )\right) \, \rmd y} \label{eq:CoveringUBNew}
\\& \leq 3^{\frac{1-2 \alpha}{\alpha}} \left( \frac{\int_\bbR P_{\sfX}^{\frac{1}{\alpha}} \left(\cB_{\infty} \left( y - \varepsilon,  \frac{1}{2} \varepsilon \right) \right ) \, \rmd y}{\int_\bbR P_{\sfX}^{\frac{1}{\alpha}} \left(\cB_\infty \left (y,\frac{1}{2}\varepsilon \right )\right) \, \rmd y} + \frac{\int_\bbR P_{\sfX}^{\frac{1}{\alpha}} \left(\cB_{\infty} \left( y ,  \frac{1}{2} \varepsilon \right) \right ) \, \rmd y}{\int_\bbR P_{\sfX}^{\frac{1}{\alpha}} \left(\cB_\infty \left (y,\frac{1}{2}\varepsilon \right )\right) \, \rmd y}+\frac{\int_\bbR P_{\sfX}^{\frac{1}{\alpha}} \left( \cB_{\infty} \left( y +  \varepsilon, \frac{1}{2} \varepsilon\right) \right ) \, \rmd y}{\int_\bbR P_{\sfX}^{\frac{1}{\alpha}} \left(\cB_\infty \left (y,\frac{1}{2}\varepsilon \right )\right) \, \rmd y}\right ) \label{eq:InequalityExpSumNew}
\\& = 3^{\frac{1-\alpha}{\alpha}} \label{eq:ChangeVarNew},
\end{align}
where: in~\eqref{eq:RatioExpMI}, we have used~\eqref{eq:ExpIAlpha1} in Lemma~\ref{lemma:UniformMIProb};~\eqref{eq:CoveringUBNew} is due to the fact that
\begin{equation}
\cB_{\infty} \left( y, \frac{3}{2} \varepsilon\right) = \cB_{\infty} \left( y - \varepsilon,  \frac{1}{2} \varepsilon \right) \cup \cB_{\infty} \left( y ,  \frac{1}{2} \varepsilon \right) \cup \cB_{\infty} \left( y +  \varepsilon, \frac{1}{2} \varepsilon\right),
\end{equation}
and hence,
\begin{equation}
P_{\sfX} \left(\cB_\infty \left (y,\frac{3}{2}\varepsilon \right )\right) \leq P_{\sfX} \left(\cB_{\infty} \left( y - \varepsilon,  \frac{1}{2} \varepsilon \right) \right ) + P_{\sfX} \left(\cB_{\infty} \left( y ,  \frac{1}{2} \varepsilon \right) \right ) + P_{\sfX} \left( \cB_{\infty} \left( y +  \varepsilon, \frac{1}{2} \varepsilon\right) \right );
\end{equation}
in~\eqref{eq:InequalityExpSumNew}, we have used the fact that, for $p>1$ and $a,b,c\geq 0$, it holds that $(a+b+c)^p \leq 3^{p-1}(a^p+b^p+c^p)$;
and in~\eqref{eq:ChangeVarNew}, we have used change of variable.
This concludes the proof of Lemma~\ref{lemma:BoundsUnifEpsWith3Eps}.

{
\section{Proof of Lemma~\ref{lemma:UBIAlphaUniWithG}}
\label{app:UBIAlphaUniWithG}
We start by proving the upper bound in~\eqref{eq:UBIAlphaUniWithG}. By applying the definition of $\alpha$-mutual information, we have that
\begin{align}
I_\alpha(\sfX;\sfY_{\varepsilon}) &\leq D_{\alpha}\left( P_{\sfX,\sfY_{\varepsilon}} \| P_{\sfX} \times Q_{\sfY,s} \right ) \label{eq:UBAlphaMIDiv}
\\& = \frac{1}{\alpha-1} \log \left(\int_{\bbR^2} f^\alpha_{\sfY_{\varepsilon}|\sfX}(y|x) \, \rmd P_\sfX(x) \, q^{1-\alpha}_{s}(y)  \, \rmd y \right ) 
\\& = \frac{1}{\alpha-1} \log \left(\left( \frac{1}{\varepsilon} \right )^\alpha  \bbE \left [ \int_{|y-\sfX|\le \tfrac{\varepsilon}{2}} q^{1-\alpha}_{s}(y) \, \rmd y\right ] \right ) \label{eq:Step2UBMIUnif2}
\\& = \frac{\alpha}{\alpha-1} \log \left( \frac{1}{ \varepsilon} \right ) + \log \left( \int_{\bbR}P_{\sfX}^{\frac{1}{\alpha}} \left(\cB_\infty \left (t,\frac{3}{2}\varepsilon \right )\right){\rm d}t \right ) \notag
\\& \quad + \frac{1}{\alpha-1} \log \left( \bbE \left [ \int_{|y-\sfX|\le \tfrac{\varepsilon}{2}} P_{\sfX}^{\frac{1-\alpha}{\alpha}} \left(\cB_\infty \left (y,\frac{3}{2}\varepsilon \right )\right)
  \, \rmd y\right ] \right )
  \\& \leq  \frac{\alpha}{\alpha-1} \log \left( \frac{1}{ \varepsilon} \right ) + \log \left( \int_{\bbR}P_{\sfX}^{\frac{1}{\alpha}} \left(\cB_\infty \left (t,\frac{3}{2}\varepsilon \right )\right){\rm d}t \right ) \notag
\\& \quad + \frac{1}{\alpha-1} \log \left( \bbE \left [ \int_{|y-\sfX|\le \tfrac{\varepsilon}{2}} P_{\sfX}^{\frac{1-\alpha}{\alpha}} \left(\cB_\infty(\sfX,\varepsilon)\right)
  \, \rmd y\right ] \right )  \label{eq:Step3UBMIUnif2}
\\& =   \log \left( \frac{1}{\varepsilon} \right ) + \log \left( \int_{\bbR}P_{\sfX}^{\frac{1}{\alpha}} \left(\cB_\infty \left (t,\frac{3}{2}\varepsilon \right )\right){\rm d}t \right ) \notag
\\& \quad + \frac{1}{\alpha-1} \log \left(  \bbE \left [ P_{\sfX}^{\frac{1-\alpha}{\alpha}} \left(\cB_\infty(\sfX,\varepsilon)\right)  \right ] \right )
\\& =   \log \left( \frac{1}{\varepsilon} \right ) + \log \left( \int_{\bbR}P_{\sfX}^{\frac{1}{\alpha}} \left(\cB_\infty \left (t,\frac{3}{2}\varepsilon \right )\right){\rm d}t \right ) + \frac{1}{\alpha} g_{\alpha} \left (\sfX,\varepsilon \right ) 
\\& =  \log (3)+ \frac{\alpha-1}{\alpha} I_\alpha \left (\sfX; \sfY_{3\varepsilon} \right ) + \frac{1}{\alpha} g_{\alpha} \left (\sfX,\varepsilon \right ),\label{eq:inserting_a_mi2}
\end{align}
where: in~\eqref{eq:UBAlphaMIDiv}, we let $Q_{\sfY,s}$ have density
\begin{equation}
q_{s}(y)=\frac{P_{\sfX}^{\frac{1}{\alpha}} \left(\cB_\infty(y,s)\right)}
{\displaystyle \int_{\bbR}P_{\sfX}^{\frac{1}{\alpha}} \left(\cB_\infty(t,s)\right){\rm d}t },
\label{eq:qs_def_ub2}
\end{equation}
with $s= \frac{3}{2}\varepsilon$;~\eqref{eq:Step2UBMIUnif2} follows from Fubini-Tonelli's theorem;~\eqref{eq:Step3UBMIUnif2} follows since, under the assumption $|y-\sfX|\le \tfrac{\varepsilon}{2}$, we have that $\cB_\infty \left (y,\frac{3}{2}\varepsilon \right )\supseteq \cB_\infty \left (\sfX,\varepsilon\right )$, which implies that
\begin{equation}
P_{\sfX}\left(\cB_\infty \left (y,\frac{3}{2}\varepsilon \right )\right)\ge P_{\sfX}\left(\cB_\infty(\sfX,\varepsilon)\right),
\end{equation}
and using the fact that the mapping $s\mapsto s^{\frac{1}{\alpha}-1}$, with $s \in (0,1)$, is decreasing for $\alpha \in (1,+\infty)$ (and the pre-log factor $\frac{1}{\alpha-1}$ is positive) and increasing for $\alpha \in (0,1)$ (and the pre-log factor $\frac{1}{\alpha-1}$ is negative); and~\eqref{eq:inserting_a_mi2} is due to~\eqref{eq:ExpIAlpha1} in Lemma~\ref{lemma:UniformMIProb}. Rearranging~\eqref{eq:inserting_a_mi2} proves the upper bound in~\eqref{eq:UBIAlphaUniWithG}.

We now prove the lower bound in~\eqref{eq:UBIAlphaUniWithG}. From~\eqref{eq:ExpIAlpha2} in Lemma~\ref{lemma:UniformMIProb}, we have that
\begin{align}
\label{eq:IaAlphaGreatGLastStep}
I_\alpha(\sfX;\sfY_\varepsilon)
=
\frac{\alpha}{\alpha-1}\log \left( \bbE \left [ \left( \frac{1}{P_{\sfX}\left(\cB_\infty \left (\sfY_{\varepsilon}, \frac{1}{2} \varepsilon \right )\right)}\right )^{\frac{\alpha-1}{\alpha}}\right ] \right ).
\end{align}
The fact that $|\sfU|\le \frac{1}{2}$  implies that $|\sfY_\varepsilon-\sfX|=|\varepsilon \sfU|\le \frac{\varepsilon}{2}$ and that
\begin{equation}
\cB_\infty(\sfY_\varepsilon,\varepsilon/2)\subseteq \cB_\infty(\sfX,\varepsilon).
\label{eq:ball_inclusion_center_shift}
\end{equation}
Consequently,
\begin{equation}
P_\sfX \left( \cB_\infty \left (\sfY_\varepsilon, \frac{\varepsilon}{2} \right ) \right ) \le
P_\sfX(\cB_\infty(\sfX,\varepsilon)),
\end{equation}
and also:
\begin{enumerate}
\item If $\alpha \in (1,+\infty)$, then the mapping $s\mapsto s^{-\frac{\alpha-1}{\alpha}}$ is decreasing for $s \in (0,1)$ and hence,
\begin{equation}
\left(
\frac{1}{P_\sfX(\cB_\infty(\sfY_\varepsilon,\varepsilon/2))}
\right)^{\frac{\alpha-1}{\alpha}}
\ge
\left(
\frac{1}{P_\sfX(\cB_\infty(\sfX,\varepsilon))}
\right)^{\frac{\alpha-1}{\alpha}};
\end{equation}
\item If $\alpha \in (0,1)$, then the mapping $s\mapsto s^{-\frac{\alpha-1}{\alpha}}$ is increasing for $s \in (0,1)$ and hence,
\begin{equation}
\left(
\frac{1}{P_\sfX(\cB_\infty(\sfY_\varepsilon,\varepsilon/2))}
\right)^{\frac{\alpha-1}{\alpha}}
\le
\left(
\frac{1}{P_\sfX(\cB_\infty(\sfX,\varepsilon))}
\right)^{\frac{\alpha-1}{\alpha}}.
\end{equation}
\end{enumerate}
Substituting the above inside~\eqref{eq:IaAlphaGreatGLastStep} (and noticing that $\frac{\alpha}{\alpha-1}<0$ for $\alpha \in (0,1)$), we arrive at
\begin{align}
I_\alpha(\sfX;\sfY_\varepsilon)
\geq
\frac{\alpha}{\alpha-1}\log \left( \bbE \left [ \left(
\frac{1}{P_\sfX(\cB_\infty(\sfX,\varepsilon))}
\right)^{\frac{\alpha-1}{\alpha}}\right ] \right ) = g_\alpha(\sfX,\varepsilon).
\end{align}
This prove the lower bound in~\eqref{eq:UBIAlphaUniWithG} and concludes the proof of Lemma~\ref{lemma:UBIAlphaUniWithG}.
}

\bibliographystyle{IEEEtran}
\bibliography{refs.bib}

\begin{thebibliography}{10}
\providecommand{\url}[1]{#1}
\csname url@samestyle\endcsname
\providecommand{\newblock}{\relax}
\providecommand{\bibinfo}[2]{#2}
\providecommand{\BIBentrySTDinterwordspacing}{\spaceskip=0pt\relax}
\providecommand{\BIBentryALTinterwordstretchfactor}{4}
\providecommand{\BIBentryALTinterwordspacing}{\spaceskip=\fontdimen2\font plus
\BIBentryALTinterwordstretchfactor\fontdimen3\font minus
  \fontdimen4\font\relax}
\providecommand{\BIBforeignlanguage}[2]{{%
\expandafter\ifx\csname l@#1\endcsname\relax
\typeout{** WARNING: IEEEtran.bst: No hyphenation pattern has been}%
\typeout{** loaded for the language `#1'. Using the pattern for}%
\typeout{** the default language instead.}%
\else
\language=\csname l@#1\endcsname
\fi
#2}}
\providecommand{\BIBdecl}{\relax}
\BIBdecl

\bibitem{renyi1961measures}
A.~R{\'e}nyi, ``On {M}easures of {E}ntropy and {I}nformation,'' in
  \emph{Proceedings of the Fourth Berkeley Symposium on Mathematical Statistics
  and Probability, Volume 1: Contributions to the Theory of Statistics}.\hskip
  1em plus 0.5em minus 0.4em\relax University of California Press, 1961, pp.
  547--562.

\bibitem{courtade2014cumulant}
T.~A. Courtade and S.~Verd{\'u}, ``Cumulant {G}enerating {F}unction of
  {C}odeword {L}engths in {O}ptimal {L}ossless {C}ompression,'' in \emph{2014
  IEEE International Symposium on Information Theory}, 2014, pp. 2494--2498.

\bibitem{blahut1974hypothesis}
R.~Blahut, ``Hypothesis {T}esting and {I}nformation {T}heory,'' \emph{IEEE
  Transactions on Information Theory}, vol.~20, no.~4, pp. 405--417, 1974.

\bibitem{han1989strong}
T.~S. Han and K.~Kobayashi, ``The {S}trong {C}onverse {T}heorem for
  {H}ypothesis {T}esting,'' \emph{IEEE Transactions on Information Theory},
  vol.~35, no.~1, pp. 178--180, 1989.

\bibitem{ben2003renyi}
M.~Ben-Bassat and J.~Raviv, ``R{\'e}nyi's {E}ntropy and the {P}robability of
  {E}rror,'' \emph{IEEE Transactions on Information Theory}, vol.~24, no.~3,
  pp. 324--331, 2003.

\bibitem{sason2017arimoto}
I.~Sason and S.~Verd{\'u}, ``Arimoto-{R}{\'e}nyi {C}onditional {E}ntropy and
  {B}ayesian $ m $-ary {H}ypothesis {T}esting,'' \emph{IEEE Transactions on
  Information theory}, vol.~64, no.~1, pp. 4--25, 2017.

\bibitem{bruno2026finite}
R.~Bruno, A.~Vandenbroucque, and A.~R. Esposito, ``A {F}inite-{S}ample {S}trong
  {C}onverse for {B}inary {H}ypothesis {T}esting via ({R}everse) {R}\'enyi
  {D}ivergence,'' \emph{arXiv preprint arXiv:2601.09550}, 2026.

\bibitem{arikan2002inequality}
E.~Arikan, ``An {I}nequality on {G}uessing and its {A}pplication to
  {S}equential {D}ecoding,'' \emph{IEEE Transactions on Information Theory},
  vol.~42, no.~1, pp. 99--105, 2002.

\bibitem{rioul2021primer}
O.~Rioul, ``A {P}rimer on {A}lpha-{I}nformation {T}heory with {A}pplication to
  {L}eakage in {S}ecrecy {S}ystems,'' in \emph{International Conference on
  Geometric Science of Information}.\hskip 1em plus 0.5em minus 0.4em\relax
  Springer, 2021, pp. 459--467.

\bibitem{arimoto1977information}
S.~Arimoto, ``Information {M}easures and {C}apacity of {O}rder $\alpha$ for
  {D}iscrete {M}emoryless {C}hannels,'' \emph{Topics in Information Theory},
  1977.

\bibitem{sibson1969information}
R.~Sibson, ``Information {R}adius,'' \emph{Zeitschrift f{\"u}r
  Wahrscheinlichkeitstheorie und verwandte Gebiete}, vol.~14, no.~2, pp.
  149--160, 1969.

\bibitem{csiszar2002generalized}
I.~Csisz{\'a}r, ``Generalized {C}utoff {R}ates and {R}{\'e}nyi's {I}nformation
  {M}easures,'' \emph{IEEE Transactions on Information Theory}, vol.~41, no.~1,
  pp. 26--34, 2002.

\bibitem{Augustin1978NoisyChannels}
U.~Augustin, ``Noisy {C}hannels,'' Ph.D. dissertation, Universit{\"a}t
  Erlangen--N{\"u}rnberg, 1978, {H}abilitation {T}hesis.

\bibitem{LapidothPfister}
A.~Lapidoth and C.~Pfister, ``Two {M}easures of {D}ependence,'' \emph{Entropy},
  vol.~21, no.~8, 2019.

\bibitem{esposito2025sibson}
A.~R. Esposito, M.~Gastpar, and I.~Issa, ``Sibson $\alpha$-{M}utual
  {I}nformation and {I}ts {V}ariational {R}epresentations,'' \emph{IEEE
  Transactions on Information Theory}, pp. 1--1, 2025.

\bibitem{polyanskiy2010arimoto}
Y.~Polyanskiy and S.~Verd{\'u}, ``Arimoto {C}hannel {C}oding {C}onverse and
  {R}{\'e}nyi {D}ivergence,'' in \emph{2010 48th Annual Allerton Conference on
  Communication, Control, and Computing (Allerton)}, 2010, pp. 1327--1333.

\bibitem{esposito2021generalization}
A.~R. Esposito, M.~Gastpar, and I.~Issa, ``Generalization {E}rror {B}ounds via
  {R}{\'e}nyi-, f-{D}ivergences and {M}aximal {L}eakage,'' \emph{IEEE
  Transactions on Information Theory}, vol.~67, no.~8, pp. 4986--5004, 2021.

\bibitem{saito2022meta}
S.~Saito, ``On {M}eta-{B}ound for {L}ower {B}ounds of {B}ayes {R}isk,'' in
  \emph{2022 IEEE International Symposium on Information Theory (ISIT)}, 2022,
  pp. 3162--3167.

\bibitem{GuoIT2005}
D.~Guo, S.~Shamai~(Shitz), and S.~Verd\'u, ``Mutual {I}nformation and {M}inimum
  {M}ean-{S}quare {E}rror in {G}aussian {C}hannels,'' \emph{IEEE Transactions
  on Information Theory}, vol.~51, no.~4, pp. 1261--1282, 2005.

\bibitem{wu2011derivative}
Y.~Wu, D.~Guo, and S.~Verd{\'u}, ``Derivative of {M}utual {I}nformation at
  {Z}ero {SNR}: The {G}aussian-{N}oise {C}ase,'' \emph{IEEE Transactions on
  Information Theory}, vol.~57, no.~11, pp. 7307--7312, 2011.

\bibitem{verdu2015alpha}
S.~Verd{\'u}, ``$\alpha$-{M}utual {I}nformation,'' in \emph{2015 Information
  Theory and Applications Workshop (ITA)}, 2015, pp. 1--6.

\bibitem{shannon2003zero}
C.~Shannon, ``The {Z}ero {E}rror {C}apacity of a {N}oisy {C}hannel,'' \emph{IRE
  Transactions on Information Theory}, vol.~2, no.~3, pp. 8--19, 1956.

\bibitem{gallager2003simple}
R.~Gallager, ``A {S}imple {D}erivation of the {C}oding {T}heorem and {S}ome
  {A}pplications,'' \emph{IEEE Transactions on Information Theory}, vol.~11,
  no.~1, pp. 3--18, 1965.

\bibitem{courtade2014variable}
T.~A. Courtade and S.~Verdú, ``Variable-{L}ength {L}ossy {C}ompression and
  {C}hannel {C}oding: {N}on-{A}symptotic {C}onverses via {C}umulant
  {G}enerating {F}unctions,'' in \emph{2014 IEEE International Symposium on
  Information Theory}, 2014, pp. 2499--2503.

\bibitem{arimoto1973converse}
S.~Arimoto, ``On the {C}onverse to the {C}oding {T}heorem for {D}iscrete
  {M}emoryless {C}hannels (corresp.),'' \emph{IEEE Transactions on Information
  Theory}, vol.~19, no.~3, pp. 357--359, 1973.

\bibitem{verdu2021error}
S.~Verd{\'u}, ``Error {E}xponents and $\alpha$-{M}utual {I}nformation,''
  \emph{Entropy}, vol.~23, no.~2, p. 199, 2021.

\bibitem{yagli2019exact}
S.~Yagli and P.~Cuff, ``Exact {E}xponent for {S}oft {C}overing,'' \emph{IEEE
  Transactions on Information Theory}, vol.~65, no.~10, pp. 6234--6262, 2019.

\bibitem{issa2019operational}
I.~Issa, A.~B. Wagner, and S.~Kamath, ``An {O}perational {A}pproach to
  {I}nformation {L}eakage,'' \emph{IEEE Transactions on Information Theory},
  vol.~66, no.~3, pp. 1625--1657, 2019.

\bibitem{liao2019tunable}
J.~Liao, O.~Kosut, L.~Sankar, and F.~du~Pin~Calmon, ``Tunable {M}easures for
  {I}nformation {L}eakage and {A}pplications to {P}rivacy-{U}tility
  {T}radeoffs,'' \emph{IEEE Transactions on Information Theory}, vol.~65,
  no.~12, pp. 8043--8066, 2019.

\bibitem{xu2016information}
A.~Xu and M.~Raginsky, ``Information-{T}heoretic {L}ower {B}ounds on {B}ayes
  {R}isk in {D}ecentralized {E}stimation,'' \emph{IEEE Transactions on
  Information Theory}, vol.~63, no.~3, pp. 1580--1600, 2016.

\bibitem{esposito2021lower}
A.~R. Esposito and M.~Gastpar, ``Lower-{B}ounds on the {B}ayesian {R}isk in
  {E}stimation {P}rocedures via {S}ibson's $\alpha$-{M}utual {I}nformation,''
  in \emph{2021 IEEE International Symposium on Information Theory}, 2021, pp.
  748--753.

\bibitem{esposito2024lower}
A.~R. Esposito, A.~Vandenbroucque, and M.~Gastpar, ``Lower {B}ounds on the
  {B}ayesian {R}isk via {I}nformation {M}easures,'' \emph{Journal of Machine
  Learning Research}, vol.~25, no. 340, pp. 1--45, 2024.

\bibitem{omanwar2025bounds}
A.~Omanwar, F.~Alajaji, and T.~Linder, ``Bounds on the {E}xcess {M}inimum
  {R}isk via {G}eneralized {I}nformation {D}ivergence {M}easures,''
  \emph{Entropy}, vol.~27, no.~7, p. 727, 2025.

\bibitem{universalPredictionGaspar}
M.~Bondaschi and M.~Gastpar, ``Alpha-{NML} {U}niversal {P}redictors,''
  \emph{IEEE Transactions on Information Theory}, vol.~71, no.~2, pp.
  1171--1183, 2025.

\bibitem{HoISIt2015}
S.-W. Ho and S.~Verd\'u', ``Convexity/{C}oncavity of {R}{\'e}nyi {E}ntropy and
  $\alpha$-{M}utual {I}nformation,'' in \emph{2015 IEEE International Symposium
  on Information Theory}, 2015, pp. 745--749.

\bibitem{kamatsuka2025alternating}
A.~Kamatsuka, K.~Kazama, and T.~Yoshida, ``Alternating {O}ptimization
  {A}pproach for {C}omputing $\alpha$-{M}utual {I}nformation and
  $\alpha$-{C}apacity,'' in \emph{2025 IEEE International Symposium on
  Information Theory (ISIT)}, 2025, pp. 1--6.

\bibitem{cai2019conditional}
C.~Cai and S.~Verd{\'u}, ``Conditional {R}{\'e}nyi {D}ivergence {S}addlepoint
  and the {M}aximization of $\alpha$-{M}utual {I}nformation,'' \emph{Entropy},
  vol.~21, no.~10, p. 969, 2019.

\bibitem{guo2009relative}
D.~Guo, ``Relative {E}ntropy and {S}core {F}unction: {N}ew
  {I}nformation-{E}stimation {R}elationships through {A}rbitrary {A}dditive
  {P}erturbation,'' in \emph{2009 IEEE International Symposium on Information
  Theory}, 2009, pp. 814--818.

\bibitem{valero2017generalization}
I.~Valero-Toranzo, S.~Zozor, and J.-M. Brossier, ``Generalization of the de
  {B}ruijn {I}dentity to {G}eneral $\phi$-{E}ntropies and $\phi$-{F}isher
  {I}nformations,'' \emph{IEEE Transactions on Information Theory}, vol.~64,
  no.~10, pp. 6743--6758, 2017.

\bibitem{wu2025entropic}
H.~Wu and L.~Yu, ``Entropic {I}soperimetric and {C}ram{\'e}r--{R}ao
  {I}nequalities for {R}{\'e}nyi--{F}isher {I}nformation,'' \emph{IEEE
  Transactions on Information Theory}, 2025.

\bibitem{jizba2021renyi}
P.~Jizba, J.~Dunningham, and M.~Prok{\v{s}}, ``From {R}{\'e}nyi {E}ntropy
  {P}ower to {I}nformation {S}can of {Q}uantum {S}tates,'' \emph{Entropy},
  vol.~23, no.~3, p. 334, 2021.

\bibitem{FunctionalMMSE}
Y.~Wu and S.~Verd\'u, ``Functional {P}roperties of {M}inimum {M}ean-{S}quare
  {E}rror and {M}utual {I}nformation,'' \emph{IEEE Transactions on Information
  Theory}, vol.~58, no.~3, pp. 1289--1301, 2012.

\bibitem{GuoIT2011}
D.~Guo, Y.~Wu, S.~Shamai~(Shitz), and S.~Verd\'u, ``Estimation in {G}aussian
  {N}oise: {P}roperties of the {M}inimum {M}ean-{S}quare {E}rror,'' \emph{IEEE
  Transactions on Information Theory}, vol.~57, no.~4, pp. 2371--2385, 2011.

\bibitem{GuoNow2012}
D.~Guo, S.~Shamai~(Shitz), and S.~Verd{\'u}, ``\BIBforeignlanguage{English
  (US)}{The {I}nterplay between {I}nformation and {E}stimation {M}easures},''
  \emph{\BIBforeignlanguage{English (US)}{Foundations and Trends in Signal
  Processing}}, vol.~6, no.~4, pp. 243--429, 2012.

\bibitem{dytso2017view}
A.~Dytso, R.~Bustin, H.~V. Poor, and S.~Shamai~(Shitz), ``A {V}iew of
  {I}nformation-{E}stimation {R}elations in {G}aussian {N}etworks,''
  \emph{Entropy}, vol.~19, no.~8, p. 409, 2017.

\bibitem{guionnet2007classical}
A.~Guionnet and D.~Shlyakhtenko, ``On {C}lassical {A}nalogues of {F}ree
  {E}ntropy {D}imension,'' \emph{Journal of Functional Analysis}, vol. 251,
  no.~2, pp. 738--771, 2007.

\bibitem{wu2011mmse}
Y.~Wu and S.~Verd{\'u}, ``{MMSE} {D}imension,'' \emph{IEEE Transactions on
  Information Theory}, vol.~57, no.~8, pp. 4857--4879, 2011.

\bibitem{jalali2017universal}
S.~Jalali and H.~V. Poor, ``Universal {C}ompressed {S}ensing for {A}lmost
  {L}ossless {R}ecovery,'' \emph{IEEE Transactions on Information Theory},
  vol.~63, no.~5, pp. 2933--2953, 2017.

\bibitem{ram2016renyi}
E.~Ram and I.~Sason, ``On {R}{\'e}nyi {E}ntropy {P}ower {I}nequalities,''
  \emph{IEEE Transactions on Information Theory}, vol.~62, no.~12, pp.
  6800--6815, 2016.

\bibitem{van2014renyi}
T.~Van~Erven and P.~Harremos, ``R{\'e}nyi {D}ivergence and {K}ullback-{L}eibler
  {D}ivergence,'' \emph{IEEE Transactions on Information Theory}, vol.~60,
  no.~7, pp. 3797--3820, 2014.

\bibitem{sason2016f}
I.~Sason and S.~Verd{\'u}, ``$f$-{D}ivergence {I}nequalities,'' \emph{IEEE
  Transactions on Information Theory}, vol.~62, no.~11, pp. 5973--6006, 2016.

\bibitem{renyi1970probability}
A.~R{\'e}nyi, \emph{Probability Theory}.\hskip 1em plus 0.5em minus 0.4em\relax
  Mineola, N.Y.: Dover Publications, 2007, unabridged republication of the work
  published by North-Holland Publishing Company, Amsterdam, 1970.

\bibitem{csiszar1962dimension}
I.~Csisz{\'a}r, ``On the {D}imension and {E}ntropy of {O}rder $\alpha$ of the
  {M}ixture of {P}robability {D}istributions,'' \emph{Acta Mathematica
  Hungarica}, vol.~13, no. 3-4, pp. 245--255, 1962.

\bibitem{WuAnlogCompr}
Y.~Wu and S.~Verd{\'u}, ``R{\'e}nyi {I}nformation {D}imension: {F}undamental
  {L}imits of {A}lmost {L}ossless {A}nalog {C}ompression,'' \emph{IEEE
  Transactions on Information Theory}, vol.~56, no.~8, pp. 3721--3748, 2010.

\bibitem{wu2014information}
Y.~Wu, S.~Shamai~(Shitz), and S.~Verdú, ``{I}nformation {D}imension and the
  {D}egrees of {F}reedom of the {I}nterference {C}hannel,'' \emph{IEEE
  Transactions on Information Theory}, vol.~61, no.~1, pp. 256--279, 2015.

\bibitem{smith1971information}
J.~G. Smith, ``The {I}nformation {C}apacity of {A}mplitude-and
  {V}ariance-{C}onstrained {S}calar {G}aussian {C}hannels,'' \emph{Information
  and Control}, vol.~18, no.~3, pp. 203--219, 1971.

\bibitem{dytso2018discrete}
A.~Dytso, M.~Goldenbaum, H.~V. Poor, and S.~Shamai~(Shitz), ``When {A}re
  {D}iscrete {C}hannel {I}nputs {O}ptimal?—{O}ptimization {T}echniques and
  {S}ome {N}ew {R}esults,'' in \emph{2018 52nd Annual Conference on Information
  Sciences and Systems}, 2018, pp. 1--6.

\bibitem{eisen2023capacity}
J.~Eisen, R.~R. Mazumdar, and P.~Mitran, ``Capacity-{A}chieving {I}nput
  {D}istributions of {A}dditive {V}ector {G}aussian {N}oise {C}hannels:
  {E}ven-{M}oment {C}onstraints and {U}nbounded or {C}ompact {S}upport,''
  \emph{Entropy}, vol.~25, no.~8, p. 1180, 2023.

\bibitem{Pinsker1964}
M.~S. Pinsker, \emph{Information and {I}nformation {S}tability of {R}andom
  {V}ariables and {P}rocesses}, ser. Holden-Day series in time series
  analysis.\hskip 1em plus 0.5em minus 0.4em\relax San Francisco: Holden-Day,
  Inc., 1964.

\bibitem{AshInformationTheory}
R.~B. Ash, \emph{Information Theory}.\hskip 1em plus 0.5em minus 0.4em\relax
  Dover Publications, 1990, originally published in 1965.

\bibitem{robbins1956empirical}
H.~Robbins, ``An {E}mpirical {B}ayes {A}pproach to {S}tatistics,'' in
  \emph{Proceedings Third Berkeley Symposium on Mathematical Statistics and
  Probabily}.\hskip 1em plus 0.5em minus 0.4em\relax Citeseer, 1956.

\bibitem{hatsell1971some}
C.~Hatsell and L.~Nolte, ``Some {G}eometric {P}roperties of the {L}ikelihood
  {R}atio (corresp.),'' \emph{IEEE Transactions on Information Theory},
  vol.~17, no.~5, pp. 616--618, 1971.

\bibitem{dytso2022conditional}
A.~Dytso, H.~V. Poor, and S.~Shamai~(Shitz), ``Conditional {M}ean {E}stimation
  in {G}aussian {N}oise: {A} {M}eta {D}erivative {I}dentity with
  {A}pplications,'' \emph{IEEE Transactions on Information Theory}, vol.~69,
  no.~3, pp. 1883--1898, 2022.

\bibitem{folland1999real}
G.~B. Folland, \emph{Real Analysis: Modern Techniques and Their
  Applications}.\hskip 1em plus 0.5em minus 0.4em\relax John Wiley \& Sons,
  1999.

\bibitem{brown1971admissible}
L.~D. Brown, ``Admissible {E}stimators, {R}ecurrent {D}iffusions, and
  {I}nsoluble {B}oundary {V}alue {P}roblems,'' \emph{The Annals of Mathematical
  Statistics}, vol.~42, no.~3, pp. 855--903, 1971.

\bibitem{brown1986fundamentals}
------, \emph{Fundamentals of Statistical Exponential Families: With
  Applications in Statistical Decision Theory}, ser. Institute of Mathematical
  Statistics Lecture Notes---Monograph Series.\hskip 1em plus 0.5em minus
  0.4em\relax Hayward, CA: Institute of Mathematical Statistics, 1986, vol.~9.

\bibitem{stam1959some}
A.~J. Stam, ``Some {I}nequalities {S}atisfied by the {Q}uantities of
  {I}nformation of {F}isher and {S}hannon,'' \emph{Information and Control},
  vol.~2, no.~2, pp. 101--112, 1959.

\bibitem{VerduIT1990}
S.~Verd\'u, ``On {C}hannel {C}apacity per {U}nit {C}ost,'' \emph{IEEE
  Transactions on Information Theory}, vol.~36, no.~5, pp. 1019--1030, 1990.

\bibitem{LapidothIT2002}
A.~Lapidoth and S.~Shamai~(Shitz), ``Fading {C}hannels: {H}ow {P}erfect {N}eed
  ``{P}erfect {S}ide {I}nformation" {B}e?'' \emph{IEEE Transactions on
  Information Theory}, vol.~48, no.~5, pp. 1118--1134, 2002.

\bibitem{VerduIT2002}
S.~Verd\'u, ``Spectral {E}fficiency in the {W}ideband {R}egime,'' \emph{IEEE
  Transactions on Information Theory}, vol.~48, no.~6, pp. 1319--1343, 2002.

\bibitem{verdu2006simple}
D.~Guo, ``A {S}imple {P}roof of the {E}ntropy-{P}ower {I}nequality,''
  \emph{IEEE Transactions on Information Theory}, vol.~52, no.~5, pp.
  2165--2166, 2006.

\bibitem{tulino2006monotonic}
A.~M. Tulino and S.~Verd{\'u}, ``Monotonic {D}ecrease of the non-{G}aussianness
  of the {S}um of {I}ndependent {R}andom {V}ariables: {A} {S}imple {P}roof,''
  \emph{IEEE Transactions on Information Theory}, vol.~52, no.~9, pp.
  4295--4297, 2006.

\bibitem{guo2006proof}
D.~Guo, S.~Shamai~(Shitz), and S.~Verd{\'u}, ``Proof of {E}ntropy {P}ower
  {I}nequalities via {MMSE},'' in \emph{2006 IEEE International Symposium on
  Information Theory}, 2006, pp. 1011--1015.

\bibitem{smieja2014renyi}
M.~{\'S}mieja and J.~Tabor, ``R{\'e}nyi {E}ntropy {D}imension of the {M}ixture
  of {M}easures,'' in \emph{2014 Science and Information Conference}, 2014, pp.
  685--689.

\bibitem{raginsky2013concentration}
M.~Raginsky and I.~Sason, ``Concentration of {M}easure {I}nequalities in
  {I}nformation {T}heory, {C}ommunications, and {C}oding,'' \emph{Foundations
  and Trends in Communications and Information Theory}, vol.~10, no. 1-2, pp.
  1--247, 2013.

\bibitem{matrixSNR}
G.~Reeves, H.~D. Pfister, and A.~Dytso, ``Mutual {I}nformation as a {F}unction
  of {M}atrix {SNR} for {L}inear {G}aussian {C}hannels,'' in \emph{2018 IEEE
  International Symposium on Information Theory (ISIT)}, 2018, pp. 1754--1758.

\bibitem{bustin2009mmse}
R.~Bustin, R.~Liu, H.~V. Poor, and S.~Shamai~(Shitz), ``An {MMSE} {A}pproach to
  the {S}ecrecy {C}apacity of the {MIMO} {G}aussian {W}iretap {C}hannel,''
  \emph{EURASIP Journal on Wireless Communications and Networking}, vol. 2009,
  no.~1, p. 370970, 2009.

\bibitem{lozano2006optimum}
A.~Lozano, A.~M. Tulino, and S.~Verd{\'u}, ``Optimum {P}ower {A}llocation for
  {P}arallel {G}aussian {C}hannels with {A}rbitrary {I}nput {D}istributions,''
  \emph{IEEE Transactions on Information Theory}, vol.~52, no.~7, pp.
  3033--3051, 2006.

\bibitem{kawabata1994rate}
T.~Kawabata and A.~Dembo, ``The {R}ate-{D}istortion {D}imension of {S}ets and
  {M}easures,'' \emph{IEEE Transactions on Information Theory}, vol.~40, no.~5,
  pp. 1564--1572, 1994.

\bibitem{dytso2019capacity}
A.~Dytso, S.~Yagli, H.~V. Poor, and S.~Shamai~(Shitz), ``The {C}apacity
  {A}chieving {D}istribution for the {A}mplitude {C}onstrained {A}dditive
  {G}aussian {C}hannel: {A}n {U}pper {B}ound on the {N}umber of {M}ass
  {P}oints,'' \emph{IEEE Transactions on Information Theory}, vol.~66, no.~4,
  pp. 2006--2022, 2019.

\bibitem{stotz2016degrees}
D.~Stotz and H.~B{\"o}lcskei, ``Degrees of freedom in vector interference
  channels,'' \emph{IEEE Transactions on Information Theory}, vol.~62, no.~7,
  pp. 4172--4197, 2016.

\bibitem{guo2008mutual}
D.~Guo, S.~Shamai~(Shitz), and S.~Verd{\'u}, ``Mutual {I}nformation and
  {C}onditional {M}ean {E}stimation in {P}oisson {C}hannels,'' \emph{IEEE
  Transactions on Information Theory}, vol.~54, no.~5, pp. 1837--1849, 2008.

\bibitem{atar2012mutual}
R.~Atar and T.~Weissman, ``Mutual {I}nformation, {R}elative {E}ntropy, and
  {E}stimation in the {P}oisson {C}hannel,'' \emph{IEEE Transactions on
  Information Theory}, vol.~58, no.~3, pp. 1302--1318, 2012.

\bibitem{jiao2017relations}
J.~Jiao, K.~Venkat, and T.~Weissman, ``Relations {B}etween {I}nformation and
  {E}stimation in {D}iscrete-{T}ime {L}{\'e}vy {C}hannels,'' \emph{IEEE
  Transactions on Information Theory}, vol.~63, no.~6, pp. 3579--3594, 2017.

\bibitem{dudley2018real}
R.~M. Dudley, \emph{Real Analysis and Probability}.\hskip 1em plus 0.5em minus
  0.4em\relax Chapman and Hall/CRC, 2018.

\bibitem{evans2010pde}
L.~C. Evans, \emph{Partial Differential Equations}, 2nd~ed., ser. Graduate
  Studies in Mathematics.\hskip 1em plus 0.5em minus 0.4em\relax Providence,
  RI: American Mathematical Society, 2010, vol.~19.

\bibitem{yuksel2012optimization}
S.~Y{\"u}ksel and T.~Linder, ``Optimization and {C}onvergence of {O}bservation
  {C}hannels in {S}tochastic {C}ontrol,'' \emph{SIAM Journal on Control and
  Optimization}, vol.~50, no.~2, pp. 864--887, 2012.

\end{thebibliography}

\end{document}